\documentclass[10pt,a4paper]{article}
	\usepackage[utf8]{inputenc}
	\usepackage{amsmath}
	\usepackage{amsfonts}
	\usepackage{amssymb}
	\usepackage{color}
	\usepackage{bussproofs}
	\usepackage{amsthm}
	\usepackage{fullpage}
	\usepackage{xargs}
	\usepackage{boxedminipage}
	\usepackage{proof}
	\usepackage{mathtools}
	\usepackage{stmaryrd}
	\usepackage{code}
	\usepackage{authblk}
	
	\newtheorem{lemma}{Lemma}
	\newtheorem{theorem}{Theorem}
	\newtheorem{definition}{Definition}
	
	\newcommand{\var}{\textit{v}}
	\newcommand{\zero}{\mathtt{0}}
	\newcommand{\suc}{\mathtt{s}}
	\newcommand{\nil}{\lbrack \rbrack}
	\newcommand{\true}{\mathtt{tt}}
	\newcommand{\false}{\mathtt{ff}}
	\newcommand{\pzero}{0}
	\newcommand{\parr}{\mid}
	\newcommand{\vect}[1]{\tilde{#1}}
	\newcommandx{\pserv}[3][1 = a, 2 = \vect{\var},3 = P, usedefault=@]{!#1(#2).#3}
	\newcommandx{\pin}[3][1 = a, 2 = \vect{\var},3 = P, usedefault=@]{#1(#2).#3}
	\newcommandx{\pout}[2][1 = a, 2 = \vect{e}, usedefault=@]{\overline{#1} \langle #2 \rangle}
	\newcommandx{\pnu}[1][1 = a]{(\nu #1)}
	\newcommandx{\pifn}[4][1 = e, 2 = P, 3 = x, 4 = Q, usedefault = @]{\mathtt{match}~#1~\mathtt{with}~\{\zero \mapsto #2 \mid \suc(#3) \mapsto #4 \}}
	\newcommandx{\pif}[3][1 = e, 2 = P, 3 = Q, usedefault = @]{\mathtt{if}~#1~\mathtt{then}~#2~\mathtt{else}~#3}
	\newcommandx{\pifl}[5][1 = e, 2 = P, 3 = x, 4 = y, 5 = Q, usedefault = @]{\mathtt{match}~#1~\mathtt{with}~\{\nil \mapsto #2 \mid #3::#4 \mapsto #5 \}}
	\newcommand{\tick}{\mathtt{tick}} 
	
	\newcommand{\congr}{\equiv}
	\newcommand{\red}{\rightarrow}
	\newcommandx{\psub}[3][1 = P, 2 = \vect{\var}, 3 = \vect{e},usedefault=@]{#1 [#2 := #3]}
	\newcommand{\tred}{\Rightarrow}
	\newcommand{\gr}[1]{\lfloor #1 \rfloor}
	\newcommand{\tickred}{\overset{tick}{\rightarrow}}
	
	\newcommand{\nat}{\mathsf{Nat}}
	\newcommand{\lis}{\mathsf{List}}
	\newcommand{\bool}{\mathsf{Bool}}
	\newcommandx{\ch}[1][1 = \vect{T}]{\mathsf{ch}(#1)}
	\newcommandx{\ich}[1][1 = \vect{T}]{\mathsf{in}(#1)}
	\newcommandx{\och}[1][1 = \vect{T}]{\mathsf{out}(#1)}
	\newcommand{\subtype}{\sqsubseteq}
	\newcommand{\IB}{\mathcal{B}}
	\newcommandx{\Isum}[3][1 = i, 2 = I, 3 = J,usedefault = @]{\sum_{#1 \le #2} #3}
	\newcommand\sem[2][]{\llbracket #2 \rrbracket_{#1}}
	\newcommand\IV{\mathcal{V}}
	\newcommandx{\Isub}[3][1 = I, 2 = i, 3 = J, usedefault = @]{#1 \lbrace #3 / #2 \rbrace }
	\newcommand{\rel}{\bowtie}
	
	\newcommandx{\Inat}[2][1 = I, 2 = J, usedefault=@]{\nat \lbrack #1 , #2 \rbrack}
	\newcommand{\Ibool}{\bool}
	\newcommandx{\Ilis}[3][1 = I, 2 = J, 3 = \IB, usedefault=@]{\lis \lbrack #1 , #2 \rbrack(#3)}
	\newcommandx{\Ich}[2][1 = I, 2 = \vect{T}, usedefault=@]{\mathsf{ch}_{#1}(#2)}
	\newcommandx{\Iich}[2][1 = I, 2 = \vect{T}, usedefault=@]{\mathsf{in}_{#1}(#2)}
	\newcommandx{\Ioch}[2][1 = I, 2 = \vect{T}, usedefault=@]{\mathsf{out}_{#1}(#2)}
	\newcommandx{\Iserv}[4][1 = I, 2 = \vect{i}, 3 = K, 4 = \vect{T}, usedefault=@]{\forall_{#1} #2. \mathsf{serv}^{#3}(#4)}
	\newcommandx{\Iiserv}[4][1 = I, 2 = \vect{i}, 3 = K, 4 = \vect{T}, usedefault=@]{\forall_{#1} #2. \mathsf{iserv}^{#3}(#4)}
	\newcommandx{\Ioserv}[4][1 = I, 2 = \vect{i}, 3 = K, 4 = \vect{T}, usedefault=@]{\forall_{#1} #2. \mathsf{oserv}^{#3}(#4)}
	\newcommandx{\Idecr}[3][1 = \Gamma,2 = I, 3 , usedefault=@]{\langle #1 \rangle_{- #2}^{#3}}
	\newcommandx{\Itype}[5][1 = \phi, 2 = \Phi, 3 = \Gamma, usedefault = @]{#1;#2;#3 \vdash #4 \lhd #5}
	\newcommandx{\Ietype}[5][1 = \phi, 2 = \Phi, 3 = \Gamma, usedefault = @]{#1;#2;#3 \vdash #4 : #5}
	\newcommandx{\Iincr}[2][1 = I]{#2_{+ #1}}
	\newcommand{\IS}{\mathcal{S}}
	\newcommandx{\Idisc}[2][1 = \IS]{\mathtt{disc}_{#1}(#2)} 
	\newcommandx{\IT}[1][1]{\mathcal{T}_{out}^{#1}}
	\newcommand{\IU}{\mathcal{U}}
	\newcommandx{\Tch}[1][1 = \vect{T}, usedefault=@]{\mathsf{ch}(#1)}
	\newcommandx{\Tich}[1][1 = \vect{T}, usedefault=@]{\mathsf{in}(#1)}
	\newcommandx{\Toch}[1][1 = \vect{T}]{\mathsf{out}(#1)}
	\newcommandx{\Tserv}[3][1 = \vect{i}, 2 = K, 3 = \vect{T}, usedefault=@]{\forall #1. \mathsf{serv}^{#2}(#3)}
	\newcommandx{\Tiserv}[3][1 = \vect{i}, 2 = K, 3 = \vect{T}, usedefault=@]{\forall #1. \mathsf{iserv}^{#2}(#3)}
	\newcommandx{\Toserv}[3][1 = \vect{i}, 2 = K, 3 = \vect{T}, usedefault=@]{\forall #1. \mathsf{oserv}^{#2}(#3)}
	\newcommand{\NN}{\mathbb{N}}
	
	\newenvironment{framed}[0]{\begin{boxedminipage}{\linewidth}}{\end{boxedminipage}}
	\newcommand{\midd}{\; \; \mbox{\Large{$\mid$}}\;\;}
	\newcommand{\vvskip}{\vspace{3mm}}
	\newcommand{\leftshift}[1]{\mathmakebox[\textwidth][c]{#1}}
	\newcommand{\leftshiftt}[1]{\makebox[\textwidth][c]{#1}}

	\title{Types for Parallel Complexity in the Pi-calculus}
	\date{} 
	\author[1]{Patrick Baillot}
	\author[1]{Alexis Ghyselen}
	
	\affil[1]{Univ Lyon, CNRS, ENS de Lyon, Universite Claude-Bernard Lyon 1, LIP \\ F-69342, Lyon Cedex 07, France}
	
\begin{document}
	\maketitle
	
	\section{Introduction} 
	
\paragraph{Context} Certifying time complexity bounds for a program
is a challenging question as it deals with properties which are
important for predicting quantitative behaviour of software but which
are of course undecidable. In the setting of sequential functional 
programs this problem, as well as the related one of time complexity 
inference, have been addressed using type systems (see e.g. 
\cite{DBLP:conf/cav/0002AH12,DalLagoGaboardiLinearDependentTypes,AvanziniDalLagoAutomatingSizedTypeInference}).
These settings provide rich type systems such that if a program can
be assigned a type, then one can extract from the type derivation a
complexity bound for its execution on any input. The type system
itself thus provides a complexity certification procedure, and if a type inference
algorithm is also provided one obtains a complexity inference
procedure. It is then quite natural to wonder whether similar kinds of analysis could be
carried out for languages that can express parallel computation and
concurrent behaviours, such as process calculi and in particular the
$\pi$-calculus. In such a setting however sequential time complexity
is not sufficient, and one would be more naturally interested in
handling notions of parallel complexity, such as the \textit{span}
and the \emph{work} of the system. This is the problem we wish to
tackle in the present work.
	
\paragraph{Approach} 	
	We want to be able to choose for different examples of
        systems the cost model we are interested in, e.g. should we
        count the number of emissions of messages, receptions,
        comparisons etc.? For this reason it will be convenient to
        consider an instrumented language, with a $tick$ construction
        that we will use to mark the operations we want to count.

        A second requirement that we have is that we wish to derive
        complexity bounds which are parametric with respect to the
        size of inputs, for instance which depend on the length of a
        list. For that it will be useful to have a language of types
        that can carry information about sizes, and this is why we
        take inspiration from size types. So data-types will be
        annotated with an index (or parameter) which will provide
        some information on the size of values. Moreover, as we want
        to bound the execution time and as we are in a setting of
        communication through channels, a second ingredient that we
        will use is that the typing of a channel will carry
        information about \textit{when} communication will be
        performed on this channel. In order to be able to
        reason differently on bounds for emission and reception it
        will be convenient for us to use the approach of input/output
        types for $\pi$-calculus.
		
			\paragraph{Contributions}  In this paper we
                        define two type systems for the $\pi$-calculus
                        which  provide upper bounds respectively on
                        the span and on the work complexity of a
                        term. For that we first define a small-step operational
                        semantics on the  $\pi$-calculus with $tick$,
                        which allows to characterize the
                        span. Intuitively it performs reduction with
                        maximal parallelism. We then introduce a type
                        system of size types with temporal
                        information. Typing judgements assign a
                        complexity $K$ to the typed process. We
                        prove a soundness result, stating that if
                        a  process $P$ can be typed and assigned a
                        complexity $K$, then $K$ bounds its reduction
                        time in the operational semantics, hence its
                        span complexity. We also describe a second
                        small-step operational semantics corresponding
                        to the work, and a variant of the first
                        type system which provides a bound on the work
                        complexity. 
		
		\paragraph{Related Work} To the author's knowledge,
                the first work to capture parallel complexity by means
                of type was given by Kobayashi
                \cite{KobayashiTypeSystemLockFree}. In this work, only
                the parallel communication complexity is considered
                and the notion of time appears both in syntax and
                types. A type contains \emph{usages}, that is
                intuitively a detailed description of its
                behaviour. With this, there is no need for time
                linearity as in our work. Moreover, the use of
                dependent types to have an extension of the type
                system was also proposed but not detailed. Then, Das,
                Hoffmann and Pfenning proposed a type system with
                temporal session types
                \cite{DasHoffmannPfenningTemporalSessionTypes,DasHoffmannPfenningLICS2018}
                to capture several notions of complexity. In this
                work, time and complexity are captured in the type
                system by the use of temporal logic  time
                modalities.
However, the use of session-types imposes a strict linearity that we
believe restricts the expressiveness of their programs. 

		The methodology of our work is inspired by implicit
                computational complexity, which aims at characterizing
                complexity classes by means of dedicated programming
                languages or logics, for instance by providing
                sequential languages whose programs characterize exactly the
                class of FPTIME functions.  Some results have already been
                adapted to the concurrent case, but mainly for the
                work complexity and not for the
                span, e.g. \cite{MadetAmadioELL} for a lambda-calculus with
                multithreading,
                \cite{DiGiamberardinoDalLagoSessionTypePolynomial} for
                a language of session types,
                \cite{DemangeonYoshidaSafePiCalculus} for
                $\pi$-calculus and
                \cite{DalLagoMartiniSangiorgi16} for  a higher-order
                  $\pi$-calculus.  Contrarily to those works we do not
                  restrict to a particular complexity class (like FPTIME)  and we
                  handle the case of the span. 

Technically, the types we use are inspired from linear dependent
types, introduced by Dal Lago and Gaboardi
\cite{DalLagoGaboardiLinearDependentTypes}. Those are one of the many
variants and descendants of size types, which were first introduced by Hughes, Pareto and Sabry 
		\cite{HughesParetoSabrySizedTypes}. 

	\section{Pi-calculus with input/output types}
	\label{s:Simpletypes}
	
	We present here a classical $\pi$-calculus with input/output types. More detail about those types or $\pi$-calculus can be found in \cite{SangiorgiWalkerPi}. The set of \emph{variables}, \emph{expressions} 
	and \emph{pre-processes} are defined by the following grammar.  
	\begin{align*}
		\var & := x,y,z \midd a,b,c \\
		e &:= \var \midd \zero \midd \suc(e) \midd \nil \midd e::e' \midd \true \midd \false \\
		P,Q &:= \pzero \midd P \parr Q \midd \pserv \midd \pin \midd \pout \midd \pnu P \midd \pifn \\
		&\midd \pifl \midd \pif 
	\end{align*}
	
	Variables $x,y,z$ denote \emph{base type variables}, so they represent integers, lists or booleans. Variables $a,b,c$ denote \emph{channel variables}. The notation $\vect{\var}$ stands for a sequence of variables $\var_1,\var_2,\dots,\var_k$. In the same way, $\vect{e}$ is a sequence of expressions. We work up to $\alpha$-renaming, and we use $\psub$ to denote the substitution of the free variables $\vect{\var}$ in $P$ by $\vect{e}$. 
	
	We define on those pre-processes a congruence relation $\congr$, such that this relation is the least congruence relation closed under:
	\[ P \parr \pzero \congr P \qquad P \parr Q \congr Q \parr P \qquad P \parr (Q \parr R) \congr (P \parr Q) \parr R \qquad \pnu[a] \pnu[b] P \congr \pnu[b] \pnu[a]P \]
	\[ \pnu[a](P \parr Q) \congr \pnu[a] P \parr Q~(\text{when } a \text{ is not free in } Q) \] 
	
	Note that the last rule can always be made by renaming. Also, note that contrary to other congruence relation for the $\pi$-calculus,
	we do not consider the rule for replicated input ($!P \congr !P | P$), and $\alpha$-conversion is not an explicit rule in the congruence. 
	With this definition, we can give a canonical form for pre-processes, as in \cite{KobayashiLinearityPiCalculus}.
	
	\begin{definition}[Guarded Pre-processes and Canonical Form]
		A pre-process is \emph{guarded} if it is an input, a replicated input, an output or a conditional. We say that a pre-process is in \emph{canonical form} if it has the form 
		\[ \pnu[\vect{a}] (P_1 \parr \cdots \parr P_n) \]
		with $P_1,\dots,P_n$ that are guarded pre-processes. 
	\label{d:guardedcanonical}
	\end{definition}

	And we now show that all processes have a somewhat unique canonical form. 
	
	\begin{lemma}[Existence of Canonical Form]
		For any pre-process $P$, there is a $Q$ in canonical form such that $P \congr Q$.
	\label{l:existcanonical}
	\end{lemma}
	
	\begin{proof}
		Le us suppose that, by renaming, all the introduction of new variables have different names and that they also differ from the free variables already in $P$. We can then proceed by induction on the structure of $P$. The only interesting case is for parallel composition. Suppose that \[ P \congr \pnu[\vect{a}] (P_1 \parr \cdots \parr P_n) \qquad Q \congr \pnu[\vect{b}] (Q_1 \parr \cdots \parr Q_m) \]
		With $P_1,\dots,P_n,Q_1,\dots,Q_m$ guarded pre-processes. Then, by hypothesis on the name of variables, we have $\vect{a}$ and $\vect{b}$ disjoint and $\vect{a}$ is not free in $Q$, as well as $\vect{b}$ is not free in $P$. So, we obtain 
		\[P \parr Q \congr \pnu[\vect{a}] \pnu[\vect{b}] (P_1 \parr \cdots \parr P_n \parr Q_1 \parr \cdots \parr Q_m) \]
	\end{proof}
	
	\begin{lemma}[Uniqueness of Canonical Form]
		If \[ \pnu[\vect{a}] (P_1 \parr \cdots \parr P_n) \congr \pnu[\vect{b}] (Q_1 \parr \cdots \parr Q_m)  \]
		with $P_1,\dots,P_n,Q_1,\dots,Q_m$ guarded pre-processes, then $m = n$ and $\vect{a}$ is a permutation of $\vect{b}$. Moreover, for some permutation 
		$Q_1',\dots,Q_n'$ of $Q_1,\dots,Q_n$, we have $P_i \congr Q_i'$ for all $i$. 
	\label{l:uniquecanonical}
	\end{lemma} 
	
	\begin{proof}
		Recall that $\alpha$-renaming is not a rule of $\congr$, otherwise this propriety would be false. As before, we suppose that all names are already  well-chosen. Then, let us define a set $\mathtt{name}$ of channel variable and a multiset $\mathtt{gp}$ of guarded pre-processes. 
		\begin{itemize}
			\item $\mathtt{name}(\pzero) = \emptyset$ and $\mathtt{gp}(\pzero) = \emptyset$.
			\item $\mathtt{name}(P \parr Q) = \mathtt{name}(P) \coprod \mathtt{name}(Q)$ and $\mathtt{gp}(P \parr Q) = \mathtt{gp}(P) + \mathtt{gp}(Q)$.
			\item $\mathtt{name}(P) = \emptyset$ and $\mathtt{gp}(P) = \lbrack P \rbrack $, when $P$ is guarded.
			\item $\mathtt{name}(\pnu P) = \mathtt{name(P)} \coprod \{ a \}$ and $\mathtt{gp}(\pnu P) = \mathtt{gp}(P)$.
		\end{itemize}
		Then, we can easily show the following lemma by definition of the congruence relation.
		\begin{lemma}
			If $P \congr Q$ then $\mathtt{name}(P) = \mathtt{name}(Q)$ and if $\mathtt{gp}(P) = \lbrack P_1, \dots, P_n \rbrack$ and $\mathtt{gp}(Q) = \lbrack Q_1, \dots, Q_m \rbrack$, then $m=n$ and for some permutation $Q_1',\dots,Q_n'$ of $Q_1,\dots,Q_n$, we have $P_i \congr Q_i'$ for all $i$. 
		\label{l:namegpcongr}
		\end{lemma}
		Finally, Lemma~\ref{l:uniquecanonical} is a direct consequence of Lemma~\ref{l:namegpcongr}. 
	\end{proof}
	
	We can now define the reduction relation $P \red Q$ for pre-processes. It is defined by the rules given in Figure~\ref{f:reduction}. Remark that substitution should be well-defined in order to do the reduction. For example, \hbox{$\psub[(\pin)][a][\true]$} is not a valid substitution, as a channel variable is replaced by a boolean. More formally, channel variables must be substituted by other channel variables and base type variables can be substituted by any expression except channel variables. However, when we have typed pre-processes, this yields well-typed substitutions.  
	
	\begin{figure}
		\centering
		\begin{framed}
			\small 
			\begin{center}
			 \AxiomC{}
			 \UnaryInfC{$\pserv \parr \pout \red \pserv \parr \psub$} 
			 \DisplayProof
			 \qquad 
			 \AxiomC{}
			 \UnaryInfC{$\pin \parr \pout \red \psub$}
			 \DisplayProof
			 \\
			 \vvskip
			 \AxiomC{}
			 \UnaryInfC{$\pifn[\zero] \red P$}
			 \DisplayProof 
			 \qquad 
			 \AxiomC{}
			 \UnaryInfC{$\pifn[\suc(e)] \red \psub[Q][x][e]$}
			 \DisplayProof
			 \\ 
			 \vvskip 
			 \AxiomC{}
			 \UnaryInfC{$\pifl[\nil] \red P$}
			 \DisplayProof 
			 \qquad 
			 \AxiomC{}
			 \UnaryInfC{$\pifl[e::e'] \red \psub[Q][x,y][e,e']$}
			 \DisplayProof
			 \\ 
			 \vvskip
			 \AxiomC{}
			 \UnaryInfC{$\pif[\true] \red P$}
			 \DisplayProof 
			 \qquad 
			 \AxiomC{}
			 \UnaryInfC{$\pif[\false] \red Q$}
			 \DisplayProof 
			 \qquad 
			 \AxiomC{$P \red Q$}
			 \UnaryInfC{$ P \parr R \red Q \parr R$}
			 \DisplayProof 
			 \\ 
			 \vvskip 
			 \AxiomC{$P \red Q$}
			 \UnaryInfC{$\pnu P \red \pnu Q$}
			 \DisplayProof 
			 \qquad 
			 \AxiomC{$P \congr P'$}
			 \AxiomC{$P' \red Q'$}
			 \AxiomC{$Q' \congr Q$}
			 \TrinaryInfC{$P \red Q$}
			 \DisplayProof  
		   \end{center}   
		\end{framed}
		\caption{Reduction Rules}
		\label{f:reduction}
	\end{figure}
	
	The set of \emph{base types} and types are given by the following grammar. 
	
	\[ \IB := \nat \midd \lis(\IB) \midd \bool \qquad T := \IB \midd \ch \midd \ich \midd \och \]
	
	When a type $T$ is not a base type, we call it a \emph{channel type}. Then, we define a subtyping relation on those types, expressed by the rules of Figure~\ref{f:subtyping}
	
	\begin{figure}
		\centering
		\begin{framed}
			\small 
			\begin{center}
				\AxiomC{}
				\UnaryInfC{$\IB \subtype \IB$}
				\DisplayProof
				\qquad 
				\AxiomC{$\vect{T} \subtype \vect{U}$}
				\AxiomC{$\vect{U} \subtype \vect{T}$}
				\BinaryInfC{$\ch \subtype \ch[\vect{U}]$}
				\DisplayProof
				\qquad 
				\AxiomC{}
				\UnaryInfC{$\ch \subtype \ich$}
				\DisplayProof 
				\qquad 
				\AxiomC{}
				\UnaryInfC{$\ch \subtype \och$}
				\DisplayProof 
				\\
				\vvskip 
				\AxiomC{$\vect{T} \subtype \vect{U}$}
				\UnaryInfC{$\ich \subtype \ich[\vect{U}]$}
				\DisplayProof 
				\qquad 
				\AxiomC{$\vect{U} \subtype \vect{T}$}
				\UnaryInfC{$\och \subtype \och[\vect{U}]$}
				\DisplayProof 
				\qquad 
				\AxiomC{$ T \subtype T'$}
				\AxiomC{$T' \subtype T''$}
				\BinaryInfC{$T \subtype T''$}
				\DisplayProof 
				
			\end{center}   
		\end{framed}
		\caption{Subtyping Rules}
		\label{f:subtyping}
	\end{figure}
	
	\begin{definition}[Typing Contexts]
		A \emph{typing context} $\Gamma$ is a sequence of hypotheses of the form $ x : \IB$ or $ a : T$ where $T$ is a channel type.  
	\label{d:typingcontexts}
	\end{definition}
	
	We can now define typing for expressions and pre-processes. This is expressed by the rules of Figure~\ref{f:typeexpression} and Figure~\ref{f:typeprocess}.
	
	\begin{figure}
		\centering
		\begin{framed}
			\small 
			\begin{center}
				\AxiomC{$\var : T \in \Gamma$}
				\UnaryInfC{$\Gamma \vdash \var : T$}
				\DisplayProof 
				\qquad
				\AxiomC{}
				\UnaryInfC{$\Gamma \vdash \zero : \nat$}
				\DisplayProof 
				\qquad
				\AxiomC{$\Gamma \vdash e : \nat$}
				\UnaryInfC{$\Gamma \vdash \suc(e) : \nat$}
				\DisplayProof 
				\qquad
				\AxiomC{}
				\UnaryInfC{$\Gamma \vdash \nil : \lis(\IB)$}
				\DisplayProof 
				\\
				\vvskip 
				\AxiomC{$\Gamma \vdash e : \IB$}
				\AxiomC{$\Gamma \vdash e' : \lis(\IB) $}
				\BinaryInfC{$\Gamma \vdash e::e' : \lis(\IB)$}
				\DisplayProof 
				\qquad
				\AxiomC{}
				\UnaryInfC{$\Gamma \vdash \true : \bool$}
				\DisplayProof 
				\qquad
				\AxiomC{}
				\UnaryInfC{$\Gamma \vdash \false : \bool$}
				\DisplayProof
				\\
				\vvskip   
				\AxiomC{$\Delta \vdash e : U$}
				\AxiomC{$\Gamma \subtype \Delta$}
				\AxiomC{$U \subtype T$}
				\TrinaryInfC{$\Gamma \vdash e : T$}
				\DisplayProof 
			\end{center}   
		\end{framed}
		\caption{Typing Rules for Expressions}
		\label{f:typeexpression}
	\end{figure} 
	
	\begin{figure}
		\centering
		\begin{framed}
			\small 
			\begin{center}
				\AxiomC{}
				\UnaryInfC{$\Gamma \vdash \pzero$}
				\DisplayProof 
				\qquad
				\AxiomC{$\Gamma \vdash P$}
				\AxiomC{$\Gamma \vdash Q$}
				\BinaryInfC{$\Gamma \vdash P \parr Q$}
				\DisplayProof 
				\qquad
				\AxiomC{$\Gamma \vdash a : \ich$}
				\AxiomC{$\Gamma, \vect{\var} : \vect{T} \vdash P$}
				\BinaryInfC{$\Gamma \vdash \pserv$}
				\DisplayProof 
				\\
				\vvskip 			
				\AxiomC{$\Gamma \vdash a : \ich$}
				\AxiomC{$\Gamma, \vect{\var} : \vect{T} \vdash P$}
				\BinaryInfC{$\Gamma \vdash \pin$}
				\DisplayProof 
				\qquad
				\AxiomC{$\Gamma \vdash a : \och$}
				\AxiomC{$\Gamma \vdash \vect{e} : \vect{T}$}
				\BinaryInfC{$\Gamma \vdash \pout$}
				\DisplayProof 
				\qquad
				\AxiomC{$\Gamma, a : T \vdash P$}
				\UnaryInfC{$\Gamma \vdash \pnu P$}
				\DisplayProof 
				\\
				\vvskip 
				\AxiomC{$\Gamma \vdash e : \nat$}
				\AxiomC{$\Gamma \vdash P$}
				\AxiomC{$\Gamma, x : \nat \vdash Q$}
				\TrinaryInfC{$\Gamma \vdash \pifn$}
				\DisplayProof 
				\qquad  
				\AxiomC{$\Gamma \vdash e : \lis(\IB)$}
				\AxiomC{$\Gamma \vdash P$}
				\AxiomC{$\Gamma, x : \IB, y : \lis(\IB) \vdash Q$}
				\TrinaryInfC{$\Gamma \vdash \pifl$}
				\DisplayProof 
				\\
				\vvskip 
				\AxiomC{$\Gamma \vdash e : \bool$}
				\AxiomC{$\Gamma \vdash P$}
				\AxiomC{$\Gamma \vdash Q$}
				\TrinaryInfC{$\Gamma \vdash \pif$}
				\DisplayProof 
				\qquad 
				\AxiomC{$\Delta \vdash P$}
				\AxiomC{$\Gamma \subtype \Delta$}
				\BinaryInfC{$\Gamma \vdash P$}
				\DisplayProof 
			\end{center}   
		\end{framed}
		\caption{Typing Rules for Pre-Processes}
		\label{f:typeprocess}
	\end{figure}
	
	Finally, we can now show the following lemma. 
	
	\begin{lemma}[Closed Typed Normal Forms]
		A pre-process $P$ such that $\vdash P$ is in normal form for $\red$ if and only if 
		\[P \congr \pnu[(a_1, \dots a_n)](\pserv[b_1][\vect{\var_1^0}][P_1] \parr \cdots \parr \pserv[b_k][\vect{\var_k^0}][P_k] \parr \pin[c_1][\vect{\var_1^1}][Q_1] \parr \cdots \parr \pin[c_{k'}][\vect{\var_{k'}^1}][Q_{k'}] \parr \pout[d_1][\vect{e_1}] \parr \cdots \parr \pout[d_{k''}][\vect{e_{k''}}]) \]
		with $(\{ b_i \mid 1 \le i \le k \} \cup \{ c_i \mid 1 \le i \le k' \}) \cap \{ d_i \mid 1 \le i \le k\} = \emptyset$. 
	\label{l:closedtypednormalform}
	\end{lemma}
	
	\begin{proof}
		In order to show that, we first give an exhaustive list of possibilities for a reduction, as in \cite{KobayashiLinearityPiCalculus}. 
		\begin{lemma}
			If $R \red R'$ then one of the following statements is true (where $R_1,\dots,R_n$ are guarded pre-processes)
			\begin{itemize}
				\item  \[ R \congr \pnu[\vect{b}] (R_1 \parr \cdots \parr R_n \parr \pserv \parr \pout) \qquad R' \congr \pnu[\vect{b}] (R_1 \parr \cdots \parr R_n \parr \pserv \parr \psub)  \]
				\item  \[ R \congr \pnu[\vect{b}] (R_1 \parr \cdots \parr R_n \parr \pin \parr \pout) \qquad R' \congr \pnu[\vect{b}] (R_1 \parr \cdots \parr R_n \parr \psub)  \]
				\item  \[ R \congr \pnu[\vect{b}] (R_1 \parr \cdots \parr R_n \parr \pifn[\zero]) \qquad R' \congr \pnu[\vect{b}] (R_1 \parr \cdots \parr R_n \parr P)  \]
				\item  \[ R \congr \pnu[\vect{b}] (R_1 \parr \cdots \parr R_n \parr \pifn[\suc(e)]) \qquad R' \congr \pnu[\vect{b}] (R_1 \parr \cdots \parr R_n \parr \psub[Q][x][e])  \]
				\item  \[ R \congr \pnu[\vect{b}] (R_1 \parr \cdots \parr R_n \parr \pifl[\nil]) \qquad R' \congr \pnu[\vect{b}] (R_1 \parr \cdots \parr R_n \parr P)  \]
				\item  \[ R \congr \pnu[\vect{b}] (R_1 \parr \cdots \parr R_n \parr \pifl[e::e']) \qquad R' \congr \pnu[\vect{b}] (R_1 \parr \cdots \parr R_n \parr \psub[Q][x,y][e,e'])  \]
				\item  \[ R \congr \pnu[\vect{b}] (R_1 \parr \cdots \parr R_n \parr \pif[\true]) \qquad R' \congr \pnu[\vect{b}] (R_1 \parr \cdots \parr R_n \parr P)  \]
				\item  \[ R \congr \pnu[\vect{b}] (R_1 \parr \cdots \parr R_n \parr \pif[\false]) \qquad R' \congr \pnu[\vect{b}] (R_1 \parr \cdots \parr R_n \parr Q)  \]
			\end{itemize}
		\label{l:exhaustivereductions}
		\end{lemma}
		
		\begin{proof}
			By induction on $R \red R'$. All base cases are straightforward. Then, in parallel composition, we use Lemma~\ref{l:existcanonical} to obtain the correct form. The contextual rule for $\nu$ is straightforward, and finally, the reduction up to congruence is straightforward using the transitivity of $\congr$.
		\end{proof}
	
		We can now show Lemma~\ref{l:closedtypednormalform}. Suppose that 
		\[P \congr \pnu[(a_1, \dots a_n)](\pserv[b_1][\vect{\var_1^0}][P_1] \parr \cdots \parr \pserv[b_k][\vect{\var_k^0}][P_k] \parr \pin[c_1][\vect{\var_1^1}][Q_1] \parr \cdots \parr \pin[c_{k'}][\vect{\var_{k'}^1}][Q_{k'}] \parr \pout[d_1][\vect{e_1}] \parr \cdots \parr \pout[d_{k''}][\vect{e_{k''}}]) \] 
		
		with $(\{ b_i \mid 1 \le i \le k \} \cup \{ c_i \mid 1 \le i \le k' \}) \cap \{ d_i \mid 1 \le i \le k\} = \emptyset$.
		By Lemma~\ref{l:uniquecanonical}, this canonical form for $P$ is unique. As $P$ cannot have any of the possible form described in Lemma~\ref{l:exhaustivereductions}, $P$ cannot be reduced by $\red$ thus $P$ is indeed in normal form for $\red$. 
		
		Conversely, suppose that $P$ is in normal form for $\red$, with $\vdash P$. Let us write the canonical form:
		
		\[P \congr \pnu[\vect{a}] (P_1 \parr \cdots \parr P_n) \]
		
		First, let us show that there is no conditional in $P_1, \dots, P_n$. Indeed, if $P_i$ has the form $\pif[@][R][R']$, then by typing, we have $ \vect{a} : \vect{T} \vdash \pif[@][R][R']$. Thus, we obtain $ \vect{a} : \vect{T} \vdash e : \bool$. By definition of expressions, as all type in $\vect{T}$ must be channel types, we have $e = \true$ or $e = \false$, thus $P$ would not be in normal form for $\red$. Then, the canonical form can be written: 
		
		\[ \pnu[(a_1, \dots a_n)](\pserv[b_1][\vect{\var_1^0}][P_1] \parr \cdots \parr \pserv[b_k][\vect{\var_k^0}][P_k] \parr \pin[c_1][\vect{\var_1^1}][Q_1] \parr \cdots \parr \pin[c_{k'}][\vect{\var_{k'}^1}][Q_{k'}] \parr \pout[d_1][\vect{e_1}] \parr \cdots \parr \pout[d_{k''}][\vect{e_{k''}}]) \]
		
		Now suppose, that $a \in ((\{ b_i \mid 1 \le i \le k \} \cup \{ c_i \mid 1 \le i \le k' \}) \cap \{ d_i \mid 1 \le i \le k\})$. In the type derivation $\vdash P$, $a$ is given a channel type $T$. As a consequence, in the (replicated) input and in the output, the type of $\vect{\var}$ in the input and $\vect{e}$ in the output corresponds, thus the substitution is well-defined and so $P$ is reducible. This is absurd, thus $(\{ b_i \mid 1 \le i \le k \} \cup \{ c_i \mid 1 \le i \le k' \}) \cap \{ d_i \mid 1 \le i \le k\}) = \emptyset$, and we obtain the desired form.
		
		This concludes the proof of Lemma~\ref{l:closedtypednormalform}. 
	\end{proof}
		
		In the following, we will use a generalized version of Lemma~\ref{l:closedtypednormalform}, with exactly the same proof.

	\section{Size Types and Complexity}
	\label{s:sizetypes}
	
	We enrich the previous set of pre-processes with a new constructor: $\tick.P$. This new set of terms is called the set of \emph{processes}. Intuitively, $\tick.P$ stands for ''after one unit of time, the process continues as P''. We extend the congruence relation and typing to this new constructor, thus we add the following rule for congruence and for typing.
	\begin{center}
		\AxiomC{$ P \congr Q$}
		\UnaryInfC{$\tick.P \congr \tick.Q$}
		\DisplayProof
		\qquad 
		\AxiomC{$ \Gamma \vdash P$}
		\UnaryInfC{$\Gamma \vdash \tick.P$}
		\DisplayProof 
	\end{center}
	
	A process of the form $\tick.P$ is considered a guarded process. Moreover, $\tick.P$ should be considered as a stuck process for the reduction $\red$. For example, the process $ \pin \parr \tick. \pout$ is not reducible for the relation $\red$. In order to reduce the tick, we define another reduction relation that stands for ''one unit of time'', thus, this new relation will be linked with our notion of complexity for our calculus. 
	
	\subsection{Time Reduction}
	\label{ss:timereduction}
	
	The time reduction $\tred$ is defined by the rules of Figure~\ref{f:timereduction}. Note that some processes have implicitly a ''wait'' instruction, for example a server or a process waiting for an input will always wait for its input according to this relation.   
	
	\begin{figure}
		\centering
		\begin{framed}
			\small 
			\begin{center}
				\AxiomC{}
				\UnaryInfC{$\pzero \tred \pzero$}
				\DisplayProof 
				\qquad 
				\AxiomC{}
				\UnaryInfC{$\pserv \tred \pserv$}
				\DisplayProof 
				\qquad 
				\AxiomC{}
				\UnaryInfC{$\pin \tred \pin$}
				\DisplayProof 
				\qquad 
				\AxiomC{}
				\UnaryInfC{$\pout \tred \pout$}
				\DisplayProof 
				\qquad  
				\AxiomC{$P \tred P'$}
				\UnaryInfC{$\pnu P \tred \pnu P'$}
				\DisplayProof
				\\ 
				\vvskip 
				\AxiomC{}
				\UnaryInfC{$\pifn \tred \pifn$}
				\DisplayProof 
				\\
				\vvskip 
				\AxiomC{}
				\UnaryInfC{$\pifl  \tred \pifl $}
				\DisplayProof 
				\\
				\vvskip 
				\AxiomC{}
				\UnaryInfC{$\pif \tred \pif$}
				\DisplayProof 
				\\
				\vvskip 
				\AxiomC{}
				\UnaryInfC{$\tick. P \tred P$}
				\DisplayProof 
				\qquad 
				\AxiomC{$P \tred P'$}
				\AxiomC{$Q \tred Q'$}
				\BinaryInfC{$P \parr Q \tred P' \parr Q'$}
				\DisplayProof 
			\end{center}   
		\end{framed}
		\caption{Time Reduction Rules}
		\label{f:timereduction}
	\end{figure}
	
	Note than for any process $P$, there is a unique $Q$ such that $P \tred Q$. Note also that $\tred$ allows the reduction of multiple ticks in parallel in one step. Indeed, we are here interested with the notion of \emph{span} for the complexity, that is to say the complexity of a process under maximal parallelism. Let us first show that this relation $\tred$ behaves well with the congruence relation. 
	
	\begin{lemma}[Time Reduction and Congruence]
		If $P \congr Q$, $P \tred P'$ and $Q \tred Q'$ then $P' \congr Q'$. 
	\label{l:timereductioncongr}
	\end{lemma}
	
	\begin{proof}
		By induction on $P \congr Q$. All the base case are direct except the last one. For this one, we first need to show that if $a$ is not free in $Q$ and $Q \tred Q'$ then $a$ is not free in $Q'$, but this is an easy induction on the definition of $\tred$. Then, the case of reflexivity, symmetry and transitivity are straightforward, as well for contextual rules. 
	\end{proof}
	
	As explained before, the relation $\tred$ stands for one unit of time, and a reduction $\red$ will be considered to take zero unit of time. Following this intuition, we impose a strategy of reduction for terms, saying that ''before doing an expensive reduction ($\tred$), we first reduce as much as possible using the zero-cost reduction ($\red$)''. So the strategy we are interested in is the following one: 
	\begin{definition}[Reduction Strategy]
		We define the \emph{tick-last strategy} by the following steps:
	\begin{enumerate}
		\item We start from a process $P$.
		\item We reduce $P$ to $P'$ such that $P \red^* P'$ and such that $P'$ is in normal form for $\red$.
		\item 
			\begin{itemize} 
				\item If $P' \tred P'$, we stop the computation. 
				\item If $P' \tred Q$ with $Q \ne P'$, then we go back to $1.$ with $Q$ instead of $P$
			\end{itemize} 
	\end{enumerate} 
	\label{d:reductionstrategy}
	\end{definition}
	
	With this strategy comes a notion of complexity: the complexity of a reduction from $P$ to $Q$ with the tick-last strategy is the number of time reductions ($\tred$) in the reduction. This corresponds to the \emph{span} of a process, that is to say the complexity of a process with maximal parallelism. Indeed, in this strategy, a process first execute all zero-cost reduction, and then all guarded processes move forward one unit of time, that is to say all top guarded processes remove one tick or stay idle.  
	
	In the same way we showed previous lemmas for pre-processes, we obtain existence and uniqueness of the canonical form for processes, we can also give an exhaustive list for the shape of the reduction $\red$ for processes as in Lemma~\ref{l:exhaustivereductions} and we can deduce the following lemma.
	
	\begin{lemma}[Typed Closed Normal Form with Tick]
		A process $P$ with $\vdash P$ is in normal form for $\red$ if and only if 
		\[\leftshift{P \congr \pnu[(a_1, \dots a_n)](\pserv[b_1][\vect{\var_1^0}][P_1] \parr \cdots \parr \pserv[b_k][\vect{\var_k^0}][P_k] \parr \pin[c_1][\vect{\var_1^1}][Q_1] \parr \cdots \parr \pin[c_{k'}][\vect{\var_{k'}^1}][Q_{k'}] \parr \pout[d_1][\vect{e_1}] \parr \cdots \parr \pout[d_{k''}][\vect{e_{k''}}] \parr \tick.R_1 \parr \cdots \parr \tick.R_{k'''})} \]
		with $(\{ b_i \mid 1 \le i \le k \} \cup \{ c_i \mid 1 \le i \le k' \}) \cap \{ d_i \mid 1 \le i \le k\} = \emptyset$.
	\label{l:closedtypednormalform2}	
	\end{lemma}
	
	Moreover, as for Lemma~\ref{l:exhaustivereductions}, we can express a generic form for the $\tred$ relation. 
	
	\begin{lemma}
	If $R \tred R'$, then 
		\item  \[ R \congr \pnu[\vect{b}] (R_1 \parr \cdots \parr R_n \parr \tick.P_1 \parr \cdots \parr \tick.P_n) \]
		and 
		\[ R' \congr \pnu[\vect{b}] (R_1 \parr \cdots \parr R_n \parr P_1 \parr \cdots \parr P_n) \]
		where $R_1 \dots, R_n$ are guarded processes that do not start with a tick. 
	\label{l:exhaustivetimereductions}
	\end{lemma}
	
	From this lemma we can deduce that a process $P$ satisfies $P \tred P$ if and only if in the top guarded processes of its canonical form, none of them start with a tick. 
	
	Now, we show that the tick-last strategy is well-behaved according to the standard reduction in $\pi$-calculus. In order to do that, let us
	first introduce a notation for the erasure of ticks. 
	
	\begin{definition}[Eliminating Tick]
		For a process $P$, we define a pre-process $\gr{P}$ corresponding to $P$ in which all tick constructor have been erased. 
		\begin{itemize}
			\item $\gr{\pzero} = \pzero$ 
			\item $\gr{P \parr Q} = \gr{P} \parr \gr{Q}$
			\item $\gr{\pserv} = \pserv[@][@][\gr{P}]$
			\item $\gr{\pin} = \pin[@][@][\gr{P}]$
			\item $\gr{\pout} = \pout$
			\item $\gr{\pnu P} = \pnu \gr{P}$
			\item $\gr{\pifn} = \pifn[@][\gr{P}][@][\gr{Q}]$
			\item $\gr{\pifl} = \pifl[@][\gr{P}][@][@][\gr{Q}]$
			\item $\gr{\pif} = \pif[@][\gr{P}][\gr{Q}]$
			\item $\gr{\tick. P} = \gr{P}$
		\end{itemize}
	\label{d:eliminatingtick}
	\end{definition}
	
	And we can now show that the tick constructor and the tick-last strategy does not create new path of reduction in a term, and if the computation terminates, then we obtain a normal form for $\red$. 
	
	\begin{theorem}
		If a process $P$ reduces to $Q$ by the strategy of Definition~\ref{d:reductionstrategy}, then $Q$ is in normal form for $\red$ and invariant by $\tred$. Moreover, $\gr{P} \red^* \gr{Q}$ and $\gr{Q}$ is in normal form for $\red$. 
	\label{t:wellbehavedtimereduction}
	\end{theorem} 

	To begin with, remark that here the notation $\red$ denote both the reduction relation on processes and the reduction relation on pre-processes. As those 
	two relations are defined by the same rules, we keep the same notation for both.
	
	\begin{proof}
		The first part of this theorem is a direct consequence of Definition~\ref{d:reductionstrategy}, as the computation stops only for processes invariant by $\tred$, and we only apply $\tred$ to processes in normal form for $\red$. 
		
		In order to show the second part of this theorem, we must first prove the following lemmas. 
		
		\begin{lemma}
			If $P \congr Q$ then $\gr{P} \congr \gr{Q}$.  
		\label{l:eliminatingtickcongr}	
		\end{lemma}
		This can be shown directly by induction on $\congr$. 
		\begin{lemma}
		  If $P \red Q$ then $\gr{P} \red \gr{Q}$.
		\label{l:eliminatingtickred}
		\end{lemma}
		Again, this is straightforward by induction on $\red$ using Lemma~\ref{l:eliminatingtickcongr}. 
		
		We can now prove Theorem~\ref{t:wellbehavedtimereduction} by recurrence on the number of time reduction ($\tred$) from $P$ to $Q$.
		\begin{itemize}
			\item If there are no time reduction from $P$ to $Q$, then by definition of the strategy, we have $P \red^* Q$, $Q$ is in normal form for $\red$ and $Q \tred Q$. By Lemma~\ref{l:eliminatingtickred}, we obtain $\gr{P} \red^* \gr{Q}$. Moreover, as $Q \tred Q$, by Lemma~\ref{l:exhaustivetimereductions}, the canonical form of $Q$ has the shape $\pnu[\vect{a}](Q_1 \parr \cdots \parr Q_n)$ where $Q_1, \dots Q_n$ are guarded processes that do not start with a tick. As a consequence, by Lemma~\ref{l:eliminatingtickcongr}, $\gr{Q} \congr \pnu[\vect{a}](\gr{Q_1} \parr \cdots \parr \gr{Q_n})$. As $Q_1,\dots,Q_n$ do not start with a tick, $\gr{Q_i}$ has the same top constructor as $Q_i$ for all $i$. Moreover, $Q$ is in normal form for $\red$. Using Lemma~\ref{l:exhaustivereductions}, we can deduce that $\gr{Q}$ is also in normal form for $\red$. Indeed, if the canonical form of $\gr{Q}$ was one of those expressed in Lemma~\ref{l:exhaustivereductions}, then the canonical form of $Q$ would have the same shape and so $Q$ would not be in normal form for $\red$. This concludes this case.
			\item If there are at least one time reduction from $P$ to $Q$, then by definition of the strategy, we have $P \red^* P'$, $P'$ in normal form for $\red$ and $P' \tred P''$ such that $P''$ can be reduced to $Q$ by the strategy. By Lemma~\ref{l:eliminatingtickred}, we have $\gr{P} \red^* \gr{P'}$. Moreover, with Lemma~\ref{l:exhaustivetimereductions}, we can see that if $P' \tred P''$ then $\gr{P'} \congr \gr{P''}$ by Lemma~\ref{l:eliminatingtickcongr}. Finally, by recurrence hypothesis, we have $\gr{P''} \red^* \gr{Q}$ and $\gr{Q}$ is in normal form for $\red$. Thus, we obtain $\gr{P} \red^* \gr{Q}$ and $\gr{Q}$ is indeed in normal form for $\red$. This concludes this case.
		\end{itemize}
		Finally, we obtain Theorem~\ref{t:wellbehavedtimereduction}.
	\end{proof}
	
	Remark that the tick-last strategy is not deterministic nor confluent, as $\red$ is not. However, the tick constructor can enforce some reduction in a term. For example, following the strategy, the process $\pin \parr \pout \parr \tick. \pout[@][\vect{e'}]$ can only reduce to $\psub \parr \tick. \pout[@][\vect{e'}]$ with $\red$, while without tick we have $\pin \parr \pout \parr \pout[@][\vect{e'}] \red \pout \parr \psub[@][@][\vect{e'}]$. A consequence of this is that sometimes, adding a tick can enforce an infinite sequence of reduction by forbidding the terminating run. In a sense, the tick constructor allows the concept of race in a process. For example, the process $\pin[@][@][\Omega] \parr \pout \parr \tick. \pin[@][@][\pzero]$, where $\Omega$ is non terminating for $\red$, will always have infinite reductions.  
	
	Remark that with this concept of race, the tick-last strategy may not be the fastest way to reach a precise normal form. Take for example this process $\tick. \pserv \parr \pserv[@][@][P'] \parr \pout$. If $P$ and $P'$ has the same behaviour (for example sorting a list given on input and sending it to another channel) but $P$ is faster than $P'$, then the tick-last strategy enforces to take the slower reduction. 
	
	\subsection{Size Types with Temporal Information}
	\label{ss:sizedtypesystem}
	
	We will now define a type system to bound the span of a process. The goal is to obtain a soundness result: if a process is typable then we can derive an integer $K$ such that the complexity of any reduction following the strategy of Definition~\ref{d:reductionstrategy} is bounded by $K$. 
	
	Our type system relies on the definition of indices that give more information about the type. Those indices were for example used in \cite{DalLagoGaboardiLinearDependentTypes} in the non-concurrent case. We also enrich type with temporal information, following the idea of \cite{DasHoffmannPfenningTemporalSessionTypes} to obtain complexity bound. 
	
	\begin{definition}
		The set of \emph{indices} for natural number is given by the following grammar. 
		\[ I,J,K := i,j,k \midd f(I_1,\dots,I_n) \]
		The variables $i,j,k$ are called index variables. The set of index variables is denoted $\IV$. We suppose given a set of function symbol containing for example the addition and the multiplication. We assume that each function symbol $f$ comes with an interpretation $\sem{f} : \NN^{\mathtt{ar}(f)} \rightarrow \NN $. 
	\label{d:indexes}
	\end{definition}
	
	Given an index valuation $\rho : \IV \rightarrow \NN$, we extend the interpretation of function symbols to indices, noted $\sem[\rho]{I}$ in the natural way. In an index $I$, the substitution of the occurences of $i$ in $I$ by $J$ is noted $\Isub$. We also assume that we have the subtraction as a function symbol, with $n - m = 0$ when $m \ge n$. 
	
	\begin{definition}[Constraints on Indices]
		Let $\phi \subset \IV$ be a set of index variables. A \emph{constraint} $C$ on $\phi$ is an expression with the shape $I \rel J$ where $I$ and $J$ are indices with free variables in $\phi$ and $\rel$ denotes a binary relation on integers. Usually, we take $\rel \in \{ \le, <, =, \ne \}$. Finite set of constraints are denoted $\Phi$. 
	\label{d:constraints}
	\end{definition}
	
	We say that a valuation $\rho : \IV \rightarrow \NN$ \emph{satisfies} a constraint $I \rel J$, noted $\rho \vDash I \rel J$ when $\sem[\rho]{I} \rel \sem[\rho]{J}$ holds. Similarly, $\rho \vDash \Phi$ holds when $\rho \vDash C$ for all $C \in \Phi$. Likewise, we note $\phi;\Phi \vDash C$ when for all valuation $\rho$ such that $\rho \vDash \Phi$ we have $\rho \vDash C$. 
	
	\begin{definition}
		The set of types and \emph{base types} are given by the following grammar.
		\begin{align*}
			\IB &:= \Inat \midd \Ilis \midd \Ibool \\
			T &:= \IB \midd \Ich \midd \Iich \midd \Ioch \midd \Iserv \midd \Iiserv \midd \Ioserv \\	
		\end{align*}
	\label{d:types}
	\end{definition}
	A type $\Ich$, $\Iich$ or $\Ioch$ is called a \emph{channel type} and a type $\Iserv$, $\Iiserv$ or $\Ioserv$ is called a \emph{server} type. 
	For a channel type or a server type, the index $I$ is called the \emph{time} of this type, and for a server type, the index $K$ is called the \emph{complexity} of this type.
	
	Intuitively, an integer $n$ of type $\Inat$ must be such that $I \le n \le J$. Likewise, a list of type $\Ilis$ have a size between $I$ and $J$. To give a channel variable the type $\Ich$ ensures that its communication should happen at time $I$. For example, a channel variable of type $\Ich[0][@]$ should do its communication before any tick occurs. Likewise, a name of type $\Iiserv$ must be used in a replicated input, and this replicated input must be ready to receive for any time greater than $I$. Typically, a process $\tick. \pserv$ enforces that the type of $a$ is $\Iiserv$ with $I$ greater than one, as the replicated input is not ready to receive at time $0$. 
	
	Moreover, a server type has a kind of polymorphism for indices, and the index $K$ stands for the complexity of the process spawned by this server. A typical example is a server taking as input a list and a channel, and send to this channel the sorted list, in time $k \cdot n$ where $n$ is the size of the list. 
	\[ P = \pserv[@][x,b][{\pout[b][\mathsf{sort}(x)]}] \]. Such a server name $a$ could be given the type $\Iserv[0][i][k \cdot i][{\Ilis[0][i][\IB], \Ioch[k \cdot i][{\Ilis[0][i][\IB]}]}]$. This means that this server is ready to receive an input and, for all integer $i$, if given a list of size at most $i$ and an output channel doing its communication at time $k \cdot i$ and waiting for a list of size at most $i$, the process spawned by this server will stop at time at most $k \cdot i$. 
	 
	We define a notion of subtyping for size types. The rules are given in Figure~\ref{f:sizesubtype}. 
	
	\begin{figure}
		\centering
		\begin{framed}
			\small 
			\begin{center}
				\AxiomC{$\phi;\Phi \vDash I' \le I$}
				\AxiomC{$\phi;\Phi \vDash J \le J'$}
				\BinaryInfC{$\phi;\Phi \vdash \Inat \subtype \Inat[I'][J']$}
				\DisplayProof 
				\qquad 
				\AxiomC{$\phi;\Phi \vDash I' \le I$}
				\AxiomC{$\phi;\Phi \vDash J \le J'$}
				\AxiomC{$\phi;\Phi \vdash \IB \subtype \IB' $}
				\TrinaryInfC{$\phi;\Phi \vdash \Ilis \subtype \Ilis[I'][J'][\IB']$}
				\DisplayProof 
				\\
				\vvskip 
				\AxiomC{}
				\UnaryInfC{$\phi; \Phi \vdash \Ibool \subtype \Ibool$}
				\DisplayProof
				\qquad  
				\AxiomC{$\phi;\Phi \vDash I = J$}
				\AxiomC{$\phi; \Phi \vdash \vect{T} \subtype \vect{U}$}
				\AxiomC{$\phi; \Phi \vdash \vect{U} \subtype \vect{T}$}
				\TrinaryInfC{$\phi; \Phi \vdash \Ich \subtype \Ich[J][\vect{U}]$}
				\DisplayProof
				\qquad 
				\AxiomC{}
				\UnaryInfC{$\phi; \Phi \vdash \Ich \subtype \Iich$}
				\DisplayProof 
				\\
				\vvskip 
				\AxiomC{}
				\UnaryInfC{$\phi;\Phi \vdash \Ich \subtype \Ioch$}
				\DisplayProof 
				\qquad 
				\AxiomC{$\phi;\Phi \vDash I = J$}
				\AxiomC{$\phi;\Phi \vdash \vect{T} \subtype \vect{U}$}
				\BinaryInfC{$\phi;\Phi \vdash \Iich \subtype \Iich[J][\vect{U}]$}
				\DisplayProof 
				\qquad 		
				\AxiomC{$\phi;\Phi \vDash I = J$}
				\AxiomC{$\phi; \Phi \vdash \vect{U} \subtype \vect{T}$}
				\BinaryInfC{$\phi;\Phi \vdash \Ioch \subtype \Ioch[J][\vect{U}]$}
				\DisplayProof 
				\\ 
				\vvskip 
				\AxiomC{$\phi;\Phi \vDash I = J$}
				\AxiomC{$(\phi,\vect{i}); \Phi \vdash \vect{T} \subtype \vect{U}$}
				\AxiomC{$(\phi,\vect{i}); \Phi \vdash \vect{U} \subtype \vect{T}$}
				\AxiomC{$(\phi,\vect{i}); \Phi \vDash K = K'$}
				\QuaternaryInfC{$\phi; \Phi \vdash \Iserv \subtype \Iserv[J][\vect{i}][K'][\vect{U}]$}
				\DisplayProof
				\\ 
				\vvskip 
				\AxiomC{}
				\UnaryInfC{$\phi; \Phi \vdash \Iserv \subtype \Iiserv$}
				\DisplayProof
				\qquad  
				\AxiomC{}
				\UnaryInfC{$\phi;\Phi \vdash \Iserv \subtype \Ioserv$}
				\DisplayProof 
				\\
				\vvskip 			
				\AxiomC{$\phi;\Phi \vDash I = J$}
				\AxiomC{$(\phi,\vect{i});\Phi \vdash \vect{T} \subtype \vect{U}$}
				\AxiomC{$(\phi,\vect{i});\Phi \vDash K' \le K$}
				\TrinaryInfC{$\phi;\Phi \vdash \Iiserv \subtype \Iiserv[J][\vect{i}][K'][\vect{U}]$}
				\DisplayProof 
				\\
				\vvskip 
				\AxiomC{$\phi;\Phi \vDash I = J$}
				\AxiomC{$(\phi,\vect{i});\Phi \vdash \vect{U} \subtype \vect{T}$}
				\AxiomC{$(\phi,\vect{i});\Phi \vDash K \le K'$}
				\TrinaryInfC{$\phi;\Phi \vdash \Ioserv \subtype \Ioserv[J][\vect{i}][K'][\vect{U}]$}
				\DisplayProof 
				\qquad 
				\AxiomC{$\phi;\Phi \vdash T \subtype T'$}
				\AxiomC{$\phi;\Phi \vdash T' \subtype T''$}
				\BinaryInfC{$\phi;\Phi \vdash T \subtype T''$}
				\DisplayProof 
				
			\end{center}   
		\end{framed}
		\caption{Subtyping Rules for Sized Types}
		\label{f:sizesubtype}
	\end{figure}
	
	The subtyping for channel type is standard, the only new thing is that we impose that the time of communication is invariant. Moreover, for servers, we can also change the complexity $K$ in subtyping: for input servers, we can always define something faster than announced, and for output, we can always consider that a server is slower than announced. 
	
	In order to present to type system of our calculus, let us first introduce some notation.
	
	\begin{definition}[Advancing Time in Types]
		Given a set of index variables $\phi$, a set of constraint $\Phi$, a type $T$ and an index $I$. We define $T$ after $I$ unit of time, denoted $\Idecr[T][@][\phi;\Phi]$ by: 
		\begin{itemize}
			\item $\Idecr[\IB][@][\phi;\Phi] = \IB$
			\item $\Idecr[{\Ich[J]}][@][\phi;\Phi] = \Ich[(J-I)]$ if $\phi;\Phi \vDash J \ge I$. It is undefined otherwise. Other channel types follow exactly the same pattern.
			\item $\Idecr[{\Iserv[J]}][@][\phi;\Phi] = \Iserv[(J-I)]$ if $\phi;\Phi \vDash J \ge I$. Otherwise, $\Idecr[{\Iserv[J]}][@][\phi;\Phi] = \Ioserv[(J-I)]$
			\item $\Idecr[{\Iiserv[J]}][@][\phi;\Phi] = \Iiserv[(J-I)]$ if $\phi;\Phi \vDash J \ge I$. It is undefined otherwise.
			\item $\Idecr[{\Ioserv[J]}][@][\phi;\Phi] = \Ioserv[(J-I)]$.  
		\end{itemize} 
		This definition can be extended to contexts, with $\Idecr[\var : T, \Gamma][@][\phi;\Phi] = \var : \Idecr[T][@][\phi;\Phi], \Idecr[@][@][\phi;\Phi]$ if $\Idecr[T][@][\phi;\Phi]$ is defined. Otherwise, $\Idecr[\var : T, \Gamma] = \Idecr[@][@][\phi;\Phi]$. We will often omit the $\phi;\Phi$ in the notation when it is clear from the context. 
	\label{d:advancetime}
	\end{definition} 
	Let us precise a bit the definition here. Intuitively, $T$ after $I$ unit of time is the type $T$ with a time decreased by $I$. For base types, there is no time thus nothing happens. Then, one can wonder what happens when the time of $T$ is smaller than $I$. For non-server channel types, we consider that their time is over, thus we erase them from the context. For servers this is a bit more complicated. Intuitively, when a server is defined, it should stay available until the end. Thus, an output to a server should always be possible, no matter the time. However, there are also some time limitation in servers in the sense that we must respect the time limit to define a server. As a consequence, when time advances too much, we should not be able to define a server anymore. That is why servers lose their input capability when time advances too much. Let us now show that this definition behaves well with subtyping.
	
	\begin{lemma}
		If $\phi;\Phi \vdash T \subtype U$ then for any $I$, either $\Idecr[U]$ is undefined, or both $\Idecr[U]$ and $\Idecr[T]$ are defined, and $\phi;\Phi \vdash \Idecr[T] \subtype \Idecr[U]$.  
	\label{l:advanceandsubtyping}
	\end{lemma} 
	
	\begin{proof}
		The proof is pretty straightforward by induction on $\phi;\Phi \vdash T \subtype U$. 
	\end{proof}
	
	A corollary of Lemma~\ref{l:advanceandsubtyping} is that if $\phi;\Phi \vdash \Gamma \subtype \Delta$, then $\Idecr = \Gamma',\Delta'$ with $\phi;\Phi \vdash \Delta' \subtype \Idecr[\Delta]$. 
	
	\begin{definition}
		Given a set of index variables $\phi$ and a set of constraints $\Phi$, a context $\Gamma$ is said to be \emph{time invariant} when it contains only base type variables or output server types $\Ioserv$ with $\phi;\Phi \vDash I = 0$.  
	\label{d:timeinvariant}
	\end{definition}   
	
	Such a context is thus invariant by the operator $\Idecr[\cdot][I]$ for any $I$. This is typically the kind of context that we need to define a server. We can now present the type system. Rules are given in Figure~\ref{f:sizetypeexpression} and Figure~\ref{f:sizetypeprocess}. A typing $\Itype{P}{K}$ means intuitively that under the constraints $\Phi$, in the context $\Gamma$, a process $P$ is typable and its complexity is bounded by $K$. And the typing for expressions $\Ietype{e}{T}$ means that under the constraints $\Phi$, in the context $\Gamma$, the expression $e$ can be given the type $T$. 
	
	\begin{figure}
		\centering
		\begin{framed}
			\small 
			\begin{center}
				\AxiomC{$\var : T \in \Gamma$}
				\UnaryInfC{$\Ietype{\var}{T}$}
				\DisplayProof 
				\qquad
				\AxiomC{}
				\UnaryInfC{$\Ietype{\zero}{\Inat[0][0]}$}
				\DisplayProof 
				\qquad
				\AxiomC{$\Ietype{e}{\Inat}$}
				\UnaryInfC{$\Ietype{\suc(e)}{\Inat[I+1][J+1]}$}
				\DisplayProof 
				\\ 
				\vvskip 
				\AxiomC{}
				\UnaryInfC{$\Ietype{\nil}{\Ilis[0][0]}$}
				\DisplayProof 
				\qquad 
				\AxiomC{$\Ietype{e}{\IB}$}
				\AxiomC{$\Ietype{e'}{\Ilis} $}
				\BinaryInfC{$\Ietype{e::e'}{\Ilis[I+1][J+1]}$}
				\DisplayProof 
				\qquad
				\AxiomC{}
				\UnaryInfC{$\Ietype{\true}{\Ibool}$}
				\DisplayProof 
				\\ 
				\vvskip 
				\AxiomC{}
				\UnaryInfC{$\Ietype{\false}{\Ibool}$}
				\DisplayProof
				\qquad 
				\AxiomC{$\Ietype[@][@][\Delta]{e}{U}$}
				\AxiomC{$\phi;\Phi \vdash \Gamma \subtype \Delta$}
				\AxiomC{$\phi; \Phi \vdash U \subtype T$}
				\TrinaryInfC{$\Ietype{e}{T}$}
				\DisplayProof 
			\end{center}   
		\end{framed}
		\caption{Typing Rules for Expressions}
		\label{f:sizetypeexpression}
	\end{figure} 
	
	\begin{figure}
		\centering
		\begin{framed}
			\footnotesize 
			\begin{center}
				\AxiomC{}
				\UnaryInfC{$\Itype{\pzero}{0}$}
				\DisplayProof 
				\qquad
				\AxiomC{$\Itype{P}{K}$}
				\AxiomC{$\Itype{Q}{K}$}
				\BinaryInfC{$\Itype{P \parr Q}{K}$}
				\DisplayProof 
				\\ 
				\vvskip 
				\AxiomC{$\Ietype[@][@][\Gamma,\Delta]{a}{\Iiserv}$}
				\AxiomC{$\Itype[(\phi,\vect{i})][@][\Gamma', \vect{\var} : \vect{T}]{P}{K}$}
				\AxiomC{$\Gamma'$ time invariant}
				\AxiomC{$\phi;\Phi \vdash \Idecr[@][@][\phi;\Phi] \subtype \Gamma'$}
				\QuaternaryInfC{$\Itype[@][@][\Gamma,\Delta]{\pserv}{I}$}
				\DisplayProof 
				\\ 
				\vvskip  			
				\AxiomC{$\Ietype{a}{\Iich}$}
				\AxiomC{$\Itype[@][@][\Idecr, \vect{\var} : \vect{T}]{P}{K}$}
				\BinaryInfC{$\Itype{\pin}{K+I}$}
				\DisplayProof 
				\qquad
				\AxiomC{$\Ietype{a}{\Ioch}$}
				\AxiomC{$\Ietype[@][@][\Idecr]{\vect{e}}{\vect{T}}$}
				\BinaryInfC{$\Itype{\pout}{I}$}
				\DisplayProof 
				\\ 
				\vvskip 
				\AxiomC{$\Ietype{a}{\Ioserv}$}
				\AxiomC{$\Ietype[@][@][\Idecr]{\vect{e}}{\Isub[\vect{T}][\vect{i}][\vect{J}]}$}
				\BinaryInfC{$\Itype{\pout}{I + \Isub[K][\vect{i}][\vect{J}]}$}
				\DisplayProof 
				\qquad
				\AxiomC{$\Itype[@][@][\Gamma, a : T]{P}{K}$}
				\UnaryInfC{$\Itype{\pnu P}{K}$}
				\DisplayProof 
				\\ 
				\vvskip 
				\AxiomC{$\Ietype{e}{\Inat}$}
				\AxiomC{$\Itype[@][(\Phi,I \le 0)][@]{P}{K}$}
				\AxiomC{$\Itype[@][(\Phi,J \ge 1)][{\Gamma, x : \Inat[I-1][J-1]}]{Q}{K}$}
				\TrinaryInfC{$\Itype{\pifn}{K}$}
				\DisplayProof 
				\\
				\vvskip   
				\AxiomC{$\Ietype{e}{\Ilis}$}
				\AxiomC{$\Itype[@][(\Phi,I \le 0)][@]{P}{K}$}
				\AxiomC{$\Itype[@][(\Phi,J \ge 1)][{\Gamma, x : \IB, y : \Ilis[I-1][J-1]}]{Q}{K}$}
				\TrinaryInfC{$\Itype{\pifl}{K}$}
				\DisplayProof 
				\\
				\vvskip 
				\AxiomC{$\Ietype{e}{\Ibool}$}
				\AxiomC{$\Itype{P}{K}$}
				\AxiomC{$\Itype{Q}{K}$}
				\TrinaryInfC{$\Itype{\pif}{K}$}
				\DisplayProof 
				\qquad 
				\AxiomC{$\Itype[@][@][{\Idecr[@][1]}]{P}{K}$}
				\UnaryInfC{$\Itype{\tick.P}{K+1}$}
				\DisplayProof 
				\\
				\vvskip 
				\AxiomC{$\Itype[@][@][\Delta]{P}{K}$}
				\AxiomC{$\phi;\Phi \vdash \Gamma \subtype \Delta$}
				\AxiomC{$\phi;\Phi \vDash K \le K'$}
				\TrinaryInfC{$\Itype{P}{K'}$}
				\DisplayProof
			\end{center}   
		\end{framed}
		\caption{Typing Rules for Processes}
		\label{f:sizetypeprocess}
	\end{figure}
	
	The type system for expressions should not be very surprising. Still, remark that to type a channel name, the only possible rule is the subtyping one. In Figure~\ref{f:sizetypeprocess}, subtyping allows to increase the bound on the complexity. Then, the rule for parallel composition shows that we consider parallel complexity as we take the maximum between the two processes instead of the sum. In the typing for input server, we integrate some weakening on context, and we want a time invariant process to type the server, as a server should not depend on the time. Note also that a server alone has no complexity, it is a call on this server that generates complexity, as we can see in the rule for output with server types. Some rules make the time advance in their continuation, for example the tick rule or input rule. This is expressed by the advance time operator on contexts. 
	
	Finally, remark that if we remove all size annotation and merge server types and channel types together to obtain back the types of Section~\ref{s:Simpletypes}, then all the rules of Figure~\ref{f:sizetypeexpression} and Figure~\ref{f:sizetypeprocess} are admissible in the type system of Figure~\ref{f:typeexpression} and Figure~\ref{f:typeprocess}.
	
	\begin{definition}[Forgetting Sizes]
	 Formally, given a sized type $T$, we define $\IU(T)$ the usual input/output type ($\IU$ is for forgetful) by:
	 \[ \IU(\Inat) := \nat \qquad \IU(\Ilis) := \lis(\IU(\IB)) \qquad \IU(\Ibool) := \bool \]
	 \[ \IU(\Ich) := \ch[\IU(\vect{T})] \qquad \IU(\Iich) := \ich[\IU(\vect{T})] \qquad \IU(\Ioch) := \och[\IU(\vect{T})] \]
	 \[ \IU(\Iserv) := \ch[\IU(\vect{T})] \qquad \IU(\Iiserv) := \ich[\IU(\vect{T})] \qquad \IU(\Ioserv) := \och[\IU(\vect{T})]\]
	\label{d:forgetful} 
	\end{definition}
	Then, we obtain the following lemma.
	
	\begin{lemma}
		If $\phi :\Phi \vdash T \subtype T'$ then $\IU(T) \subtype \IU(T')$. Moreover, if $\Ietype{e}{T}$ then $\IU(\Gamma) \vdash e : \IU(T)$ and if $\Itype{P}{K}$ then $\IU(\Gamma) \vdash P$
	\label{l:forgetfultyping}
	\end{lemma}
	
	\begin{proof}
		Once we have weakening for input/output types and that advancing time does not change the underlying input/output type when it is defined, the proof can be made by induction on the subtyping derivation or the typing.
	\end{proof}
	\section{Subject Reduction and Complexity Bound}
	\label{s:subjectreduction}
	
	In this section, we prove that our type system can indeed give a bound on the number of time reduction of a process following the strategy of Definition~\ref{d:reductionstrategy}.
	
	\subsection{Intermediate Lemmas}
	\label{ss:intlemmas}
	
	We first show some usual and intermediate lemmas on the typing system.  
	
	\begin{lemma}[Weakening]
		Let $\phi,\phi'$ be disjoint set of index variables, $\Phi$ be a set of constraint on $\phi$, $\Phi'$ be a set of constraints on $(\phi,\phi')$, $\Gamma$ and $\Gamma'$ be contexts on disjoint set of variables. 
		\begin{enumerate}
			\item If $\phi;\Phi \vDash C$ then $(\phi,\phi');(\Phi,\Phi') \vDash C$. 
			\item If $\phi;\Phi \vdash T \subtype U$ then $(\phi,\phi');(\Phi,\Phi') \vdash T \subtype U$.
			\item If $\Ietype{e}{T}$ then $\Ietype[(\phi,\phi')][(\Phi,\Phi')][\Gamma,\Gamma']{e}{T}$. 
			\item $\Idecr[@][@][(\phi,\phi');(\Phi,\Phi')] = \Delta ,\Delta'$ with $(\phi;\phi');(\Phi;\Phi') \vdash \Delta \subtype \Idecr[@][@][\phi;\Phi]$. 
			\item If $\Itype{P}{K}$ then $\Itype[(\phi,\phi')][(\Phi,\Phi')][\Gamma,\Gamma']{P}{K}$.  
		\end{enumerate}
	\label{l:weakening}
	\end{lemma}
	
	\begin{proof}
		Point~1 is a direct consequence of the definition of $\phi;\Phi \vDash C$. Point~2 is proved by induction on the subtyping derivation, and it uses explicitly Point~1. Point~4 is a consequence of Point~1: everything that is defined in $\Idecr[@][@][\phi;\Phi]$ is also defined in $\Idecr[@][@][(\phi,\phi');(\Phi,\Phi')]$, and the subtyping condition is here since with more constraints, a server may not be changed into an output server by the advance of time. Point~3 and Point~5 are proved by induction on the typing derivation, and each point uses crucially the previous ones. Note that the weakening integrated in the rule for input servers is primordial to obtain Point~5. Note also that when the advance time operator is used, the weakened typing is obtained with the use of a subtyping rule. 
	\end{proof}
	
	We also show that we can remove some useless hypothesis. 
	
	\begin{lemma}[Strengthening]
		Let $\phi$ be a set of index variables, $\Phi$ be a set of constraint on $\phi$, and $C$ a constraint on $\phi$ such that $\phi;\Phi \vDash C$. 
		\begin{enumerate}
			\item If $\phi;(\Phi,C) \vDash C'$ then $\phi;\Phi \vDash C'$. 
			\item If $\phi;(\Phi,C) \vdash T \subtype U$ then $\phi;\Phi \vdash T \subtype U$.
			\item If $\Ietype[@][(\Phi,C)][\Gamma,\Gamma']{e}{T}$ and the variables in $\Gamma'$ are not free in $e$, then $\Ietype{e}{T}$.
			\item $\Idecr[@][@][\phi;(\Phi,C)] = \Idecr[@][@][\phi;\Phi]$. 
			\item If $\Itype[@][(\Phi,C)][\Gamma,\Gamma']{P}{K}$ and the variables in $\Gamma'$ are not free in $P$, then $\Itype{P}{K}$.  
		\end{enumerate}
		\label{l:strengthening}
	\end{lemma}
	
	\begin{proof}
		Point~1 is a direct consequence of the definition. Point~2 is proved by induction on the subtyping derivation. Point~4 is straightforward with Point~1 of this lemma and Point~1 of Lemma~\ref{l:weakening}. Point~3 and Point~5 are proved by induction on the typing derivation. 
	\end{proof}
	
	Then, we prove that index variables can indeed be substituted by any other indexes. 
	
	\begin{lemma}[Index Substitution]
		Let $\phi$ be a set of index variable and $i \notin \phi$. Let $J$ be an index with free variables in $\phi$. Then, 
		\begin{enumerate}
			\item $\sem[\rho]{\Isub} = \sem[{\rho \lbrack i \mapsto \sem[\rho]{J} \rbrack }]{I}$.
			\item If $(\phi,i);\Phi \vDash C$ then $\phi;\Isub[\Phi] \vDash \Isub[C]$. 
			\item If $(\phi,i);\Phi \vdash T \subtype U$ then $\phi;\Isub[\Phi] \vdash \Isub[T] \subtype \Isub[U]$.
			\item If $\Ietype[(\phi,i)]{e}{T}$ then $\Ietype[@][{\Isub[\Phi]}][{\Isub[\Gamma]}]{e}{\Isub[T]}$.
			\item $\Idecr[{\Isub[\Gamma]}][\Isub][{\phi;\Isub[\Phi]}] = \Delta ,\Delta'$ with $\phi;\Isub[\Phi] \vdash \Delta \subtype \Isub[({\Idecr[@][@][(\phi,i);\Phi]})]$. 
			\item If $\Itype[(\phi,i)]{P}{K}$ then $\Itype[@][{\Isub[\Phi]}][{\Isub[\Gamma]}]{P}{\Isub[K]}$.
		\end{enumerate}
	\label{l:indexsubstitution}
	\end{lemma}
	
	\begin{proof}
		Point~1 is proved by induction on $I$. Then, Point~2 is a rather direct consequence of Point~1. Point~3 is proved by induction on the subtyping derivation, then Point~4 is proved by induction on the typing derivation. Point~5 is direct with the use of Point~2. And finally Point~6 is proved by induction on $P$. The induction is on $P$ and not the typing derivation because of Point~5 that forces the use of weakening (Lemma~\ref{l:weakening}). 
	\end{proof}
	
	Those lemmas are rather usual in an index based system. However, the following one relies directly on our notion of time and the type system. 
	
	\begin{lemma}[Delaying]
		Given a type $T$ and an index $I$, we define the \emph{delaying} of $T$ by $I$ units of time, denoted $\Iincr{T}$: 
		\[ \Iincr{\IB} = \IB \qquad \Iincr{(\Ich[J])} = \Ich[J+I] \]
		and for other channel and server types, the definition is the same as the one on the right above. This definition can be extended to contexts. With this, we have:
		\begin{enumerate}
			\item If $\phi;\Phi \vdash T \subtype U$ then $\phi;\Phi \vdash \Iincr{T} \subtype \Iincr{U}$.
			\item If $\Ietype{e}{T}$ then $\Ietype[@][@][\Iincr{\Gamma}]{e}{\Iincr{T}}$.
			\item $\Idecr[\Iincr{\Gamma}][J][\phi;\Phi] = \Delta, \Delta'$ with $\phi;\Phi \vdash \Delta \subtype \Iincr{({\Idecr[\Gamma][J][\phi;\Phi]})}$.
			\item $\Idecr[\Iincr{\Gamma}][(J + I)] = \Idecr[\Gamma][J]$. 
			\item If $\Itype{P}{K}$ then $\Itype[@][@][\Iincr{\Gamma}]{P}{K + I}$.
			\item For any context $\Gamma$, $\Gamma = \Gamma',\Delta $ with $\phi;\Phi \vdash \Gamma' \subtype \Iincr{(\Idecr[\Gamma])}$.
		\end{enumerate}
	\label{l:delaying}
	\end{lemma}
	
	\begin{proof}
		Point~1, Point~2, Point~3 and Point~4 are straightforward. Then, Point~5 is proved by induction on $P$. Point~4 is used on every rule for channel or servers, and Point~3 is used in the rule for tick. Point~6 is another straightforward proof. It is not used in the proof of Point~5 but it is useful for our subject reduction.  
	\end{proof}
	
	We can now show the usual variable substitution lemmas. 
	
	\begin{lemma}[Substitution]
		\begin{enumerate}
			\item If $\Ietype[@][@][\Gamma, \var : T]{e'}{U}$ and $\Ietype{e}{T}$ then $\Ietype{\psub[e'][\var][e]}{U}$.
			\item If $\Itype[@][@][\Gamma, \var : T]{P}{K}$ and $\Ietype{e}{T}$ then $\Itype{\psub[P][\var][e]}{K}$.
		\end{enumerate}
	\label{l:substitution}
	\end{lemma}
	
	The proof is pretty straightforward. 
	
	We can now show the subject reduction of this calculus with the reduction $\red$. 
	
	\subsection{Non-Quantitative Subject Reduction}
	\label{ss:nonquantitativesubjectreduction}
	
	The goal of this section is to prove the following theorem. 
	\begin{theorem}[Non-Quantitative Subject Reduction]
		If $\Itype{P}{K}$ and $P \red Q$ then $\Itype{Q}{K}$. 
	\label{t:nonquantsubjectreduction}
	\end{theorem}
	In order to do that, let us first show that the congruence relation behave well with typing. 
	
	\begin{lemma}[Congruence and Typing]
		Let $P$ and $Q$ be processes such that $P \congr Q$. Then, $\Itype{P}{K}$ if and only if $\Itype{Q}{K}$.
	\label{l:congrtype}
	\end{lemma}
	
	\begin{proof}
		In fact, we will prove something more precise, by showing that the typing conserves the typing of channel name, so if we restrain our calculus with the constructor $(\nu : T)$ to force the typing, the property still holds. Note that all previous lemmas also holds with this forced typing.
		We prove this by induction on $P \congr Q$. Remark that for a process $P$, the typing system is not syntax-directed because of the subtyping rule. However, by reflexivity and transitivity of subtyping, we can always assume that a proof has exaclty \emph{one} subtyping rule before any syntax-directed rule. 
		We first show this propriety for base case of congruence. The reflexivity is trivial then we have:
		\begin{itemize}
			\item \textbf{Case} $ P \parr \pzero \congr P$. Suppose $\Itype{P \parr \pzero}{K}$. Then the proof has the shape:
			\begin{prooftree}
				\footnotesize
				\AxiomC{$\pi$}
				\UnaryInfC{$\Itype[@][@][\Delta]{P}{K'}$}
				\AxiomC{}
				\UnaryInfC{$\Itype[@][@][\Delta']{\pzero}{0}$}
				\AxiomC{$\phi;\Phi \vdash \Delta \subtype \Delta'; 0 \le K'$}
				\BinaryInfC{$\Itype[@][@][\Delta]{\pzero}{K'}$}
				\BinaryInfC{$\Itype[@][@][\Delta]{P \parr \pzero}{K'}$}
				\AxiomC{$\phi;\Phi \vdash \Gamma \subtype \Delta ; K' \le K$}
				\BinaryInfC{$\Itype{P \parr \pzero}{K}$} 
			\end{prooftree}
		
			So, we can derive the following proof:
			\begin{prooftree}
				\footnotesize
				\AxiomC{$\pi$}
				\UnaryInfC{$\Itype[@][@][\Delta]{P}{K'}$}
				\AxiomC{$\phi;\Phi \vdash \Gamma \subtype \Delta ; K' \le K$}
				\BinaryInfC{$\Itype{P}{K}$} 
			\end{prooftree}
			
			Reciprocally, given a proof $\pi$ of $\Itype{P}{K}$, we can derive the proof:
			
			\begin{prooftree}
				\footnotesize
				\AxiomC{$\pi$}
				\UnaryInfC{$\Itype{P}{K}$}
				\AxiomC{}
				\UnaryInfC{$\Itype{\pzero}{0}$}
				\AxiomC{$\phi;\Phi \vDash 0 \le K$}
				\BinaryInfC{$\Itype{\pzero}{K}$}
				\BinaryInfC{$\Itype{P \parr \pzero}{K}$} 
			\end{prooftree}	
			
			\item \textbf{Case} $P \parr Q \congr Q \parr P$. Suppose $\Itype{P \parr Q}{K}$. Then the proof has the shape:
			
			\begin{prooftree}
				\footnotesize
				\AxiomC{$\pi$}
				\UnaryInfC{$\Itype[@][@][\Delta]{P}{K'}$}	
				\AxiomC{$\pi'$}
				\UnaryInfC{$\Itype[@][@][\Delta]{Q}{K'}$}
				\BinaryInfC{$\Itype[@][@][\Delta]{P \parr Q}{K'}$}
				\AxiomC{$\phi;\Phi \vdash \Gamma \subtype \Delta ; K' \le K$}
				\BinaryInfC{$\Itype{P \parr Q}{K}$} 
			\end{prooftree}
			
			And so we can derive:
			
			\begin{prooftree}
				\footnotesize
				\AxiomC{$\pi'$}
				\UnaryInfC{$\Itype[@][@][\Delta]{Q}{K'}$}	
				\AxiomC{$\pi$}
				\UnaryInfC{$\Itype[@][@][\Delta]{P}{K'}$}
				\BinaryInfC{$\Itype[@][@][\Delta]{Q \parr P}{K'}$}
				\AxiomC{$\phi;\Phi \vdash \Gamma \subtype \Delta ; K' \le K$}
				\BinaryInfC{$\Itype{Q \parr P}{K}$}
			\end{prooftree}
			
			We also have the reverse in the same way. 
	
			\item \textbf{Case} $P \parr (Q \parr R) \congr (P \parr Q) \parr R$. Suppose $\Itype{P \parr (Q \parr R)}{K}$. Then the proof has the shape:
			
			\medskip 
			
			\leftshiftt{
				\footnotesize
				\AxiomC{$\pi$}
				\UnaryInfC{$\Itype[@][@][\Delta]{P}{K'}$}
				\AxiomC{$\pi'$}
				\UnaryInfC{$\Itype[@][@][\Delta']{Q}{K''}$}
				\AxiomC{$\pi''$}
				\UnaryInfC{$\Itype[@][@][\Delta']{R}{K''}$}
				\BinaryInfC{$\Itype[@][@][\Delta']{Q \parr R}{K''}$}
				\AxiomC{$\phi;\Phi \vdash \Delta \subtype \Delta' ; K'' \le K'$}
				\BinaryInfC{$\Itype[@][@][\Delta]{Q \parr R}{K'}$}
				\BinaryInfC{$\Itype[@][@][\Delta]{P \parr (Q \parr R)}{K'}$}
				\AxiomC{$\phi;\Phi \vdash \Gamma \subtype \Delta ; K' \le K$}
				\BinaryInfC{$\Itype{P \parr (Q \parr R)}{K}$}
				\DisplayProof}
		
			We can derive the proof:
			
			\medskip 
			
			\leftshiftt{
				\footnotesize
				\AxiomC{$\pi$}
				\UnaryInfC{$\Itype[@][@][\Delta]{P}{K'}$}
				\AxiomC{$\pi'$}
				\UnaryInfC{$\Itype[@][@][\Delta']{Q}{K''}$}
				\AxiomC{$\phi;\Phi \vdash \Delta \subtype \Delta' ; K'' \le K'$}
				\BinaryInfC{$\Itype[@][@][\Delta]{Q}{K'}$}
				\BinaryInfC{$\Itype[@][@][\Delta]{P \parr Q}{K'}$}
				\AxiomC{$\pi''$}
				\UnaryInfC{$\Itype[@][@][\Delta']{R}{K''}$}
				\AxiomC{$\phi;\Phi \vdash \Delta \subtype \Delta' ; K'' \le K'$}
				\BinaryInfC{$\Itype[@][@][\Delta]{R}{K'}$}
				\BinaryInfC{$\Itype[@][@][\Delta]{(P \parr Q) \parr R}{K'}$}
				\AxiomC{$ \fCenter \phi;\Phi \vdash \Gamma \subtype \Delta ; K' \le K$}
				\insertBetweenHyps{\hskip -2cm}
				\BinaryInf$\Itype{(P \parr Q) \parr R}{K} \fCenter$
				\DisplayProof }
			
			\medskip 
			
			The reverse follows the same pattern.
		
		\item \textbf{Case} $\pnu \pnu[b] P \congr \pnu[b] \pnu P$. Suppose $\Itype{\pnu \pnu[b] P}{K}$. Then the proof has the shape:
		
		\begin{prooftree}
			\footnotesize
			\AxiomC{$\pi$}
			\UnaryInfC{$\Itype[@][@][\Delta', a : T', b : U]{P}{K''}$}
			\UnaryInfC{$\Itype[@][@][\Delta', a : T']{\pnu[b] P}{K''}$}
			\AxiomC{$\phi;\Phi \vdash \Delta \subtype \Delta' ; K'' \le K' ; T \subtype T'$}
			\BinaryInfC{$\Itype[@][@][\Delta, a : T]{\pnu[b] P}{K'}$}
			\UnaryInfC{$\Itype[@][@][\Delta]{\pnu \pnu[b] P}{K'}$}
			\AxiomC{$\phi;\Phi \vdash \Gamma \subtype \Delta ; K' \le K$}
			\BinaryInfC{$\Itype{\pnu \pnu[b] P}{K}$}
		\end{prooftree} 
	
		We can derive the proof:
		
		\begin{prooftree}
			\footnotesize
			\AxiomC{$ \pi'$}
			\UnaryInfC{$\Itype[@][@][\Delta', a : T', b : U]{P}{K''}$}
			\AxiomC{$\phi;\Phi \vdash T \subtype T' $}
			\BinaryInfC{$\Itype[@][@][\Delta', a : T, b : U]{P}{K''}$}
			\UnaryInfC{$\Itype[@][@][\Delta', b : U]{\pnu P}{K''}$}
			\UnaryInfC{$\Itype[@][@][\Delta']{\pnu[b] \pnu P}{K''}$}
			\AxiomC{$\phi;\Phi \vdash \Gamma \subtype \Delta' ; K'' \le K$}
			\BinaryInfC{$\Itype{\pnu[b] \pnu P}{K}$}
		\end{prooftree}
	
		\item \textbf{Case} $\pnu P \parr Q \congr \pnu (P \parr Q)$ with $a$ not free in $Q$. Suppose $\Itype{\pnu P \parr Q}{K}$. Then the proof has the shape:
		
		\begin{prooftree}
			\footnotesize
			\AxiomC{$\pi$}
			\UnaryInfC{$\Itype[@][@][\Delta', a : T]{P}{K''}$}
			\UnaryInfC{$\Itype[@][@][\Delta']{\pnu P}{K''}$}
			\AxiomC{$\phi;\Phi \vdash \Delta \subtype \Delta' ; K'' \le K'$}
			\BinaryInfC{$\Itype[@][@][\Delta]{\pnu P}{K'}$}
			\AxiomC{$\pi'$}
			\UnaryInfC{$\Itype[@][@][\Delta]{Q}{K'}$}
			\BinaryInfC{$\Itype[@][@][\Delta]{\pnu P \parr Q}{K'}$}
			\AxiomC{$\phi;\Phi \vdash \Gamma \subtype \Delta ; K' \le K$}
			\BinaryInfC{$\Itype{\pnu P \parr Q}{K}$}
		\end{prooftree}
	
		By weakening (Lemma~\ref{l:weakening}), we obtain a proof $\pi'_w$ of $\Itype[@][@][\Delta, a : T]{Q}{K'}$. Thus, we have the following derivation:
			
			\medskip 
			
		\begin{prooftree}
			\footnotesize	
			\AxiomC{$\pi$}
			\UnaryInfC{$\Itype[@][@][\Delta', a : T]{P}{K''}$}
			\AxiomC{$\phi;\Phi \vdash \Delta \subtype \Delta' ; T \subtype T ; K'' \le K'$}
			\BinaryInfC{$\Itype[@][@][\Delta, a : T]{P}{K'}$}
			\AxiomC{$\pi'_w$}
			\UnaryInfC{$\Itype[@][@][\Delta, a : T]{Q}{K'}$}
			\BinaryInfC{$\Itype[@][@][\Delta, a : T]{P \parr Q}{K'}$}
			\UnaryInfC{$\Itype[@][@][\Delta]{\pnu (P \parr Q)}{K'}$}
			\AxiomC{$\phi;\Phi \vdash \Gamma \subtype \Delta ; K' \le K$}
			\BinaryInfC{$\Itype{\pnu (P \parr Q)}{K}$}
		\end{prooftree}
		
		For the converse, suppose $\Itype{\pnu (P \parr Q)}{K}$. Then the proof has the shape:
		
		\begin{prooftree}
			\footnotesize	
			\AxiomC{$\pi$}
			\UnaryInfC{$\Itype[@][@][\Delta', a : T']{P}{K''}$}
			\AxiomC{$\pi'$}
			\UnaryInfC{$\Itype[@][@][\Delta', a : T']{Q}{K''}$}
			\BinaryInfC{$\Itype[@][@][\Delta', a : T']{P \parr Q}{K''}$}
			\AxiomC{$\phi;\Phi \vdash \Delta \subtype \Delta' ; T \subtype T' ; K'' \le K'$}
			\BinaryInfC{$\Itype[@][@][\Delta, a : T]{P \parr Q}{K'}$}
			\UnaryInfC{$\Itype[@][@][\Delta]{\pnu (P \parr Q)}{K'}$}
			\AxiomC{$\phi;\Phi \vdash \Gamma \subtype \Delta ; K' \le K$}
			\BinaryInfC{$\Itype{\pnu (P \parr Q)}{K}$}
		\end{prooftree}
	
		Since $a$ is not free in $Q$, by Lemma~\ref{l:strengthening}, from $\pi'$ we obtain a proof $\pi'_c$ of $\Itype[@][@][\Delta']{Q}{K''}$. We can then derive the following typing:
		
		\begin{prooftree}
			\footnotesize
			\AxiomC{$\pi$}
			\UnaryInfC{$\Itype[@][@][\Delta', a : T']{P}{K''}$}
			\AxiomC{$\phi;\Phi \vdash T \subtype T'$}
			\BinaryInfC{$\Itype[@][@][\Delta', a : T]{P}{K''}$}
			\UnaryInfC{$\Itype[@][@][\Delta']{\pnu P}{K''}$}
			\AxiomC{$\pi'_c$}
			\UnaryInfC{$\Itype[@][@][\Delta']{Q}{K''}$}
			\BinaryInfC{$\Itype[@][@][\Delta']{\pnu P \parr Q}{K''}$}
			\AxiomC{$\phi;\Phi \vdash \Gamma \subtype \Delta' ; K'' \le K$}
			\BinaryInfC{$\Itype{\pnu P \parr Q}{K}$}
		\end{prooftree}		
		\end{itemize}
	
		This concludes all the base case. We can then prove Lemma~\ref{l:congrtype} by induction on $P \congr Q$. All the base case have been done, symmetry and transitivity are direct by induction hypothesis. For the cases of contextual congruence, the proof is again straightforward by considering proofs in which there is exaclty one subtyping rule before any syntax-directed rule. 
	\end{proof}
	
	Now that we have Lemma~\ref{l:congrtype}, we can work up to the congruence relation. We now give an exhaustive description of the subtyping relation. 
	
		\begin{lemma}[Exhaustive Description of Subtyping]
			If $\phi; \Phi \vdash T \subtype U$, then one of the following case holds. 
			\begin{itemize}
			\small 
				\item \[T = \Inat \qquad U = \Inat[I'][J'] \qquad \phi;\Phi \vDash I' \le I \qquad \phi;\Phi \vDash J \le J' \]
				\item \[T = \Ilis \qquad U = \Ilis[I'][J'][\IB'] \qquad \phi;\Phi \vDash I' \le I \qquad \phi;\Phi \vDash J \le J' \qquad \phi;\Phi \vdash \IB \subtype \IB' \]
				\item \[ T = \Ibool \qquad U = \Ibool \]
				\item \[T = \Ich \qquad U = \Ich[J][\vect{U}] \qquad \phi;\Phi \vDash I = J \qquad \phi;\Phi \vdash \vect{T} \subtype \vect{U} \qquad \phi;\Phi \vdash \vect{U} \subtype \vect{T} \]
				\item \[T = \Ich \qquad U = \Iich[J][\vect{U}] \qquad \phi;\Phi \vDash I = J \qquad \phi;\Phi \vdash \vect{T} \subtype \vect{U} \]
				\item \[T = \Ich \qquad U = \Ioch[J][\vect{U}] \qquad \phi;\Phi \vDash I = J \qquad \phi;\Phi \vdash \vect{U} \subtype \vect{T} \]
				\item \[T = \Iich \qquad U = \Iich[J][\vect{U}] \qquad \phi;\Phi \vDash I = J \qquad \phi;\Phi \vdash \vect{T} \subtype \vect{U} \]
				\item \[T = \Ioch \qquad U = \Ioch[J][\vect{U}] \qquad \phi;\Phi \vDash I = J \qquad \phi;\Phi \vdash \vect{U} \subtype \vect{T} \]
				\item \[\leftshift{T = \Iserv \qquad U = \Iserv[J][@][K'][\vect{U}] \qquad \phi;\Phi \vDash I = J \qquad (\phi,\vect{i});\Phi \vdash \vect{T} \subtype \vect{U} \qquad (\phi,\vect{i});\Phi \vdash \vect{U} \subtype \vect{T} \qquad (\phi,\vect{i});\Phi \vDash K = K'} \]
				\item \[T = \Iserv \qquad U = \Iiserv[J][@][K'][\vect{U}] \qquad \phi;\Phi \vDash I = J \qquad (\phi,\vect{i});\Phi \vdash \vect{T} \subtype \vect{U} \qquad (\phi,\vect{i});\Phi \vDash K' \le K \]
				\item \[T = \Iserv \qquad U = \Ioserv[J][@][K'][\vect{U}] \qquad \phi;\Phi \vDash I = J \qquad (\phi,\vect{i});\Phi \vdash \vect{U} \subtype \vect{T} \qquad (\phi,\vect{i});\Phi \vDash K \le K' \]
				\item \[T = \Iiserv \qquad U = \Iiserv[J][@][K'][\vect{U}] \qquad \phi;\Phi \vDash I = J \qquad (\phi,\vect{i});\Phi \vdash \vect{T} \subtype \vect{U} \qquad (\phi,\vect{i});\Phi \vDash K' \le K \]
				\item \[T = \Ioserv \qquad U = \Ioserv[J][@][K'][\vect{U}] \qquad \phi;\Phi \vDash I = J \qquad (\phi,\vect{i});\Phi \vdash \vect{U} \subtype \vect{T} \qquad (\phi,\vect{i});\Phi \vDash K \le K' \]			
			\end{itemize}
	\label{l:exhaustivesubtyping}		
		\end{lemma}
	
	\begin{proof}
		The proof is rather straightforward, we proceed by induction on the subtyping relation. All base cases are indeed of this form, and then for transitivity, we can use the induction hypothesis and consider all cases in which the second member of a subtyping relation can match with the first one, and all cases are simple.  
	\end{proof}
	
	 Let us now show Theorem~\ref{t:nonquantsubjectreduction}. We do this by induction on $P \red Q$. Let us first remark that when considering the typing of $P$, the first subtyping rule has no importance since we can always start the typing of $Q$ with the exact same subtyping rule. One can see it in the detailed proof of Lemma~\ref{l:congrtype}.  We now proceed by doing the case analysis on the rules of Figure~\ref{f:reduction}. 
	
	\begin{itemize}
		\item \textbf{Case} $\pserv \parr \pout \red \pserv \parr \psub$. Consider the typing $\Itype{\pserv \parr \pout}{K}$. The first rule is the rule for parallel composition, then the proof is split into the two following subtree:
		
		\begin{prooftree}
			\footnotesize 
	
			\AxiomC{$\Ietype[@][@][\Gamma_1]{a}{\Ioserv[I_1][@][K_1][\vect{T_1}]}$} 
			\AxiomC{$\pi_e$}
			\UnaryInfC{$\Ietype[@][@][{\Idecr[\Gamma_1][I_1]}]{\vect{e}}{\Isub[\vect{T_1}][\vect{i}][\vect{J}]}$}
			\BinaryInfC{$\Itype[@][@][\Gamma_1]{\pout}{I_1 + \Isub[K_1][\vect{i}][\vect{J}]} $}
			\AxiomC{$\phi;\Phi \vdash \Gamma \subtype \Gamma_1 ; I_1 + \Isub[K_1][\vect{i}][\vect{J}] \le K$}
			\BinaryInfC{$\Itype{\pout}{K}$}
		\end{prooftree}
	    
	    \medskip 
	    
	    \leftshiftt{
	    	\footnotesize
			\AxiomC{$\Ietype[@][@][\Gamma_0, \Delta_0]{a}{\Iiserv[I_0][@][K_0][\vect{T_0}]}$}
			\AxiomC{$\pi_P$}
			\UnaryInfC{$\Itype[(\phi,\vect{i})][@][\Gamma', \vect{\var} : \vect{T_0}]{P}{K_0}$}
			\AxiomC{$\Gamma'$ time invariant}
			\AxiomC{$\phi;\Phi \vdash \Idecr[\Gamma_0][I_0] \subtype \Gamma'$}
			\QuaternaryInfC{$\Itype[@][@][\Gamma_0,\Delta_0]{\pserv}{I_0}$}
			\AxiomC{$\phi;\Phi \vdash \Gamma \subtype \Gamma_0,\Delta_0 ; I_0 \le K$}
			\insertBetweenHyps{\hskip -2cm}
			\BinaryInfC{$\Itype{\pserv}{K}$}
			\DisplayProof}
		
		\medskip 
		
		The second subtree can be used exactly in the same way to type the server in the right part of the reduction relation. Furthermore, as the name $a$ is used as an input and as an output, so the original type in $\Gamma$ for this name must be a server type $\Iserv[@][@][K']$. By Lemma~\ref{l:exhaustivesubtyping}, we have:
		\[\phi;\Phi \vDash I_0 = I \qquad (\phi,\vect{i});\Phi \vdash \vect{T} \subtype \vect{T_0} \qquad (\phi,\vect{i});\Phi \vDash K_0 \le K' \]
		\[\phi;\Phi \vDash I = I_1 \qquad (\phi,\vect{i});\Phi \vdash \vect{T_1} \subtype \vect{T} \qquad (\phi,\vect{i});\Phi \vDash K' \le K_1 \]
		So, we obtain directly: 
		\[\phi;\Phi \vDash I_0 = I_1 \qquad (\phi,\vect{i});\Phi \vdash \vect{T_1} \subtype \vect{T_0} \qquad (\phi,\vect{i});\Phi \vDash K_0 \le K_1 \]
		
		Thus, by subtyping, from $\pi_P$ we can obtain a proof of $\Itype[(\phi,\vect{i})][@][{\Idecr[\Gamma_0][I_0], \vect{\var} : \vect{T_1}}]{P}{K_1}$. By Lemma~\ref{l:indexsubstitution}, we have a proof of $\Itype[@][{\Isub[\Phi][\vect{i}][\vect{J}]}][{\Isub[{(\Idecr[\Gamma_0][I_0], \vect{\var} : \vect{T_1})}][\vect{i}][\vect{J}]}]{P}{\Isub[K_1][\vect{i}][\vect{J}]}$. As $\vect{i}$ only appears in $\vect{T_1}$ and $K_1$, we obtain a proof of 
		$\Itype[@][@][{\Idecr[\Gamma_0][I_0], \vect{\var} : \Isub[\vect{T_1}][\vect{i}][\vect{J}]}]{P}{\Isub[K_1][\vect{i}][\vect{J}]}$. 
		
		Now, by Lemma~\ref{l:advanceandsubtyping}, we have: 
		\[ \Idecr = \Gamma_0', \Delta_0' \qquad \phi;\Phi \vdash \Gamma_0' \subtype \Idecr[\Gamma_0] \qquad \Idecr = \Gamma_1', \Delta_1' \qquad \phi;\Phi \vdash \Gamma_1' \subtype \Idecr[\Gamma_1] \]
		By Lemma~\ref{l:weakening}, and as $\phi;\Phi \vdash I = I_0 = I_1$  we can obtain two proofs:
		\[ \Itype[@][@][{\Idecr[\Gamma_0], \Delta_0', \vect{\var} : \Isub[\vect{T_1}][\vect{i}][\vect{J}]}]{P}{\Isub[K_1][\vect{i}][\vect{J}]} \qquad \Ietype[@][@][{\Idecr[\Gamma_1], \Delta_1'}]{\vect{e}}{\Isub[\vect{T_1}][\vect{i}][\vect{J}]} \]
		Finally, by subtyping, with the remark above, we obtain:
		\[ \Itype[@][@][{\Idecr, \vect{\var} : \Isub[\vect{T_1}][\vect{i}][\vect{J}]}]{P}{\Isub[K_1][\vect{i}][\vect{J}]} \qquad \Ietype[@][@][\Idecr]{\vect{e}}{\Isub[\vect{T_1}][\vect{i}][\vect{J}]} \]
		Thus, by the substitution lemma (Lemma~\ref{l:substitution}), we have $\Itype[@][@][\Idecr]{\psub}{\Isub[K_1][\vect{i}][\vect{J}]}$.
		Then, by delaying (Lemma~\ref{l:delaying}), we have $\Itype[@][@][\Iincr{(\Idecr)}]{\psub}{I + \Isub[K_1][\vect{i}][\vect{J}]}$, and $\Gamma = \Gamma', \Delta$ with $\phi;\Phi \vdash \Gamma' \subtype \Iincr{(\Idecr)}$. Recall that $\phi;\Phi \vDash I + \Isub[K_1][\vect{i}][\vect{J}] \le K$. Thus, again by subtyping and weakening, we obtain 
		\[ \Itype{\psub}{K} \]
		And this concludes this case. 
		
		\item \textbf{Case} $\pin \parr \pout \red \psub$. Consider the typing $\Itype{\pserv \parr \pout}{K}$. The first rule is the rule for parallel composition, then the proof is split into the two following subtree:
		
		\begin{prooftree}
			\footnotesize 		
			\AxiomC{$\Ietype[@][@][\Gamma_1]{a}{\Ioch[I_1][\vect{T_1}]}$} 
			\AxiomC{$\pi_e$}
			\UnaryInfC{$\Ietype[@][@][{\Idecr[\Gamma_1][I_1]}]{\vect{e}}{\vect{T_1}}$}
			\BinaryInfC{$\Itype[@][@][\Gamma_1]{\pout}{I_1} $}
			\AxiomC{$\phi;\Phi \vdash \Gamma \subtype \Gamma_1 ; I_1 \le K$}
			\BinaryInfC{$\Itype{\pout}{K}$}
		\end{prooftree}
		
		\begin{prooftree}
			\footnotesize
			\AxiomC{$\Ietype[@][@][\Gamma_0]{a}{\Iich[I_0][\vect{T_0}]}$}
			\AxiomC{$\pi_P$}
			\UnaryInfC{$\Itype[@][@][{\Idecr[\Gamma_0][I_0]}, \vect{\var} : \vect{T_0}]{P}{K_0}$}
			\BinaryInfC{$\Itype[@][@][\Gamma_0]{\pin}{I_0 + K_0}$}
			\AxiomC{$\phi;\Phi \vdash \Gamma \subtype \Gamma_0 ; I_0 + K_0 \le K$}
			\BinaryInfC{$\Itype{\pin}{K}$}
		 \end{prooftree}
	 
	 	As the name $a$ is used as an input and as an output, so the original type in $\Gamma$ for this name must be a channel type $\Ich$. By Lemma~\ref{l:exhaustivesubtyping}, we have:
	 	\[\phi;\Phi \vDash I_0 = I \qquad \phi;\Phi \vdash \vect{T} \subtype \vect{T_0} \qquad \phi;\Phi \vDash I = I_1 \qquad \phi;\Phi \vdash \vect{T_1} \subtype \vect{T}\]
	 	So, we obtain directly: 
	 	\[\phi;\Phi \vDash I_0 = I_1 \qquad (\phi,\vect{i});\Phi \vdash \vect{T_1} \subtype \vect{T_0} \]
	 	
	 	Thus, by subtyping, from $\pi_P$ we can obtain a proof of $\Itype[@][@][{\Idecr[\Gamma_0][I_0], \vect{\var} : \vect{T_1}}]{P}{K_0}$.
	 	Now, by Lemma~\ref{l:advanceandsubtyping}, we have: 
	 	\[ \Idecr = \Gamma_0', \Delta_0' \qquad \phi;\Phi \vdash \Gamma_0' \subtype \Idecr[\Gamma_0] \qquad \Idecr = \Gamma_1', \Delta_1' \qquad \phi;\Phi \vdash \Gamma_1' \subtype \Idecr[\Gamma_1] \]
	 	By Lemma~\ref{l:weakening}, and as $\phi;\Phi \vdash I = I_0 = I_1$  we can obtain two proofs:
	 	\[ \Itype[@][@][{\Idecr[\Gamma_0], \Delta_0', \vect{\var} : \vect{T_1}}]{P}{K_0} \qquad \Ietype[@][@][{\Idecr[\Gamma_1], \Delta_1'}]{\vect{e}}{\vect{T_1}} \]
	 	Then, by subtyping, with the remark above, we obtain:
	 	\[ \Itype[@][@][\Idecr, \vect{\var} : \vect{T_1}]{P}{K_0} \qquad \Ietype[@][@][\Idecr]{\vect{e}}{\vect{T_1}} \]
	 	Thus, by the substitution lemma (Lemma~\ref{l:substitution}), we have $\Itype[@][@][\Idecr]{\psub}{K_0}$.
	 	Then, by delaying (Lemma~\ref{l:delaying}), we have $\Itype[@][@][\Iincr{(\Idecr)}]{\psub}{I + K_0}$, and $\Gamma = \Gamma', \Delta$ with $\phi;\Phi \vdash \Gamma' \subtype \Iincr{(\Idecr)}$. Recall that $\phi;\Phi \vDash I + K_0 \le K$. Thus, again by subtyping and weakening, we obtain 
	 	\[ \Itype{\psub}{K} \]
	 	And this concludes this case. 
	 	
	 	\item \textbf{Case} $\pifl[\nil] \red P$. This case is similar to its counterpart for natural number and the two case for booleans, so we only detail this one. Suppose given a derivation $\Itype{\pifl[\nil]}{K}$. Then the derivation has the shape:
	 	\begin{prooftree}
	 		\footnotesize
	 		\AxiomC{}
	 		\UnaryInfC{$\Ietype[@][@][\Delta]{\nil}{\Ilis[0][0][\IB']}$}
	 		\AxiomC{$\phi;\Phi \vdash \Gamma \subtype \Delta; \Ilis[0][0][\IB'] \subtype \Ilis$}
	 		\BinaryInfC{$\Ietype{\nil}{\Ilis}$}
	 		\AxiomC{$\pi_P$}
	 		\UnaryInfC{$\Itype[@][(\Phi,I \le 0)][@]{P}{K}$}
	 		\AxiomC{$\pi_Q$}
	 		\TrinaryInfC{$\Itype{\pifl[\nil]}{K}$}
	 	\end{prooftree}
		Where $\pi_Q$ is the typing for $Q$ that does not interest us in this case. By Lemma~\ref{l:exhaustivesubtyping}, we obtain:
		\[ \phi;\Phi \vDash I \le 0 \qquad \phi;\Phi \vDash 0 \le J \qquad \phi;\Phi \vdash \IB' \subtype \IB \]
		As $\phi;\Phi \vDash I \le 0$, by Lemma~\ref{l:strengthening}, we obtain directly from $\pi_P$ a proof $\Itype[@][@][@]{P}{K}$. 
		
		\item \textbf{Case} $\pifl[e::e'] \red \psub[Q][x,y][e,e']$. This case is more difficult than its counterpart for integers, thus we only detail this case and the one for integers can easily be deduced from this one. Suppose given a derivation $\Itype{\pifl[e::e']}{K}$. Then the proof has the shape:
		
		\medskip 
		
		\leftshiftt{
			\footnotesize
			\AxiomC{$\pi_e$}
			\UnaryInfC{$\Ietype[@][@][\Delta]{e}{\IB'}$}
			\AxiomC{$\pi_{e'}$}
			\UnaryInfC{$\Ietype[@][@][\Delta]{e'}{\Ilis[I'][J'][\IB']}$}
			\BinaryInfC{$\Ietype[@][@][\Delta]{e::e'}{\Ilis[I'+1][J'+1][\IB']}$}
			\AxiomC{$\phi;\Phi \vdash \Gamma \subtype \Delta; \Ilis[I'+1][J'+1][\IB'] \subtype \Ilis$}
			\BinaryInfC{$\Ietype{e::e'}{\Ilis}$}
			\AxiomC{$\pi_P$}
			\AxiomC{$\pi_Q$}
			\TrinaryInfC{$\Itype{\pifl[e::e']}{K}$}
			\DisplayProof}
		
		\medskip 
		
		Where $\pi_Q$ is a proof of $\Itype[@][(\Phi,J \ge 1)][\Gamma, x : \IB, y : {\Ilis[I-1][J-1][\IB]}]{Q}{K}$, and $\pi_P$ is a typing derivation for $P$ that does not interest us in this case.
		
		Lemma~\ref{l:exhaustivesubtyping} gives us the following information:
		\[ \phi;\Phi \vDash I \le I' + 1 \qquad \phi;\Phi \vDash J' + 1 \le J \qquad \phi;\Phi \vdash \IB' \subtype \IB \] 
		From this, we can deduce the following constraints:
		\[ \phi;\Phi \vDash J \ge 1 \qquad \phi;\Phi \vDash I - 1 \le I' \qquad \phi;\Phi \vDash J' \le J-1 \]
		Thus, with the subtyping rule and the proofs $\pi_e$ and $\pi_{e'}$ we obtain:
		\[ \Ietype{e}{\IB} \qquad \Ietype{e'}{\Ilis[I-1][J-1][\IB]} \]
		Then, by Lemma~\ref{l:strengthening}, from $\pi_Q$ we obtain a proof of $\Itype[@][@][\Gamma, x : \IB, y : {\Ilis[I-1][J-1][\IB]}]{Q}{K}$. By
		the substitution lemma (Lemma~\ref{l:substitution}), we obtain $\Itype[@][@][@]{\psub[Q][x,y][e,e']}{K}$. This concludes this case. 
		
		\item \textbf{Case} $P \parr R \red Q \parr R$ with $P \red Q$. Suppose that $\Itype{P \parr R}{K}$. Then the proof has the shape:
		\begin{prooftree}
			\footnotesize
			\AxiomC{$\pi_P$}
			\UnaryInfC{$\Itype{P}{K}$}
			\AxiomC{$\pi_R$}
			\UnaryInfC{$\Itype{R}{K}$}
			\BinaryInfC{$\Itype{P \parr R}{K}$}
		\end{prooftree}
		By induction hypothesis, with the proof $\pi_P$ of $\Itype{P}{K}$, we obtain a proof $\pi_Q$ of $\Itype{P}{Q}$. Then, we can derive the following proof:
		\begin{prooftree}
			\footnotesize
			\AxiomC{$\pi_Q$}
			\UnaryInfC{$\Itype{Q}{K}$}
			\AxiomC{$\pi_R$}
			\UnaryInfC{$\Itype{R}{K}$}
			\BinaryInfC{$\Itype{Q \parr R}{K}$}
		\end{prooftree}
		This concludes this case. 
		
		\item \textbf{Case} $\pnu P \red \pnu Q$ with $P \red Q$. Suppose that $\Itype{\pnu P}{K}$. Then the proof has the shape:
		\begin{prooftree}
			\footnotesize
			\AxiomC{$\pi_P$}
			\UnaryInfC{$\Itype[@][@][\Gamma, a : T]{P}{K}$}
			\UnaryInfC{$\Itype{\pnu P}{K}$}
		\end{prooftree}
		By induction hypothesis, with the proof $\pi_P$ of $\Itype[@][@][\Gamma, a : T]{P}{K}$, we obtain a proof $\pi_Q$ of $\Itype[@][@][\Gamma, a : T]{Q}{K}$
		We can then derive the proof:
		\begin{prooftree}
			\footnotesize
			\AxiomC{$\pi_Q$}
			\UnaryInfC{$\Itype[@][@][\Gamma, a : T]{Q}{K}$}
			\UnaryInfC{$\Itype{\pnu Q}{K}$}
		\end{prooftree} 
		This concludes this case. 
		
		\item \textbf{Case} $P \red Q$ with $P \congr P'$, $P' \red Q'$ and $Q \congr Q'$. Suppose that $\Itype{P}{K}$. By Lemma~\ref{l:congrtype}, we have $\Itype{P'}{K}$. By induction hypothesis, we obtain $\Itype{Q'}{K}$. Then, again by Lemma~\ref{l:congrtype}, we have $\Itype{Q}{K}$. This concludes this case. 
	\end{itemize}
	
	This concludes the proof of Theorem~\ref{t:nonquantsubjectreduction}. 
	
	\subsection{Quantitative Subject Reduction}
	\label{ss:quantitativesubjectreduction}
	
	We now want to prove that our type system can effectively give a bound on the number of time reduction. However, the subject reduction for time reduction does not hold as expected, in fact our type system relies crucially on the tick-last strategy. To see where the problem is, let us consider the following process:
	\[P = \pin[@][][\tick.\pzero] \parr \pout[@][] \parr \tick.\pzero \] 
	In an unrestricted setting, this process could need two time reductions to reach a normal form. 
	\begin{equation}
	P \tred (\pin[@][][\tick.\pzero] \parr \pout[@][] \parr \pzero) \red (\tick.\pzero \parr \pzero) \tred (\pzero \parr \pzero)
	\end{equation} 
	However, with the tick-last strategy, we obtain:
	\begin{equation}
	P \red (\tick.\pzero \parr \tick. \pzero) \tred (\pzero \parr \pzero)
	\end{equation}
	And this corresponds to a reduction with ''maximal parallelism'', as we considered the tick to be the costly operation. 
	As we wanted, our type system can give this process a complexity $1$, with for example the following typing:
	
	\medskip 
	
	\leftshiftt{
		\footnotesize
		\AxiomC{}
		\doubleLine 
		\UnaryInfC{$\Ietype[\cdot][\cdot][{a : \Ich[0][]}]{a}{\Iich[0][]}$}
		\AxiomC{}
		\UnaryInfC{$\Itype[\cdot][\cdot][\cdot]{\pzero}{0}$}
		\UnaryInfC{$\Itype[\cdot][\cdot][{a : \Ich[0][]}]{\tick.0}{1}$}
		\BinaryInfC{$\Itype[\cdot][\cdot][{a : \Ich[0][]}]{\pin[@][][\tick.\pzero]}{1}$}
		\AxiomC{}
		\UnaryInfC{$\Itype[\cdot][\cdot][{a : \Ioch[0][]}]{\pout[@][]}{0}$}
		\AxiomC{$\cdot;\cdot \vdash \Ich[0][] \subtype \Ioch[0][] ; 0 \le 1$}
		\BinaryInfC{$\Itype[\cdot][\cdot][{a : \Ich[0][]}]{\pout[@][]}{1}$}
		\AxiomC{}
		\UnaryInfC{$\Itype[\cdot][\cdot][\cdot]{\pzero}{0}$}
		\UnaryInfC{$\Itype[\cdot][\cdot][{a : \Ich[0][]}]{\tick.0}{1}$}
		\TrinaryInfC{$\Itype[\cdot][\cdot][{a : \Ich[0][]}]{P}{1}$}
		\DisplayProof}
	
	\medskip 
	
	As a consequence, this typing for $P$ does not give a bound on the number of time reduction in $(1)$. Intuitively, this is because the typing $a : \Ich[0][0]$ announces that $a$ will do its communication at time $0$, whereas the reduction $(1)$ does this reduction at time $1$. However, in the tick-last strategy, all communications are made as early as possible. As a consequence, a name of type $\Ich[0]$ will not do any communication at a time greater than $0$. So, what we will show is that if $\Itype{P}{K}$ and $P \tred P'$, then, there is some $\Gamma',K',P''$ with $\Itype[@][@][\Gamma']{P''}{K'}$, such that $\Gamma'$ is close to $\Idecr[@][1]$, $K' + 1 \le K$ and $P''$ can simulate $P'$. 
	
	For this simulation of $P'$, we first work with the type system of Figure~\ref{f:typeexpression} and Figure~\ref{f:typeprocess}, to present general result for input/output types and not specific to our type system. However, we still have the tick with its associated typing rules presented in the beginning of Section~\ref{s:sizetypes}. Then, using this generic definition, we will present something linked with our type system, and we will directly obtain that we have a simulation. 
	
	\begin{definition}[Discarding Deadlocked Processes]
		Let $R$ be a process and $\IS$ a multiset of guarded processes included in the top guarded processes of $P$ (see the proof of Lemma~\ref{l:uniquecanonical} for a formal definition) such that:
		\begin{itemize}
			\item $R \congr \pnu[\vect{a}] (P_1 \parr \cdots \parr P_n \parr Q_1 \parr \cdots \parr Q_m)$ with $P_1,\dots,P_n,Q_1,\dots,Q_m$ guarded processes.
			\item $[P_1 \dots P_n] = \IS$.
			\item $\IS$ contains only non-replicated input and output processes. 
			\item $P_1 \parr \cdots \parr P_n$ is in normal form for $\red$. 
			\item For each name $a$ at a top of a process in $\IS$, $a$ cannot appear both at the top of an input process and an output process in $\IS$. 
			\item There exist $\Gamma$ such that $\Gamma \vdash Q_1 \parr \cdots \parr Q_m$ and for each name $a$ at a top of a process in $\IS$, either $a$ appears only at the top of input processes in $\IS$, and $a : \ich \in \Gamma$ for some $\vect{T}$ or $a$ appears only at the top of output processes in $\IS$, and $a : \och \in \Gamma$ for some $\vect{T}$.    
		\end{itemize}
		Then, we say that $P'$ is a \emph{$\IS$-discarding} of $P$, noted $\Idisc{P,P'}$ if $P' \congr \pnu[\vect{a}] (Q_1 \parr \cdots \parr Q_n)$. 
	\label{d:discarddeadlock}
	\end{definition}
	
	With this generic notion, we can show that we define a simulation of $R$ in the following sense.
	
	\begin{lemma}
		If $R_0 \red R_0'$ and $\Idisc{R_0,R_1}$, then there exists $R_1'$ such that $\Idisc{R_0',R_1'}$ and $R_1 \red R_1'$. We also have the same simulation for $\tred$. Moreover, if $R_0 \tred R_0$ and $\Idisc{R_0,R_1}$ then $R_1 \tred R_1$.     
	\label{l:discardsimulation}
	\end{lemma} 
	 
	\begin{proof}
		First, with Lemma~\ref{l:exhaustivereductions}, we can do a case analysis on the reduction $R_0 \red R_0'$.
		\begin{itemize}
			\item If we are in the case: 
			\[ R_0 \congr \pnu[\vect{b}] (Q_1 \parr \cdots \parr Q_n \parr \pserv \parr \pout) \qquad R_0' \congr \pnu[\vect{b}] (Q_1 \parr \cdots \parr Q_n \parr \pserv \parr \psub)  \]
			By definition of $\Idisc{R_0,R_1}$, $\pserv \notin \IS$ because $\IS$ cannot have replicated input. Moreover, $\pout \notin \IS$ since we cannot have $\Gamma \vdash \pserv$ if $a : \och \in \Gamma$. So, we can write 
			\[ R_0 \congr \pnu[\vect{b}] (\IS \parr Q_1' \parr \cdots \parr Q_m' \parr \pserv \parr \pout) \qquad R_0' \congr \pnu[\vect{b}] (\IS \parr Q_1' \parr \cdots \parr Q_m' \parr \pserv \parr \psub)  \]
			Now, we want to show that the guarded processes of $R_0'$ without $\IS$ are typable. By definition of $\Idisc{R_0,R_1}$, $\Gamma \vdash Q_1' \parr \cdots \parr Q_m' \parr \pserv \parr \pout$ with the good restrictions on $\Gamma$. By subject reduction of input/output types \cite{SangiorgiWalkerPi}, we obtain directly $\Gamma \vdash Q_1' \parr \cdots \parr Q_m' \parr \pserv \parr \psub$. So, if we pose $R_1' = \pnu[\vect{b}] (Q_1' \parr \cdots \parr Q_m' \parr \pserv \parr \psub)$, we have $\Idisc{R_0',R_1'}$ and $R_1 \red R_1'$ because:
			\[ R_1 \congr \pnu[\vect{b}] (Q_1' \parr \cdots \parr Q_m' \parr \pserv \parr \pout) \qquad R_1' = \pnu[\vect{b}] (Q_1' \parr \cdots \parr Q_m' \parr \pserv \parr \psub)  \]
			 
			\item If we are in the case: 
			\[ R_0 \congr \pnu[\vect{b}] (Q_1 \parr \cdots \parr Q_n \parr \pin \parr \pout) \qquad R_0' \congr \pnu[\vect{b}] (Q_1 \parr \cdots \parr Q_n \parr \pin \parr \psub)  \]
			By definition of $\Idisc{R_0,R_1}$, either $\pin \notin \IS$ or $\pout \notin \IS$ because $\IS$ cannot have a name both at a top of an input and output process. Suppose for example $\pin \notin \IS$. Then, $\pout \notin \IS$ since we cannot have $\Gamma \vdash \pin$ if $a : \och \in \Gamma$. Symmetrically, in $\pout \notin \IS$ then $\pin \notin \IS$. So, we can write 
			\[ R_0 \congr \pnu[\vect{b}] (\IS \parr Q_1' \parr \cdots \parr Q_m' \parr \pin \parr \pout) \qquad R_0' \congr \pnu[\vect{b}] (\IS \parr Q_1' \parr \cdots \parr Q_m' \parr \pin \parr \psub)  \]
			Then, we can conclude this proof as in the previous case. 
			\item  If we are in the case:
			\[ R_0 \congr \pnu[\vect{b}] (Q_1 \parr \cdots \parr Q_n \parr \pifn[\zero]) \qquad R_0' \congr \pnu[\vect{b}] (Q_1 \parr \cdots \parr Q_n \parr P)  \]
			Then by definition, $\pifn[\zero] \notin \IS$, so we go back to the previous cases where $\IS$ is only in the $R_i$. All the other conditionals behaves the same way.
		\end{itemize}
		Then, with Lemma~\ref{l:exhaustivetimereductions}, we know the shape of a time reduction, and as processes starting with a tick cannot be in $\IS$, we obtain the result in the same way as conditionals for $\red$. For this, remark that the subject reduction of input/output types for $\tred$ is straightforward by definition of the tick rule. Finally, we can prove that if $R_0 \tred R_0$ and $\Idisc{R_0,R_1}$ then $R_1 \tred R_1$ by remarking that $P \tred P$ if and only if in the top guarded processes of the canonical form of $P$, none of them start with a tick. 
	\end{proof}
	
	We also show that discarding preserves normal form. 
	
	\begin{lemma}
		If $P$ is in normal form for $\red$ and $\Idisc{P,P'}$, then $P'$ is in normal form for $\red$.
	\label{l:discardnormalform}
	\end{lemma}
	
	\begin{proof}
		With Lemma~\ref{l:exhaustivereductions}, one can see that if $\Idisc{P}$ is not in normal form for $\red$, then $P$ cannot be in normal form, thus we obtain directly the propriety.
	\end{proof}
	
	As a consequence of Lemma~\ref{l:discardsimulation} and Lemma~\ref{l:discardnormalform}, if a process $P_0$ can be reduced to $Q_0$ by the strategy of Definition~\ref{d:reductionstrategy}, then, if we have $\Idisc{P_0,P_1}$, $P_1$ can also be reduced to $Q_1$ with $\Idisc{Q_0,Q_1}$ by the same strategy. Indeed, by Lemma~\ref{l:discardsimulation}, when $P_0 \red^* P_0'$ then $P_1 \red^* P_1'$ with $\Idisc{P_0',P_1'}$, and if $P_0'$ is in normal form, so is $P_1'$ by Lemma~\ref{l:discardnormalform}. And then, if $P_0' \tred P_0'$ and the computation stops, then $P_1' \tred P_1'$ and so the computation stops. Otherwise, if $P_0' \tred Q_0'$, then $P_1' \tred Q_1'$ with $\Idisc{Q_0',Q_1'}$, and we can continue the simulation. 
	
	Now, we want to use this with a notion of discarding linked with our type system. 
	
	\begin{definition}[Discarding Time Out Processes]
		Given a process $P$ in normal form for $\red$ with a proof $\pi$ of $\Itype[@][@][\cdot]{P}{K}$, by Lemma~{\ref{l:closedtypednormalform2}}, its canonical form is: 
		\[\leftshift{P \congr \pnu[\vect{a}](\pserv[b_1][\vect{\var_1^0}][P_1] \parr \cdots \parr \pserv[b_k][\vect{\var_k^0}][P_k] \parr \pin[c_1][\vect{\var_1^1}][Q_1] \parr \cdots \parr \pin[c_{k'}][\vect{\var_{k'}^1}][Q_{k'}] \parr \pout[d_1][\vect{e_1}] \parr \cdots \parr \pout[d_{k''}][\vect{e_{k''}}] \parr \tick.R_1 \parr \cdots \parr \tick.R_{k'''})} \]
		with $(\{ b_i \mid 1 \le i \le k \} \cup \{ c_i \mid 1 \le i \le k' \}) \cap \{ d_i \mid 1 \le i \le k\} = \emptyset$.
		In the proof $\pi$, a type is given to each name in $\vect{a}$ with the rule for the $\pnu$ constructor. Note that during the proof, this type is not fixed because of subtyping, but its time is, as the time of a type is invariant by subtyping. 
		We define the multiset of \emph{timed out processes} of $P$ according to $\pi$, noted $\IT[\pi](P)$, included in the top guarded processes of $P$ by the following rules:
		\begin{itemize}
			\item A server $\pserv[b_i][\vect{\var_i^0}][P_i]$ is never in $\IT[\pi](P)$. 
			\item A ticked process $\tick.R_i$ is never in $\IT[\pi](P)$. 
			\item An input $\pin[c_i][\vect{\var_i^1}][Q_i]$ is in $\IT[\pi](P)$ if and only if the time $I$ of the associated type of $c_i$ is such that $\phi;\Phi \not\vDash I \ge 1$
			\item An output $\pout[d_i][\vect{e_i}]$ is in $\IT[\pi](P)$ if and only if the time $I$ of the associated type of $d_i$ is such that $\phi;\Phi \not\vDash I \ge 1$ 
		\end{itemize}
		\label{d:discardtime}
	\end{definition} 
	
	When the proof $\pi$ is not ambiguous, we use $\IT(P)$. Now what we want to show is that if $\Idisc[\IT(P)]{P,P'}$, then $P'$ has exactly the same behaviour as $P$ for the reductions. In order to do this, we show that $\IT(P)$ is a special case of Definition~\ref{d:discarddeadlock}.
	
	\begin{lemma}[Time Out Processes and Discarding]
		Let $P$ be a process in normal form for $\red$ with a proof $\pi$ of $\Itype[@][@][\cdot]{P}{K}$. Then, $\IT(P)$ satisfies the condition of Definition~\ref{d:discarddeadlock}.
	\label{l:timeoutisdeadlock}  
	\end{lemma}
	
	\begin{proof}
		First, the fact that $\IT(P)$ is indeed a multiset of guarded processes included in the top guarded processes of $P$ is direct. Moreover, by definition, $\IT(P)$ contains only non-replicated input and output processes, and it is indeed in normal form for $\red$, as $P$ is in normal form. Now, let us show the two remaining points.
		
		Let $a$ be a name at a top of a process in $\IT(P)$. Suppose that $a$ appears both at the top of an input process $\pin$ and an output process $\pout$ in $\IT(P)$. As $P$ is typable with sized types, it is also typable without sizes. So, by the usual result on input/output type, $\vect{v}$ and $\vect{e}$ have the same arity, and base type variables are matched with base type expressions, and channel variable are matched with other channel variables, thus the reduction $\pin \parr \pout \red \psub$ is defined, which contradicts the fact that $P$ is in normal form. 
		
		Finally, let us consider the canonical form of $P$: 
		\[\leftshift{P \congr \pnu[\vect{a}](\pserv[b_1][\vect{\var_1^0}][P_1] \parr \cdots \parr \pserv[b_k][\vect{\var_k^0}][P_k] \parr \pin[c_1][\vect{\var_1^1}][Q_1] \parr \cdots \parr \pin[c_{k'}][\vect{\var_{k'}^1}][Q_{k'}] \parr \pout[d_1][\vect{e_1}] \parr \cdots \parr \pout[d_{k''}][\vect{e_{k''}}] \parr \tick.R_1 \parr \cdots \parr \tick.R_{k'''})} \]
		with $(\{ b_i \mid 1 \le i \le k \} \cup \{ c_i \mid 1 \le i \le k' \}) \cap \{ d_i \mid 1 \le i \le k\} = \emptyset$.
		By Lemma~\ref{l:congrtype}, the typing $\pi$ of $P$ gives us a typing of this canonical form. Moreover, in this typing, the name are given the same type as in the original typing $\pi$ (see the proof of Lemma~\ref{l:congrtype}). If we look at the shape of the proof for the canonical form, it starts with rules for $\nu$ and subtyping rules. And then it uses the rule of parallel composition and again subtyping rules to type each of the guarded processes in the canonical form of $P$.
		
		 Let us first show that we can always push the subtyping rule in the typing of guarded processes in this case. For this, we show that the rule for subtyping and $\nu$ can be swapped, and the same for parallel composition. 
		 \begin{itemize}
		 	\item If we have the typing:
		 	\begin{prooftree}
		 		\footnotesize 
		 		\AxiomC{$\pi$}
		 		\UnaryInfC{$\Itype[@][@][\Delta, a : T]{P}{K'}$}
		 		\UnaryInfC{$\Itype[@][@][\Delta]{\pnu P}{K'}$}
		 		\AxiomC{$\phi;\Phi \vDash \Gamma \subtype \Delta$}
		 		\AxiomC{$\phi;\Phi \vDash K' \le K $}
		 		\TrinaryInfC{$\Itype{\pnu P}{K}$}
		 	\end{prooftree}
	 		Then, we can push the subtyping rule with the following derivation:
	 		\begin{prooftree}
	 			\footnotesize 
	 			\AxiomC{$\pi$}
	 			\UnaryInfC{$\Itype[@][@][\Delta, a : T]{P}{K'}$}
	 			\AxiomC{$\phi;\Phi \vDash \Gamma, a : T \subtype \Delta, a : T$}
	 			\AxiomC{$\phi;\Phi \vDash K' \le K $}
	 			\TrinaryInfC{$\Itype[@][@][\Gamma, a : T]{P}{K}$}
	 			\UnaryInfC{$\Itype{\pnu P}{K}$}
	 		\end{prooftree}
	 		\item If we have the typing: 
	 		\begin{prooftree}
	 			\footnotesize 
	 			\AxiomC{$\pi_P$}
	 			\UnaryInfC{$\Itype[@][@][\Delta]{P}{K'}$}
	 			\AxiomC{$\pi_Q$}
	 			\UnaryInfC{$\Itype[@][@][\Delta]{Q}{K'}$}
	 			\BinaryInfC{$\Itype[@][@][\Delta]{P \parr Q}{K'}$}
	 			\AxiomC{$\phi;\Phi \vDash \Gamma \subtype \Delta$}
	 			\AxiomC{$\phi;\Phi \vDash K' \le K $}
	 			\TrinaryInfC{$\Itype{P \parr Q}{K}$}
	 		\end{prooftree}
	 		Then, we can push the subtyping rule with the following derivation:
	 		 \begin{prooftree}
	 			\footnotesize 
	 			\AxiomC{$\pi_P$}
	 			\UnaryInfC{$\Itype[@][@][\Delta]{P}{K'}$}
	 			\AxiomC{$\phi;\Phi \vDash \Gamma \subtype \Delta$}
	 			\AxiomC{$\phi;\Phi \vDash K' \le K $}
	 			\TrinaryInfC{$\Itype{P}{K}$}
	 			\AxiomC{$\pi_Q$}
	 			\UnaryInfC{$\Itype[@][@][\Delta]{Q}{K'}$}
	 			\AxiomC{$\phi;\Phi \vDash \Gamma \subtype \Delta$}
	 			\AxiomC{$\phi;\Phi \vDash K' \le K $}
	 			\TrinaryInfC{$\Itype{Q}{K}$}
	 			\BinaryInfC{$\Itype{P \parr Q}{K}$}
	 		\end{prooftree}
		 \end{itemize}
	 	As a consequence, by pushing the subtyping rule just before the rule for the typing of guarded processes, we obtain the following derivation for the typing of the canonical form:
	 	\begin{prooftree}
	 		\footnotesize
	 		\AxiomC{$\pi_1,\dots,\pi_n$}
	 		\UnaryInfC{$\forall i,~ \Itype[@][@][\vect{a} : \vect{U}]{S_i}{K}$}
	 		\doubleLine 
	 		\UnaryInfC{$\Itype[@][@][\vect{a} : \vect{U}]{(S_1 \parr \cdots \parr S_n)}{K}$}
	 		\doubleLine 
	 		\UnaryInfC{$\Itype[@][@][\cdot]{\pnu[\vect{a}]}{(S_1 \parr \cdots \parr S_n)}{K}$}
	 	\end{prooftree}
	 	Now, we want to show that there exists $\Gamma'$ (without size types) such that for all guarded process \hbox{$S_i \notin \IT(P)$}, $\Gamma' \vdash S_i$, with $\Gamma'$ giving output type to name in output processes in $\IT(P)$ and input type to name in input processes in $\IT(P)$. We define $\Gamma' = \vect{a} : \vect{U'}$ with:
	 	\begin{itemize}
	 		\item If $U_i = \Ich$ with $\phi;\Phi \vDash I \ge 1$, then $U_i' = \ch[\IU(\vect{T})]$. (See Definition~\ref{d:forgetful} for $\IU$).
	 		\item If $U_i = \Ich$ with $\phi;\Phi \not\vDash I \ge 1$, then if $a_i$ is a top name of a process in $\IT(P)$, we pose $U_i' = \ich[\IU(\vect{T})]$ if it appears only at the top of input processes in $\IT(P)$, and $U_i' = \och[\IU(\vect{T})]$ if it only appears at the top of output processes in $\IT(P)$.  Otherwise, we pose $U_i' = \ch[\IU(\vect{T})]$  
	 		\item If $U_i = \Iich$, then we pose $U_i' = \ich[\IU(\vect{T})]$
	 		\item If $U_i = \Ioch$, then we pose $U_i' = \och[\IU(\vect{T})]$
	 		\item If $U_i = \Iserv$ with $\phi;\Phi \vDash I \ge 1$, then $U_i' = \ch[\IU(\vect{T})]$.
	 		\item If $U_i = \Iserv$ with $\phi;\Phi \not\vDash I \ge 1$, then if $a_i$ is a top name of a process in $\IT(P)$, we pose $U_i' = \och[\IU(\vect{T})]$. (Note that as $\IT(P)$ do not contain replicated input, $a$ can only appear as an output in $\IT(P)$) Otherwise, we pose $U_i' = \ch[\IU(\vect{T})]$.  
	 		\item If $U_i = \Iiserv$, then we pose $U_i' = \ich[\IU(\vect{T})]$.
	 		\item If $U_i = \Ioserv$, then we pose $U_i' = \och[\IU(\vect{T})]$.
	 	\end{itemize} 
	 	$\Gamma'$ satisfies the restriction of Definition~\ref{d:discarddeadlock}. Remark that the only differences between $\Gamma'$ and $\IU(\vect{a} : \vect{U})$ is for names in $\IT(P)$ with types that were originally both input and output. Now we need to show that under the context $\Gamma'$, all the top guarded processes of $P$ not in $\IT(P)$ can be typed. Let us proceed by case analysis. In order to simplify the notation, we use usual generic notation for the guarded processes instead of the notation specified above.
	 	
	 	\begin{itemize}
	 		\item If $\pserv$ is a top guarded processes of $P$. By definition, $\pserv \notin \IT(P)$. Let us consider the typing of $\pserv$. 
	 		
	 		\medskip 
	 		
	 		\leftshiftt{
	 			\footnotesize 
	 			\AxiomC{$\Ietype[@][@][\Gamma,\Delta]{a}{\Iiserv[@][@][K']}$}
	 			\AxiomC{$\Itype[(\phi,\vect{i})][@][\Gamma'', \vect{v} : \vect{T}]{P}{K'}$}
	 			\AxiomC{$\Gamma''$ time invariant}
	 			\AxiomC{$\phi;\Phi \vdash \Idecr[@][@][\phi;\Phi] \subtype \Gamma''$}
	 			\QuaternaryInfC{$\Itype[@][@][\Gamma, \Delta]{\pserv}{I}$}
	 			\AxiomC{$\phi;\Phi \vdash (\vect{a} : \vect{U}) \subtype \Gamma,\Delta; I \le K$}
	 			\insertBetweenHyps{\hskip -2cm}
	 			\BinaryInfC{$\Itype[@][@][\vect{a} : \vect{U}]{\pserv}{K}$}
	 			\DisplayProof}
	 		
	 		\medskip
	 		
	 		As $\Gamma'$ is time invariant, there is no channel type in $\Gamma''$, and all server type must have the shape $\Ioserv[0]$ for some $\vect{i},K,\vect{T}$.
	 		By Lemma~\ref{l:forgetfultyping}, we have:
	 		\[ \leftshift{\IU(\Gamma), \IU(\Delta) \vdash a : \ich[\IU(\vect{T})] \qquad \IU(\Gamma''), \vect{v} : \IU(\vect{T}) \vdash P \qquad \IU(\Gamma) = \Gamma_0, \Gamma_1 \text{ with } \Gamma_0 \subtype \IU(\Gamma'') \qquad \vect{a} : \vect{\IU(T)} \subtype \IU(\Gamma),\IU(\Delta)}  \]
	 		With weakening, all this could give us a proof of $\vect{a} : \vect{\IU(U)} \vdash \pserv$. However, we want a proof of $\Gamma' \vdash \pserv$. As explained before, the only difference between $\Gamma'$ and $\vect{a} : \vect{\IU(U)}$ is for names in $\IT(P)$ that were originally both input and output. Let us track the role of those names. For channel name, they are not useful in the typing of $a$, as $a$ is a server, and they do not appear in $\IU(\Gamma'')$ because $\Gamma''$ is time invariant. Now we only need to work on server name that were originally both input and output but became only output in $\Gamma'$. As $P$ is in normal form, $\pout \notin \IT(P)$. In particular, the server name $a$ is not a top name in $\IT(P)$. As a consequence, changing the type of server name in $\IT(P)$ has no consequence for the typing of $a$. Moreover, as $\Gamma''$ is time invariant, changing input and output servers type to output types has no consequence for the typing of $P$. So, in the end, we have indeed $\Gamma' \vdash \pserv$.  
	 		
	 		\item If $\pin$ is a top guarded processes of $P$ not in $\IT(P)$. Let us consider the typing of $\pin$. 
	 		\begin{prooftree}
	 			\footnotesize
	 			\AxiomC{$\Ietype{a}{\Iich}$}
	 			\AxiomC{$\Itype[@][@][\Idecr]{P}{K'}$}
	 			\BinaryInfC{$\Itype{\pin}{I + K'}$}
	 			\AxiomC{$\phi;\Phi \vdash \vect{a} : \vect{U} \subtype \Gamma; I + K' \le K$}
	 			\BinaryInfC{$\Itype[@][@][\vect{a} : \vect{U}]{\pin}{K}$}
	 		\end{prooftree}
	 		By Lemma~\ref{l:forgetfultyping}, we have:
	 		\[ \IU(\Gamma) \vdash a : \ich[\IU(\vect{T})] \qquad \IU(\Gamma) = \IU(\Idecr), \Delta \qquad \IU(\Idecr) \vdash P \qquad \vect{a} : \vect{\IU(T)} \subtype \IU(\Gamma) \]
	 		So, by weakening we obtain a proof $\vect{a} : \IU(\vect{U}) \vdash \pin$. Again, we want a proof of $\Gamma' \vdash \pin$. As $\pin \notin \IT(P)$, we have $\phi;\Phi \vDash I \ge 1$. As a consequence, the typing of $a$ is not modified in $\Gamma'$. Moreover, all time out channels are erased in $\Idecr$, so they have no incidence on the typing. Finally, the timed out input/output server names are changed to output server names by $\Idecr$, thus we have indeed $\Gamma' \vdash \pin$.
	 		\item If $\pout$ is a top guarded processes of $P$ not in $\IT(P)$. There are two cases, $a$ is a channel name or $a$ is a server name. In both cases, the proof has the same reasoning as the one for $\pin$. 
	 		\item If $\tick.P$ is a top guarded processes of $P$. By definition, $\tick.P \notin \IT(P)$. Then, the proof has the same reasoning as the one for $\pin$. 		
	 	\end{itemize}
	 	This concludes the proof of Lemma~\ref{l:timeoutisdeadlock}.
	\end{proof}
	As a consequence, when given a typed process $P$ is normal form, we can define $\IS = \IT(P)$ and we now that if $\Idisc{P,P'}$, then $P'$ can simulate the strategy of Definition~\ref{d:reductionstrategy} on $P$.
	
	Now we want to show the following theorem:
	\begin{theorem}[Quantitative Subject Reduction]
		If $P$ is in normal form for $\red$, $P \tred Q$ with $Q \ne P$ and $\Itype[@][@][\cdot]{P}{K}$ then, if we pose $\IS = \IT(P)$, we have $\Itype[@][@][\cdot]{Q'}{K'}$ with $\Idisc{Q,Q'}$ and $\phi;\Phi \vDash K' + 1 \le K$. 
	\label{t:quantitativesubjred} 
	\end{theorem}
	\begin{proof}
		By Lemma \ref{l:closedtypednormalform2}, we know the canonical form of $P$. As $P \tred Q$ with $Q \ne P$, we know that $P$ has at least one top guarded process starting with a $\tick$. Let us write:
		\[ P \congr \pnu[\vect{a}] ( \IT(P) \parr P_1 \parr \cdots \parr P_n \parr \tick.R_1 \parr \cdots \parr \tick.R_m )\] 
		With $m \ge 1$. Then, let us pose:
		\[ Q' = \pnu[\vect{a}] ( P_1 \parr \cdots \parr P_n \parr R_1 \parr \cdots \parr R_m )\] 
		By Lemma~\ref{l:exhaustivetimereductions}, we have:
		\[ Q \congr \pnu[\vect{a}] ( \IT(P) \parr P_1 \parr \cdots \parr P_n \parr R_1 \parr \cdots \parr R_m )\]
		Thus, we have indeed $\Idisc{Q,Q'}$ by Lemma~\ref{l:timeoutisdeadlock}. Now, let us consider the typing for the canonical form of $P$ to give a typing for $Q'$. As previously, we consider that subtyping have been pushed to the typing of guarded processes.
		\begin{prooftree}
			\footnotesize
			\AxiomC{$\Itype[@][@][\vect{a} : \vect{U}]{\IT(P)}{K}$}
			\AxiomC{$\forall i,~ \Itype[@][@][\vect{a} : \vect{U}]{P_i}{K}$}
			\AxiomC{$\forall j,~ \Itype[@][@][\vect{a} : \vect{U}]{\tick. R_j}{K}$}
			\doubleLine 
			\TrinaryInfC{$\Itype[@][@][\vect{a} : \vect{U}]{(\IT(P) \parr P_1 \parr \cdots \parr P_n \parr \tick.R_1 \parr \cdots \parr	\tick.R_m)}{K}$}
			\doubleLine 
			\UnaryInfC{$\Itype[@][@][\cdot]{\pnu[\vect{a}](\IT(P) \parr P_1 \parr \cdots \parr P_n \parr \tick.R_1 \parr \cdots \parr \tick.R_m)}{K}$}
		\end{prooftree} 
		First, we can see that because of the rule for $\tick$, we have $\phi;\Phi \vDash K \ge 1$. So, we will show that $\Itype[@][@][\cdot]{Q'}{K - 1}$, and we have indeed $\phi;\Phi \vDash (K-1)+1 \le K$. (Note that if $K = 0$ this last inequation is not true, that is why we first need $\phi;\Phi \vDash K \ge 1$).
		In order to do this, let us decide a new assignment to the name in $\vect{a}$. In order to do this, we take the assignment $\vect{a} : \vect{U}$ and we apply the following function on $\vect{U}$:  
		\begin{itemize}
			\item $\Ich \mapsto \Ich[(I-1)]$. Other channel types follow exactly the same pattern.
			\item $\Iserv \mapsto \Iserv[(I-1)]$. Other server types follow exactly the same pattern.  
		\end{itemize} 
		We denote this new assignment by $\vect{a} : \vect{U'}$. Remark that we have $\vect{a} : \vect{U'} = \Gamma', \Delta'$ with the subtyping $\phi;\Phi \vdash \Gamma' \subtype \Idecr[(\vect{a}: \vect{U})][1]$. Note also that we have the following lemma:
		\begin{lemma}
			$\Idecr = \Idecr[{\Idecr[@][1]}][(I-1)]$ when $\phi;\Phi \vDash I \ge 1$. 
		\label{l:doubletimeadvance}
		\end{lemma}
		The proof is simple as when $\phi;\Phi \vDash I \ge 1$, we have  $\phi;\Phi \vDash J \ge I \Leftrightarrow \phi;\Phi \vDash J \ge 1 \land (J-1) \ge (I-1)$.  
		  Now, let us give a typing for $Q'$.
		\begin{prooftree}
			\footnotesize
			\AxiomC{$\forall i,~ \Itype[@][@][\vect{a} : \vect{U'}]{P_i}{K-1}$}
			\AxiomC{$\forall j,~ \Itype[@][@][\vect{a} : \vect{U'}]{R_j}{K-1}$}
			\doubleLine 
			\BinaryInfC{$\Itype[@][@][\vect{a} : \vect{U'}]{(P_1 \parr \cdots \parr P_n \parr R_1 \parr \cdots \parr	R_m)}{K-1}$}
			\doubleLine 
			\UnaryInfC{$\Itype[@][@][\cdot]{\pnu[\vect{a}](P_1 \parr \cdots \parr P_n \parr R_1 \parr \cdots \parr R_m)}{K-1}$}
		\end{prooftree}
		Where the proofs for the $P_i$ are:
		\begin{itemize}
			\item If $P_i = \pserv$. Then, the original proof for the canonical form of $P$ was:
			
			\medskip 
			
			\leftshiftt{
				\footnotesize 
				\AxiomC{$\Ietype[@][@][\Gamma,\Delta]{a}{\Iiserv[@][@][K']}$}
				\AxiomC{$\Itype[(\phi,\vect{i})][@][\Gamma'', \vect{v} : \vect{T}]{P}{K'}$}
				\AxiomC{$\Gamma''$ time invariant}
				\AxiomC{$\phi;\Phi \vdash \Idecr[@][@][\phi;\Phi] \subtype \Gamma''$}
				\QuaternaryInfC{$\Itype[@][@][\Gamma, \Delta]{\pserv}{I}$}
				\AxiomC{$\phi;\Phi \vdash (\vect{a} : \vect{U}) \subtype \Gamma,\Delta; I \le K$}
				\insertBetweenHyps{\hskip -2cm}
				\BinaryInfC{$\Itype[@][@][\vect{a} : \vect{U}]{\pserv}{K}$}
				\DisplayProof}

			\medskip
			
			Then, an easy case to consider is when $\phi;\Phi \vDash I \ge 1$. In this case, we can give the following typing, by using previous remarks, Lemma~\ref{l:advanceandsubtyping} and Lemma~\ref{l:doubletimeadvance}. Indeed, Lemma~\ref{l:advanceandsubtyping} gives us:
			\[ \Idecr[(\vect{a}: \vect{U})][1] = \Gamma_0, \Delta'' \text{ with } \phi;\Phi \vdash \Gamma_0 \subtype \Idecr[@][1],\Idecr[\Delta][1] \]
			and so we obtain the following proof (without recalling the subtyping):
			
			\medskip 
			
			\leftshiftt{
				\footnotesize 
				\AxiomC{$\Ietype[@][@][{\Idecr[\Gamma][1],\Idecr[\Delta][1]},\Delta',\Delta'']{a}{\Iiserv[(I-1)][@][K']}$}
				\AxiomC{$\Itype[(\phi,\vect{i})][@][\Gamma'', \vect{v} : \vect{T}]{P}{K'}$}
				\AxiomC{$\Gamma''$ time invariant}
				\AxiomC{$\phi;\Phi \vdash \Idecr[@][@][\phi;\Phi] \subtype \Gamma''$}
				\QuaternaryInfC{$\Itype[@][@][{\Idecr[\Gamma][1],\Idecr[\Delta][1]},\Delta',\Delta'']{\pserv}{I-1}$}
				\doubleLine
				\UnaryInfC{$\Itype[@][@][{\Idecr[(\vect{a}: \vect{U})][1],\Delta'}]{\pserv}{K-1}$}
				\doubleLine 
				\UnaryInfC{$\Itype[@][@][\Gamma',\Delta']{\pserv}{K-1}$}
				\DisplayProof}
				
			\medskip
			
			Now we need to consider the more difficult case $\phi;\Phi \not\vDash I \ge 1$. Let us consider that type $T$ assigned to $a$ in $\vect{a} : \vect{U}$. By Lemma~\ref{l:exhaustivesubtyping}, as $\phi;\Phi \vdash T \subtype \Iiserv[@][@][K']$ we have: 
			\[T = \Iserv[J][@][K''][\vect{T'}] \qquad \phi;\Phi \vDash I = J \qquad (\phi,\vect{i});\Phi \vdash \vect{T'} \subtype \vect{T} \qquad (\phi,\vect{i});\Phi \vDash K' \le K'' \]
			or,
			\[T = \Iiserv[J][@][K''][\vect{T'}] \qquad \phi;\Phi \vDash I = J \qquad (\phi,\vect{i});\Phi \vdash \vect{T'} \subtype \vect{T} \qquad (\phi,\vect{i});\Phi \vDash K' \le K'' \]
			In both cases, we can see that the type $T'$ assigned to $a$ in $\vect{a} : \vect{U'}$ is such that we have the subtyping	 $\phi;\Phi \vdash T' \subtype \Iiserv[(I-1)][@][K']$
	
			Now, let us look at what happens to server name in $\Gamma''$ the original proof for the canonical form of $P$. For all axiom $b : \Ioserv[0]$ in $\Gamma''$, we have $b : T \in \vect{a} : \vect{U}$ such that:
			\[ \phi;\Phi \vdash T \subtype T' \qquad \phi;\Phi \vdash \Idecr[T'] \subtype  \Ioserv  \]
			So, with Lemma~\ref{l:exhaustivesubtyping}, we have:
			\[\Idecr[T'] = \Iserv[J][@][K'][\vect{T'}] \qquad \phi;\Phi \vDash J = 0 \qquad (\phi,\vect{i});\Phi \vdash \vect{T} \subtype \vect{T'} \qquad (\phi,\vect{i});\Phi \vDash K' \le K \]
			or
			\[\Idecr[T'] = \Ioserv[J][@][K'][\vect{T'}] \qquad \phi;\Phi \vDash J = 0 \qquad (\phi,\vect{i});\Phi \vdash \vect{T} \subtype \vect{T'} \qquad (\phi,\vect{i});\Phi \vDash K' \le K \]
			Then, by Definition~\ref{d:advancetime}, this gives the following possibilities:
			\[T' = \Iserv[(J+I)][@][K'][\vect{T'}] \qquad \phi;\Phi \vDash J = 0  \qquad (\phi,\vect{i});\Phi \vdash \vect{T} \subtype \vect{T'} \qquad (\phi,\vect{i});\Phi \vDash K' \le K \]
			or
			\[T' = \Iserv[J'][@][K'][\vect{T'}] \qquad \phi;\Phi \not\vDash J' \ge I \qquad \phi;\Phi \vDash J'-I = 0 \qquad (\phi,\vect{i});\Phi \vdash \vect{T} \subtype \vect{T'} \qquad (\phi,\vect{i});\Phi \vDash K' \le K \]
			or 
			\[T' = \Ioserv[J'][@][K'][\vect{T'}] \qquad \phi;\Phi \vDash J'-I = 0 \qquad (\phi,\vect{i});\Phi \vdash \vect{T} \subtype \vect{T'} \qquad (\phi,\vect{i});\Phi \vDash K' \le K \]
			Note that the two first case can be combined by a type with a time $J'$ such that $\phi;\Phi \vDash J'-I = 0$. Then, we can use again Lemma~\ref{l:exhaustivereductions}, and we obtain:
			\[T = \Iserv[J''][@][K''][\vect{T''}] \qquad \phi;\Phi \vDash J''-I = 0 \qquad (\phi,\vect{i});\Phi \vdash \vect{T} \subtype \vect{T''} \qquad (\phi,\vect{i});\Phi \vDash K'' \le K \]
			or 
			\[T' = \Ioserv[J''][@][K''][\vect{T''}] \qquad \phi;\Phi \vDash J''-I = 0 \qquad (\phi,\vect{i});\Phi \vdash \vect{T} \subtype \vect{T''} \qquad (\phi,\vect{i});\Phi \vDash K'' \le K \]
			So, in the assignment $\vect{a} : \vect{U'}$, this type $T$ is sent to a type $T_{new}$ corresponding to $T$ with a time $J'' - 1$ instead of $J''$. It is easy to see that in both cases, we have $\phi;\Phi \vdash \Idecr[T_{new}] \subtype \Ioserv[0]$. So, we can write $\vect{a} : \vect{U'} = \Gamma_1,\Delta_1$ with $\phi;\Phi \vdash \Idecr[\Gamma_1] \subtype \Gamma''$. And we obtain the following typing:
			
			\medskip 
			
			\leftshiftt{
				\footnotesize 
				\AxiomC{$\Ietype[@][@][\Gamma_1,\Delta_1]{a}{\Iiserv[(I-1)][@][K']}$}
				\AxiomC{$\Itype[(\phi,\vect{i})][@][\Gamma'', \vect{v} : \vect{T}]{P}{K'}$}
				\AxiomC{$\Gamma''$ time invariant}
				\AxiomC{$\phi;\Phi \vdash \Idecr[\Gamma_1] \subtype \Gamma''$}
				\QuaternaryInfC{$\Itype[@][@][\Gamma_1,\Delta_1]{\pserv}{I-1}$}
				\doubleLine 
				\UnaryInfC{$\Itype[@][@][\Gamma_1,\Delta_1]{\pserv}{K-1}$}
				\DisplayProof}
			
			\medskip
			
			This concludes this case. 
			
			\item If $P_i = \pin$. Then, the original proof for the canonical form of $P$ was:
			\begin{prooftree}
				\footnotesize 
				\AxiomC{$\Ietype{a}{\Iich}$}
				\AxiomC{$\Itype[@][@][\Idecr, \vect{v} : \vect{T}]{P}{K'}$}
				\BinaryInfC{$\Itype{\pserv}{I + K'}$}
				\AxiomC{$\phi;\Phi \vdash (\vect{a} : \vect{U}) \subtype \Gamma; I + K' \le K$}
				\BinaryInfC{$\Itype[@][@][\vect{a} : \vect{U}]{\pin}{K}$}
			\end{prooftree}
			
			As $P_i \notin \IT(P)$, we have $\phi;\Phi \vDash I \ge 1$. We can give the following typing, by using previous remarks, Lemma~\ref{l:advanceandsubtyping} and Lemma~\ref{l:doubletimeadvance}. Indeed, Lemma~\ref{l:advanceandsubtyping} gives us:
			\[ \Idecr[(\vect{a}: \vect{U})][1] = \Gamma_0, \Delta'' \text{ with } \phi;\Phi \vdash \Gamma_0 \subtype \Idecr[@][1] \]
			and so we obtain the following proof (without recalling the subtyping):
			\begin{prooftree}
				\footnotesize 
				\AxiomC{$\Ietype[@][@][{\Idecr[\Gamma][1]},\Delta',\Delta'']{a}{\Iich[(I-1)]}$}
				\AxiomC{$\Itype[@][@][\Idecr, \vect{v} : \vect{T}]{P}{K'}$}
				\BinaryInfC{$\Itype[@][@][{\Idecr[\Gamma][1],\Delta',\Delta''}]{\pin}{(I-1) + K'}$}
				\doubleLine
				\UnaryInfC{$\Itype[@][@][{\Idecr[(\vect{a}: \vect{U})][1],\Delta'}]{\pin}{K-1}$}
				\doubleLine 
				\UnaryInfC{$\Itype[@][@][\Gamma',\Delta']{\pin}{K-1}$}
			\end{prooftree}
			\item If $P_i = \pout$, then we can do as the previous case for $\pin$. 	
		\end{itemize} 
		Then, we need to type $R_j$. By the remark that $\vect{a} : \vect{U'}$ can be written $\Gamma',\Delta'$ with the subtyping \hbox{$\phi;\Phi \vdash \Gamma' \subtype \Idecr[(\vect{a} : \vect{U})][1]$}, then it can be done by weakening and subtyping. 
		This concludes the proof for Theorem~\ref{t:quantitativesubjred}. 
	\end{proof}
	
	\subsection{Complexity Bound}
	\label{ss:complexitybound}
	
	This short section is to prove the main theorem of this paper. 
	
	\begin{theorem}
		If $\Itype[@][@][\cdot]{P}{K}$ and $P$ reduces to $Q$ by the strategy of Definition~\ref{d:reductionstrategy} with $n$ time reductions, then $\phi;\Phi \vDash K \ge n$.
	\label{t:complexitybound}
	\end{theorem}
	
	\begin{proof}
		We prove the following lemma:
		\begin{lemma}
			For all integer $n$, for all set of constraints $\Phi$ over $\phi$, for all index $K$, for all processes $P,Q$, if $\Itype[@][@][\cdot]{P}{K}$ and $P$ can be reduced to $Q$ by the strategy of Definition~\ref{d:reductionstrategy} with $n$ time reductions, then $\phi;\Phi \vDash K \ge n$. 
		\label{l:complexityboundaux}
		\end{lemma}
		By induction on $n$. 
		\begin{itemize}
			\item \textbf{Case} $n=0$. For any $\phi,\Phi,K$, we have $\phi;\Phi \vDash K \ge 0$, so we obtain directly this case. 
			\item \textbf{Case} $n+1$. By Definition~\ref{d:reductionstrategy}, we have $P \red^* P'$ with $P'$ in normal form for $\red$. Moreover, $P' \tred P_0$, with $P_1 \ne P'$ and $P_1$ can be reduced to $Q$ by the strategy of Definition~\ref{d:reductionstrategy} with $n$ time reductions. By hypothesis, we have
			$\Itype[@][@][\cdot]{P}{K}$. By Theorem~\ref{t:nonquantsubjectreduction}, we have $\Itype[@][@][\cdot]{P'}{K}$. By Theorem~\ref{t:quantitativesubjred}, if we pose $\IS = \IT(P')$, we have $\Itype[@][@][\cdot]{P_1}{K'}$ with $\Idisc{P_0,P_1}$ and $\phi;\Phi \vDash K' + 1 \le K$. 
			
			By Lemma~\ref{l:discardsimulation} and Lemma~\ref{l:discardnormalform}, there exists $Q_1$ such that $\Idisc{Q,Q_1}$ and $P_1$ can be reduced to $Q_1$ by the strategy with $n$ time reductions. By induction hypothesis, we obtain $\phi;\Phi \vDash K' \ge n$. Thus, $\phi;\Phi \vDash K \ge K' + 1 \ge n + 1$. This conludes the proof.  	
		\end{itemize}
	\end{proof}
	
	So we have indeed that the typing of a process can give a bound on its complexity under maximal parallelism. 
	
	\section{Examples}
	\label{s:examples}
	
	We present here an example for this complexity, showing that under maximal parallelism, merge sort has a linear number of comparison. Suppose given a replicated input with name $compare$ doing the comparison between two elements of type $\IB$. We want to count the number of comparison in merge sort. The processes are described in Figure~\ref{f:fusionsort}.   
	
	\begin{figure}
		\centering
		\begin{framed}
			\begin{lstlisting}[style=ocaml,emph={l,a,x,y,q,r,b,z,c,a',b',c',d},mathescape=true]
			!merge($l_0$,$l_1$,a). match $l_0$ with 
				| [] -> $\overline{a}$<$l_1$>
				| x::q -> match $l_1$ with 
					| [] -> $\overline{a}$<$l_0$>
					| y::r -> (nub)($\mathtt{tick}.\overline{compare}$<x,y,b> | b(z). if z then 
							(nuc)($\overline{merge}$<q,$l_1$,c> | c($l_2$).$\overline{a}$<x::$l_2$>) else
							(nuc)($\overline{merge}$<$l_0$,r,c> | c($l_2$).$\overline{a}$<y::$l_2$>)
							)
			
			
			!decompose(l,$a_0$,$a_1$). match l with 
				| [] -> $\overline{a_0}$<[]> | $\overline{a_1}$<[]>
				| x::q -> match q with 
					| [] -> $\overline{a_0}$<x::[]> | $\overline{a_1}$<[]> 
					| y::r -> (nu$b_0$)(nu$b_1$).($\overline{decompose}$<r,$b_0$,$b_1$> | $b_0(l_0)$.$\overline{a_0}$<x::$l_0$> | $b_1(l_1)$.$\overline{a_1}$<y::$l_1$>)
					
					
			!mergesort(l,a). match l with 
				| [] -> $\overline{a}$<[]>
				| x::q -> match q with 
					| [] -> $\overline{a}$<x::[]>
					| y::r -> (nu$b_0$)(nu$b_1$)(nu$c_0$)(nu$c_1$)(nud)($\overline{decompose}$<l,$b_0$,$b_1$> 
						| $b_0(l_0)$.$\overline{mergesort}$<$l_0$,$c_0$> | $b_1(l_1)$.$\overline{mergesort}$<$l_1$,$c_1$> 
						| $c_0(q_0)$.$c_1(q_1)$.$\overline{merge}$<$q_0$,$q_1$,d> | d($q_2$).$\overline{a}$<$q_2$> 
						)
					
			
			
			\end{lstlisting}
		\end{framed}
		\caption{Merge Sort}
		\label{f:fusionsort}
	\end{figure}
	
	We now describe the typing for those processes. To take in account the complexity of the comparison, we do a tick before each call to compare, and we give this server the type $\Iserv[0][][0][\IB,\IB,{\Ioch[0][\Ibool]}]$. Note that we could have equivalently given a complexity $1$ to the server and removed the tick. In order to simplify the notation, we bound the sizes of the lists in mergesort by an exponent of $2$. 
	Let us pose the following context:
	\[ \Gamma := compare : \Iserv[0][][0][\IB,\IB,{\Ioch[0][\Ibool]}], merge : \Iserv[0][(i,j)][i+j][{\Ilis[0][i], \Ilis[0][j],\Ioch[i+j][{\Ilis[0][i+j]}]}], \]
	\[ decompose : \Iserv[0][i][0][{\Ilis[0][2i],\Ioch[0][{\Ilis[0][i]}],\Ioch[0][{\Ilis[0][i]}]}], \]
	\[ mergesort : \Iserv[0][i][2^{i+1}][{\Ilis[0][2^i],\Ioch[2^{i+1}][{\Ilis[0][2^i]}]}] \]
	
	And we pose $\Gamma_o$ the same context with output server instead of input/output, thus $\Gamma_o$ is time unlimited, and we have $\cdot \vdash \Gamma \subtype \Gamma_o$. Moreover, $\Idecr[\Gamma][0] = \Gamma$.  
	
	We want to show that the servers are well typed under this context. In the typing, we omit the typing of expressions when it is obvious, that is to say only syntax directed rule or a subtyping where an input/output becomes only input or only output without changing anything else. 
	
	We start with the server for merge. We pose:
	\[ \Gamma' := \Gamma_o, l_0 : \Ilis[0][i], l_1 : \Ilis[0][j], a : \Ioch[i+j][{\Ilis[0][i+j]}] \]
	\[ \Gamma'' := \Gamma', x : \IB, q : \Ilis[0][i-1], y: \IB, r : \Ilis[0][j-1] \]
	
	\begin{prooftree}
		\footnotesize
		\AxiomC{}
		\UnaryInfC{$\Itype[(i,j)][0 \le 0][\Gamma']{\pout[a][l_1]}{i+j}$}
		\AxiomC{}
		\UnaryInfC{$\Itype[(i,j)][(i \ge 1, j \ge 1)][{\Idecr[\Gamma''][1],b: \Ioch[0][\Ibool]}]{\pout[compare][x,y,b]}{0}$}
		\UnaryInfC{$\Itype[(i,j)][(i \ge 1, j \ge 1)][{\Gamma'',b: \Ioch[1][\Ibool]}]{\tick. \pout[compare][x,y,b]}{1}$}
		\UnaryInfC{$\Itype[(i,j)][(i \ge 1, j\ge 1)][{\Gamma'',b: \Ich[1][\Ibool]}]{\tick. \pout[compare][x,y,b]}{i+j}$}
		\AxiomC{$\pi$}
		\BinaryInfC{$\Itype[(i,j)][(i \ge 1, j \ge 1)][{\Gamma'',b: \Ich[1][\Ibool]}]{\tick. \pout[compare][x,y,b] \parr \pin[b][z][\cdots]}{i+j}$}
		\UnaryInfC{$\Itype[(i,j)][(i \ge 1, j \ge 1)][\Gamma'']{\pnu[b] (\tick . \pout[compare][x,y,b] \parr \pin[b][z][\cdots])}{i+j}$}
		\doubleLine
		\UnaryInfC{$\Itype[(i,j)][(i \ge 1)][{\Gamma', x : \IB, q : \Ilis[0][i-1]}]{\pifl[l_1][{\pout[a][l_0]}][y][r][\cdots]}{i+j}$}
		\BinaryInfC{$\Itype[(i,j)][\cdot][\Gamma']{\pifl[l_0][{\pout[a][l_1]}][x][q][\cdots]}{i+j}$}
		\UnaryInfC{$\Itype[\cdot][\cdot][\Gamma]{\pserv[merge][l_0,l_1,a][\cdots]}{0}$}
	\end{prooftree}
	
	where $\pi$ is the following proof, in which we pose:
	\[ \Delta = \Gamma_o, l_0 : \Ilis[0][i], l_1 : \Ilis[0][j], a : \Ioch[i+j-1][{\Ilis[0][i+j]}], \]
	\[  x : \IB, q : \Ilis[0][i-1], y: \IB, r : \Ilis[0][j-1], b : \Ich[0][\Ibool], z : \Ibool \]
	\[ \Delta' = \Delta, c : \Ich[i+j-1][{\Ilis[0][i+j-1]}] \]
	\[ \Delta'' = \Idecr[\Delta'][(i+j-1)],l_2 : \Ilis[0][i+j-1] \]
	
	\begin{prooftree}
		\footnotesize
		\AxiomC{}
		\UnaryInfC{$\Itype[(i,j)][(i \ge 1, j \ge 1)][\Delta']{\pout[merge][q,l_1,c]}{i+j-1}$}
		\AxiomC{}
		\UnaryInfC{$\Itype[(i,j)][(i \ge 1, j \ge 1)][\Delta'']{\pout[a][x::l_2]}{0}$} 
		\UnaryInfC{$\Itype[(i,j)][(i \ge 1, j \ge 1)][\Delta']{\pin[c][l_2][{\pout[a][x::l_2]}]}{i+j-1}$}
		\BinaryInfC{$\Itype[(i,j)][(i \ge 1, j \ge 1)][{\Delta, c : \Ich[i+j-1][{\Ilis[0][i+j-1]}]}]{\pout[merge][q,l_1,c] \parr \pin[c][l_2][{\pout[a][x::l_2]}]}{i+j-1}$}
		\UnaryInfC{$\Itype[(i,j)][(i \ge 1, j \ge 1)][\Delta]{\pnu[c] (\pout[merge][q,l_1,c] \parr \pin[c][l_2][{\pout[a][x::l_2]}])}{i+j-1}$}
		\AxiomC{$\cdots$} 
		\BinaryInfC{$\Itype[(i,j)][(i \ge 1, j \ge 1)][\Delta]{\pif[z][{\pnu[c] \cdots}][{\pnu[c] \cdots}]}{i+j-1}$}
		\UnaryInfC{$\Itype[(i,j)][(i \ge 1, j \ge 1)][{\Gamma'',b: \Ich[1][\Ibool]}]{\pin[b][z][\cdots]}{i+j}$}
	\end{prooftree}
	And the typing for the other branch of the conditional is similar. 
	
	Now, we type the server for the decompose function. We pose:
	\[ \Gamma' := \Gamma_o, l : \Ilis[0][2i], a_0 : \Ioch[0][{\Ilis[0][i]}], a_1 : \Ioch[0][{\Ilis[0][i]}] \]
	\[ \Gamma'' := \Gamma', x : \IB, q : \Ilis[0][2i-1], y: \IB, r : \Ilis[0][2(i-1)] \]
	\[ \Delta = \Gamma'', b_0 : \Ich[0][{\Ilis[0][i-1]}], b_1 : \Ich[0][{\Ilis[0][i-1]}] \]
	
	\begin{prooftree}
		\footnotesize
		\AxiomC{}
		\UnaryInfC{$\Itype[i][\cdot][\Gamma']{(\pout[a_0][\nil] \parr \pout[a_1][\nil])}{0}$} 
		\AxiomC{}
		\UnaryInfC{$\Itype[i][(i \ge 1)][{\Delta, l_0 : \Ilis[0][i-1]}]{\pout[a_0][x::l_0]}{0}$}
		\UnaryInfC{$\Itype[i][(i \ge 1)][\Delta]{\pin[b_0][l_0][{\pout[a_0][x::l_0]}]}{0}$}
		\AxiomC{$\cdots$}
		\BinaryInfC{$\Itype[i][(i \ge 1)][\Delta]{\pin[b_0][l_0][{\pout[a_0][x::l_0]}] \parr \pin[b_1][l_1][{\pout[a_1][y::l_1]}]}{0}$}
		\doubleLine 
		\UnaryInfC{$\Itype[i][(i \ge 1)][{\Gamma'',b_0 : \Ich[0][{\Ilis[0][i-1]}], b_1 : \Ich[0][{\Ilis[0][i-1]}]}]{\pout[decompose][r,b_0,b_1] \parr \cdots}{0}$}
		\UnaryInfC{$\Itype[i][(i \ge 1)][\Gamma'']{\pnu[b_0] \pnu[b_1] (\pout[decompose][r,b_0,b_1] \parr \cdots)}{0}$}
		\doubleLine
		\UnaryInfC{$\Itype[i][(2i \ge 1)][{\Gamma', x : \IB, q : \Ilis[0][2i-1]}]{\pifl[q][{(\pout[a_0][x::\nil] \parr \pout[a_1][\nil])}][y][r][\cdots]}{0}$}
		\BinaryInfC{$\Itype[i][\cdot][\Gamma']{\pifl[l][{(\pout[a_0][\nil] \parr \pout[a_1][\nil])}][x][q][\cdots]}{0}$}
		\UnaryInfC{$\Itype[\cdot][\cdot][\Gamma]{\pserv[decompose][l,a_0,a_1][\cdots]}{0}$}
	\end{prooftree}
	
	And finally, the typing for the server computing the merge sort. We pose:
	\[ \Gamma' := \Gamma_o, l : \Ilis[0][2^i], a : \Ioch[2^{i+1}][{\Ilis[0][2^i]}] \]
	\[ \Gamma'' := \Gamma', x : \IB, q : \Ilis[0][2^i-1], y: \IB, r : \Ilis[0][2^i-2] \]
	\[ \Delta = \Gamma'', b_0 : \Ich[0][{\Ilis[0][2^{i-1}]}], b_1 : \Ich[0][{\Ilis[0][2^{i-1}]}], \]
	\[ c_0 :\Ich[2^i][{\Ilis[0][2^{i-1}]}], c_1 :\Ich[2^i][{\Ilis[0][2^{i-1}]}] , d: \Ich[2^{i+1}][{\Ilis[0][2^i]}] \]
	
	\begin{prooftree}
		\footnotesize
		\AxiomC{}
		\UnaryInfC{$\Itype[i][\cdot][\Gamma']{\pout[a][\nil]}{2^{i+1}}$} 
		\AxiomC{See below} 
		\UnaryInfC{$\Itype[i][(2^i \ge 2)][\Delta]{ \pout[decompose][l,b_0,b_1] \parr \cdots}{2^{i+1}}$}
		\UnaryInfC{$\Itype[i][(2^i \ge 2)][\Gamma'']{\pnu[(b_0,b_1,c_0,c_1,d)] \cdots}{2^{i+1}}$}
		\doubleLine
		\UnaryInfC{$\Itype[i][(2^i \ge 1)][{\Gamma', x : \IB, q : \Ilis[0][2^i-1]}]{\pifl[q][{(\pout[a][x::\nil])}][y][r][\cdots]}{2^{i+1}}$}
		\BinaryInfC{$\Itype[i][\cdot][\Gamma']{\pifl[l][{\pout[a][\nil]}][x][q][\cdots]}{2^{i+1}}$}
		\UnaryInfC{$\Itype[\cdot][\cdot][\Gamma]{\pserv[mergesort][l,a][\cdots]}{0}$}
	\end{prooftree}
	
	Where we detail here the different proofs for the $5$ processes in parallel.
	
	\begin{prooftree}
		\footnotesize
		\AxiomC{$\Ietype[i][(2^i \ge 2)][\Delta]{l}{\Ilis[0][2*2^{i-1}]}$}
		\AxiomC{$\Ietype[i][(2^i \ge 2)][\Delta]{b_i}{\Ioch[0][{\Ilis[0][2^{i-1}]}]}$}
		\BinaryInfC{$\Itype[i][(2^i \ge 2)][\Delta]{ \pout[decompose][l,b_0,b_1]}{0}$} 
		\UnaryInfC{$\Itype[i][(2^i \ge 2)][\Delta]{ \pout[decompose][l,b_0,b_1]}{2^{i+1}}$}
	\end{prooftree}
	
	\begin{prooftree}
		\footnotesize
		\AxiomC{$\Ietype[i][(2^i \ge 2)][\Delta]{l_0}{\Ilis[0][2^{i-1}]}$}
		\AxiomC{$\Ietype[i][(2^i \ge 2)][\Delta]{c_0}{\Ioch[2^i][{\Ilis[0][2^{i-1}]}]}$}
		\BinaryInfC{$\Itype[i][(2^i \ge 2)][{\Delta, l_0 : \Ilis[0][2^{i-1}]}]{\pout[mergesort][l_0,c_0]}{2^i}$} 
		\UnaryInfC{$\Itype[i][(2^i \ge 2)][{\Delta, l_0 : \Ilis[0][2^{i-1}]}]{\pout[mergesort][l_0,c_0]}{2^{i+1}}$} 
		\UnaryInfC{$\Itype[i][(2^i \ge 2)][\Delta]{\pin[b_0][l_0][{\pout[mergesort][l_0,c_0]}]}{2^{i+1}}$}
	\end{prooftree}
	
	\begin{prooftree}
		\footnotesize
		\AxiomC{$\Ietype[i][(2^i \ge 2)][{\Idecr[\Delta][2^i], q_0 : \Ilis[0][2^{i-1}], q_1 : \Ilis[0][2^{i-1}]}]{d}{\Ich[2^i][{\Ilis[0][2^i]}]}$}  
		\UnaryInfC{$\Itype[i][(2^i \ge 2)][{\Idecr[\Delta][2^i], q_0 : \Ilis[0][2^{i-1}], q_1 : \Ilis[0][2^{i-1}]}]{\pout[merge][q_0,q_1,d]}{2^i}$} 
		\UnaryInfC{$\Itype[i][(2^i \ge 2)][{\Idecr[\Delta][2^i], q_0 : \Ilis[0][2^{i-1}]}]{\pin[c_1][q_1][{\pout[merge][q_0,q_1,d]}]}{2^i}$} 
		\UnaryInfC{$\Itype[i][(2^i \ge 2)][\Delta]{\pin[c_0][q_0][{\pin[c_1][q_1][{\pout[merge][q_0,q_1,d]}]}]}{2^{i+1}}$}
	\end{prooftree}
	
	\begin{prooftree}
		\footnotesize
		\AxiomC{$\Ietype[i][(2^i \ge 2)][{\Idecr[\Delta][2^{i+1}], q_2 : \Ilis[0][2^i]}]{a}{\Ioch[0][{\Ilis[0][2^i]}]}$}
		\UnaryInfC{$\Itype[i][(2^i \ge 2)][{\Idecr[\Delta][2^{i+1}], q_2 : \Ilis[0][2^i]}]{\pout[a][q_2]}{0}$} 
		\UnaryInfC{$\Itype[i][(2^i \ge 2)][\Delta]{\pin[d][q_2][{\pout[a][q_2]}]}{2^{i+1}}$}
	\end{prooftree}
	
	So those servers are well typed under this context. As a consequence, when put in parallel with a call to mergesort with a list of size less than $N = 2^n$, we have a bound on the number of comparisons in the computation, and the bound is $2N$. In order to explicit the results of Theorem~\ref{t:complexitybound}, we show how the strategy of Definition~\ref{d:reductionstrategy} works on a call to mergesort. 
	
	First, we describe a call to merge with a list of size $n$ and a list of size $m$. In order to simplify the notations, we omit to recall the definition of servers, and we do not write them in the reduction as they are invariant. We consider the case $n > 0$ and $m > 0$, otherwise, there is no tick so there is no time reduction. We also consider $x_1 \le y_1$ and the other case is symmetric. 
	
	\begin{align*}
	\pout[merge][{[x_1;\cdots;x_n],[y_1;\cdots;y_m],a}] &\red^* \pnu[b] (\tick. \pout[compare][x_1,y_1,b]) \parr \pin[b][z][\cdots] \\
	& \tred \pnu[b] (\pout[compare][x_1,y_1,b]) \parr \pin[b][z][\cdots] \\
	& \red^* \pnu[b] \pnu[c] (\pout[merge][{[x_2,\cdots,x_n],[y_1,\cdots,y_m],c}] \parr \pin[c][l_2][{\pout[a][x_1 :: l_2]}])
	\end{align*}
	Then, we can show by induction that the number of step reduction is less than $n + m$ and that a call to merge produces at the end an output on the channel (and some name variables that we omit for simplicity). And then, we have:
	\begin{align*}
	\pnu[b] \pnu[c] (\pout[c][{[z_1;\dots;z_{n+m-1}]}] \parr \pin[c][l_2][{\pout[a][x_1 :: l_2]}]) \red^* \pnu[b] \pnu[c] (\pout[a][x_1 :: {[z_1;\cdots;z_{n+m-1}]}])  \\
	\end{align*}
	So we have indeed less than $n+m$ time reductions. 
	In a call to decompose, no time reduction occurs, so the reduction has indeed complexity $0$.
	For the example, we show a call to mergesort on a list of size $4$ and show that we have a bound of $8$ time reductions. The general behaviour can be deduced from this example. 
	
	\begin{align*}
	&\pout[mergesort][{[4;6;7;2]},a] \red^* \pnu[(b_0,b_1,c_0,c_1,d)] (\pout[decompose][{[4;6;7;2]},b_0,b_1] \parr \cdots)  \\
	 &\red^* \pnu[(b_0,b_1,c_0,c_1,d)] (\pout[mergesort][{[4;7],c_0}] \parr \pout[mergesort][{[6;2],c_1}] \parr \pin[c_0][q_0][{\pin[c_1][q_1][\cdots]}] \parr \pin[d][q_2][\cdots] ) \\ 
	 &\red^* \pnu[(\cdots)] (\pout[mergesort][{[4],c_0'}] \parr \pout[mergesort][{[7],c_1'}] \parr \pin[c_0'][q_0][{\pin[c_1'][q_1][\cdots]}] \parr \pin[d'][q_2][\cdots] \\
	&\parr \pout[mergesort][{[6],c_0''}] \parr \pout[mergesort][{[2],c_1''}] \parr \pin[c_0''][q_0][{\pin[c_1''][q_1][\cdots]}] \parr \pin[d''][q_2][\cdots] \parr \pin[c_0][q_0][{\pin[c_1][q_1][\cdots]}] \parr \pin[d][q_2][\cdots] ) \\ 
	&\red^* \pnu[(\cdots)] (\pout[c_0'][{[4]}] \parr \pout[c_1'][{[7]}] \parr \pin[c_0'][q_0][{\pin[c_1'][q_1][\cdots]}] \parr \pin[d'][q_2][\cdots] \\
	&\parr \pout[c_0''][{[6]}] \parr \pout[c_1''][{[2]}] \parr \pin[c_0''][q_0][{\pin[c_1''][q_1][\cdots]}] \parr \pin[d''][q_2][\cdots] \parr \pin[c_0][q_0][{\pin[c_1][q_1][\cdots]}] \parr \pin[d][q_2][\cdots] ) \\
	&\red^* \pnu[(\cdots)] (\pout[merge][{[4],[7],d'}] \parr \pin[d'][q_2][\cdots] \parr \pout[merge][{[6],[2],d''}] \parr \pin[d''][q_2][\cdots] \parr \pin[c_0][q_0][{\pin[c_1][q_1][\cdots]}] \parr \pin[d][q_2][\cdots] ) \\  
	&\red^* \pnu[(\cdots)] (\tick. \pout[compare][{4,7,b}] \parr \pin[b][z][\cdots] \parr \pin[d'][q_2][\cdots] \parr \\
	&\tick. \pout[compare][{6,2,b'}] \parr \pin[b'][z][\cdots] \parr \pin[d''][q_2][\cdots] \parr \pin[c_0][q_0][{\pin[c_1][q_1][\cdots]}] \parr \pin[d][q_2][\cdots] ) \\ 
	\end{align*}
	And this process is in normal form for $\red$. Thus, we can do the time reduction. 
	\begin{align*}
	&\tred \pnu[(\cdots)] (\pout[compare][{4,7,b}] \parr \pin[b][z][\cdots] \parr \pin[d'][q_2][\cdots] \parr \\
	&\pout[compare][{6,2,b'}] \parr \pin[b'][z][\cdots] \parr \pin[d''][q_2][\cdots] \parr \pin[c_0][q_0][{\pin[c_1][q_1][\cdots]}] \parr \pin[d][q_2][\cdots] ) \\ 
	&\red^* \pnu[(\cdots)] (\pout[d'][{[4;7]}] \parr \pin[d'][q_2][\cdots] \parr \pout[d''][{[2;6]}] \parr \pin[d''][q_2][\cdots] \parr \pin[c_0][q_0][{\pin[c_1][q_1][\cdots]}] \parr \pin[d][q_2][\cdots] ) \\
	&\red^* \pnu[(\cdots)] (\pout[c_0][{[4;7]}] \parr \pout[c_1][{[2;6]}] \parr \pin[c_0][q_0][{\pin[c_1][q_1][\cdots]}] \parr \pin[d][q_2][\cdots] ) \\
	&\red^* \pnu[(\cdots)] (\pout[merge][{[4;7],[2;6],d}] \parr \pin[d][q_2][{\pout[a][q_2]}] ) \\
	\end{align*} 
	And, with what we saw before, we can do the merging in at most $4$ time reductions, and finish the computation.
	
	If we want a more generic notion of complexity for programs, we believe it is best to consider the number of communications on channels. A good way to count this is to add a tick after each input (we could also add a tick before each input, or before each output, however in this case, stuck programs would have a complexity 1 even if they do not communicate, that is why we believe it is better to put them after the input). Without detailing the typing, we give the complexity for mergesort in this case. We suppose that the complexity of compare is $K_c$. 
	
	For the merge function, if we call $f(i,j)$ the complexity of a call to merge on input of sizes between $0$ and $i$ and $0$ and $j$, we obtain the following restrictions in the typing:
	\[ f(i,j) \ge 1 \qquad \forall i,j \ge 1, f(i,j) \ge f(i-1,j) + K_c + 3 \qquad \forall i,j \ge 1, f(i,j) \ge f(i,j-1) + K_c + 3  \]  
	Indeed, we always start by a tick, so the complexity for any call is more than $1$. Then, in the computation, when the lists are not empty, we have a total of $3$ input, a call to compare and a call to merge where the size of one of the list decrease by one. So, we could for example take 
	\[ f(i,j) = (3 + K_c)(i+j) + 1 \]
	
	Now, for the decompose function, if we call $f(i)$ the complexity for an input list of size smaller than $2i$, we obtain the following restrictions:
	\[ f(i) \ge 1 \qquad \forall i \ge 1, f(i) \ge 2 + f(i-1)  \]
	Thus, we obtain the complexity 
	\[ f(i) = 1 + 2i \]
	Finally, for the merge sort, we obtain the following restrictions, if we call $f(i)$ the complexity on an input list of size smaller than $2^i$. 
	\[ f(i) \ge 1 \qquad \forall i \ge 1, f(i) \ge 1 + 1 + 2^i + 1 + f(i-1) + 2 + (3 + K_c)2^i + 1 + 1  \]
	Indeed, we start with a tick ($+1$), then a call to decompose ($1 + 2^i$), then we get back the results ($+1$) and call  two mergesort ($f(i-1)$) in parallel, then we get back the results ($+2$) and call merge $((3 + K_c)2^i + 1)$, and finally we get back the result $(+1)$ and send it to the output channel.  
	So, finally, we obtain 
	\[ f(i) \ge 1 \qquad \forall i \ge 1, f(i) \ge  7 + f(i-1) + (4 + K_c)2^i \]
	, thus 
	\[ f(i) = 1 + 7i + (4 + K_c)(2^{i+1} - 2) \]
	Remark that if we only look at the coefficient for $K_c$, we get back the complexity for the number of comparison. Anyway, mergesort is indeed linear under maximal parallelism with this result.

	\section{Work of a Process}
	
	We now want to obtain the total complexity of a process, that is to say the total number of tick without parallelism. We will see that this notion of complexity is far easier to obtain. First, let us define the new time reduction we are interested in $\tickred$. This is defined in Figure~\ref{f:tickreduction}. 
	
	\begin{figure}
		\centering
		\begin{framed}
			\small 
			\begin{center} 
				\AxiomC{$P \tickred P'$}
				\UnaryInfC{$P \parr Q \tickred P' \parr Q$}
				\DisplayProof 
				\qquad 
				\AxiomC{$Q \tickred Q'$}
				\UnaryInfC{$P \parr Q \tickred P \parr Q$}
				\DisplayProof 
				\qquad
				\AxiomC{$P \tickred P'$}
				\UnaryInfC{$\pnu P \tickred \pnu P'$}
				\DisplayProof 
				\qquad 
				\AxiomC{}
				\UnaryInfC{$\tick. P \tickred P$}
				\DisplayProof 
			\end{center}   
		\end{framed}
		\caption{Tick Reduction Rules}
		\label{f:tickreduction}
	\end{figure}
	
	And then, from any process $P$, a reduction to $Q$ is just a sequence of one-step reductions with $\red$ or $\tickred$, and the complexity of this reduction is the number of $\tickred$. Contrary to the tick-last strategy, we do not add any restrictions on this semantic. We will now again design a type system to obtain a bound on the complexity of all possible reductions from $P$. We will see that this type system is more permissive than the previous one, and is a simplification of the previous one. 
	
	\begin{definition}
		The set of types and \emph{base types} are given by the following grammar.
		\begin{align*}
		\IB &:= \Inat \midd \Ilis \midd \Ibool \\
		T &:= \IB \midd \Tch \midd \Tich \midd \Toch \midd \Tserv \midd \Tiserv \midd \Toserv \\	
		\end{align*}
		\label{d:types}
	\end{definition}
	Note that there are no time indication in those types. Then, the subtyping system is given in Figure~\ref{f:worksubtype}. It is very close to the previous one. 
	
	\begin{figure}
		\centering
		\begin{framed}
			\small 
			\begin{center}
				\AxiomC{$\phi;\Phi \vDash I' \le I$}
				\AxiomC{$\phi;\Phi \vDash J \le J'$}
				\BinaryInfC{$\phi;\Phi \vdash \Inat \subtype \Inat[I'][J']$}
				\DisplayProof 
				\qquad 
				\AxiomC{$\phi;\Phi \vDash I' \le I$}
				\AxiomC{$\phi;\Phi \vDash J \le J'$}
				\AxiomC{$\phi;\Phi \vdash \IB \subtype \IB' $}
				\TrinaryInfC{$\phi;\Phi \vdash \Ilis \subtype \Ilis[I'][J'][\IB']$}
				\DisplayProof 
				\\
				\vvskip 
				\AxiomC{}
				\UnaryInfC{$\phi; \Phi \vdash \Ibool \subtype \Ibool$}
				\DisplayProof
				\qquad  
				\AxiomC{$\phi; \Phi \vdash \vect{T} \subtype \vect{U}$}
				\AxiomC{$\phi; \Phi \vdash \vect{U} \subtype \vect{T}$}
				\BinaryInfC{$\phi; \Phi \vdash \Tch \subtype \Tch[\vect{U}]$}
				\DisplayProof
				\qquad 
				\AxiomC{}
				\UnaryInfC{$\phi; \Phi \vdash \Tch \subtype \Tich$}
				\DisplayProof 
				\\
				\vvskip 
				\AxiomC{}
				\UnaryInfC{$\phi;\Phi \vdash \Tch \subtype \Toch$}
				\DisplayProof 
				\qquad 
				\AxiomC{$\phi;\Phi \vdash \vect{T} \subtype \vect{U}$}
				\UnaryInfC{$\phi;\Phi \vdash \Tich \subtype \Tich[\vect{U}]$}
				\DisplayProof 
				\qquad 		
				\AxiomC{$\phi; \Phi \vdash \vect{U} \subtype \vect{T}$}
				\UnaryInfC{$\phi;\Phi \vdash \Toch \subtype \Toch[\vect{U}]$}
				\DisplayProof 
				\\ 
				\vvskip 
				\AxiomC{$(\phi,\vect{i}); \Phi \vdash \vect{T} \subtype \vect{U}$}
				\AxiomC{$(\phi,\vect{i}); \Phi \vdash \vect{U} \subtype \vect{T}$}
				\AxiomC{$(\phi,\vect{i}); \Phi \vDash K = K'$}
				\TrinaryInfC{$\phi; \Phi \vdash \Tserv \subtype \Tserv[@][K'][\vect{U}]$}
				\DisplayProof
				\\ 
				\vvskip 
				\AxiomC{}
				\UnaryInfC{$\phi; \Phi \vdash \Tserv \subtype \Tiserv$}
				\DisplayProof
				\qquad  
				\AxiomC{}
				\UnaryInfC{$\phi;\Phi \vdash \Tserv \subtype \Toserv$}
				\DisplayProof 
				\\
				\vvskip 			
				\AxiomC{$(\phi,\vect{i});\Phi \vdash \vect{T} \subtype \vect{U}$}
				\AxiomC{$(\phi,\vect{i});\Phi \vDash K' \le K$}
				\BinaryInfC{$\phi;\Phi \vdash \Tiserv \subtype \Tiserv[@][K'][\vect{U}]$}
				\DisplayProof 
				\\
				\vvskip 
				\AxiomC{$(\phi,\vect{i});\Phi \vdash \vect{U} \subtype \vect{T}$}
				\AxiomC{$(\phi,\vect{i});\Phi \vDash K \le K'$}
				\BinaryInfC{$\phi;\Phi \vdash \Toserv \subtype \Toserv[@][K'][\vect{U}]$}
				\DisplayProof 
				\qquad 
				\AxiomC{$\phi;\Phi \vdash T \subtype T'$}
				\AxiomC{$\phi;\Phi \vdash T' \subtype T''$}
				\BinaryInfC{$\phi;\Phi \vdash T \subtype T''$}
				\DisplayProof 
				
			\end{center}   
		\end{framed}
		\caption{Subtyping Rules for Sized Types}
		\label{f:worksubtype}
	\end{figure}
	
	And then, the typing for expressions is the same as before, and for processes we take the rules of Figure~\ref{f:worktypeprocess}.
	
	\begin{figure}
		\centering
		\begin{framed}
			\footnotesize 
			\begin{center}
				\AxiomC{}
				\UnaryInfC{$\Itype{\pzero}{0}$}
				\DisplayProof 
				\qquad
				\AxiomC{$\Itype{P}{K}$}
				\AxiomC{$\Itype{Q}{K'}$}
				\BinaryInfC{$\Itype{P \parr Q}{K + K'}$}
				\DisplayProof 
				\\ 
				\vvskip 
				\AxiomC{$\Ietype{a}{\Tiserv}$}
				\AxiomC{$\Itype[(\phi,\vect{i})][@][\Gamma, \vect{\var} : \vect{T}]{P}{K}$}
				\BinaryInfC{$\Itype{\pserv}{0}$}
				\DisplayProof 
				\\ 
				\vvskip  			
				\AxiomC{$\Ietype{a}{\Tich}$}
				\AxiomC{$\Itype[@][@][\Gamma, \vect{\var} : \vect{T}]{P}{K}$}
				\BinaryInfC{$\Itype{\pin}{K}$}
				\DisplayProof 
				\qquad
				\AxiomC{$\Ietype{a}{\Toch}$}
				\AxiomC{$\Ietype{\vect{e}}{\vect{T}}$}
				\BinaryInfC{$\Itype{\pout}{0}$}
				\DisplayProof 
				\\ 
				\vvskip 
				\AxiomC{$\Ietype{a}{\Toserv}$}
				\AxiomC{$\Ietype{\vect{e}}{\Isub[\vect{T}][\vect{i}][\vect{J}]}$}
				\BinaryInfC{$\Itype{\pout}{\Isub[K][\vect{i}][\vect{J}]}$}
				\DisplayProof 
				\qquad
				\AxiomC{$\Itype[@][@][\Gamma, a : T]{P}{K}$}
				\UnaryInfC{$\Itype{\pnu P}{K}$}
				\DisplayProof 
				\\ 
				\vvskip 
				\AxiomC{$\Ietype{e}{\Inat}$}
				\AxiomC{$\Itype[@][(\Phi,I \le 0)][@]{P}{K}$}
				\AxiomC{$\Itype[@][(\Phi,J \ge 1)][{\Gamma, x : \Inat[I-1][J-1]}]{Q}{K}$}
				\TrinaryInfC{$\Itype{\pifn}{K}$}
				\DisplayProof 
				\\
				\vvskip   
				\AxiomC{$\Ietype{e}{\Ilis}$}
				\AxiomC{$\Itype[@][(\Phi,I \le 0)][@]{P}{K}$}
				\AxiomC{$\Itype[@][(\Phi,J \ge 1)][{\Gamma, x : \IB, y : \Ilis[I-1][J-1]}]{Q}{K}$}
				\TrinaryInfC{$\Itype{\pifl}{K}$}
				\DisplayProof 
				\\
				\vvskip 
				\AxiomC{$\Ietype{e}{\Ibool}$}
				\AxiomC{$\Itype{P}{K}$}
				\AxiomC{$\Itype{Q}{K}$}
				\TrinaryInfC{$\Itype{\pif}{K}$}
				\DisplayProof 
				\qquad 
				\AxiomC{$\Itype{P}{K}$}
				\UnaryInfC{$\Itype{\tick.P}{K+1}$}
				\DisplayProof 
				\\
				\vvskip 
				\AxiomC{$\Itype[@][@][\Delta]{P}{K}$}
				\AxiomC{$\phi;\Phi \vdash \Gamma \subtype \Delta$}
				\AxiomC{$\phi;\Phi \vDash K \le K'$}
				\TrinaryInfC{$\Itype{P}{K'}$}
				\DisplayProof
			\end{center}   
		\end{framed}
		\caption{Typing Rules for Processes}
		\label{f:worktypeprocess}
	\end{figure}
	
	With this type system, we obtain as before some lemmas such as weakening (Lemma~\ref{l:weakening}), strengthening (Lemma~\ref{l:strengthening}), index substitution (Lemma~\ref{l:indexsubstitution}) and finally substitution (Lemma~\ref{l:substitution}). With those, we can show with a simpler proof than before the non-quantitative subject reduction (Theorem~\ref{t:nonquantsubjectreduction}). Then, we can show the following theorem:
	
	\begin{theorem}[Work Complexity]
		If $P \tickred Q$ and $\Itype{P}{K}$ then we have $\Itype{Q}{K'}$ with $\phi;\Phi \vDash K' + 1 \le K$. 
	\end{theorem}
	
	\begin{proof}
		By induction on $P \tickred Q$. All the cases are direct, since the rule for parallel composition is the sum of complexity and the rule for $\nu$ does not change the complexity. Finally, the rule for tick gives directly this propriety.
	\end{proof}
	
	So, as a consequence we obtain quasi immediately that $K$ is indeed a bound on the complexity of $P$ if we have $\Itype{P}{K}$. As we can see, this complexity is far more easier to obtain than the span as the parallelism is not really taken in account. That is why we think the span is a good notion of complexity if we want to focus on parallelism. 
	
	\section{Other Results} 
	
	\subsection{Work for Mergesort in the Number of Comparisons}
	
	Without detailing the typing derivation, we give the work for mergesort. The detailed derivation looks like the one for the span, and so we only give the equations that the complexity must satisfy. 
	
	For merge, if we call $f(i,j)$ the complexity of a call to merge on a list of size smaller than $i$ and a list of size smaller than $j$, we have:
	\[ \forall i,j \ge 1, f(i,j) \ge 1 + f(i-1,j) \text{ and } f(i,j) \ge 1 + f(i,j-1) \]
	So we can take the complexity $f(i,j) = i + j$. 
	
	Then, the decompose has a complexity $0$ since it does not involve any comparison. Finally, if we denote $f(i)$ the complexity of a call to mergesort on a list of size smaller than $2^i$, we have:
	\[ \forall i \ge 1, f(i) \ge f(i-1) + f(i-1) + 2^i \] 
	So, we can take the complexity $f(i) = i * 2^i$, and we obtain as expected a complexity in $n \mathtt{log}(n)$. 
	
	\subsection{Work for Mergesort for the Number of Communication}
	Then, as before, we could also consider a tick after each input, in order to take in account the communication complexity. As before, let $K_c$ be the complexity of a call to compare, we then obtain the following complexities:
	
	For merge, we have:
	\[ f(i,j) \ge 1 \qquad \forall i,j \ge 1, f(i,j) \ge 1 + K_c + 1 + f(i-1,j) + 1 \text{ and } f(i,j) \ge 1 + K_c + 1 + f(i,j-1) + 1 \]
	So, we obtain $f(i,j) = (3 + K_c)(i + j) + 1$.
	
	For decompose, we have, on a list of size $2i$:
	\[ f(i) \ge 1 \qquad \forall i \ge 1, f(i) \ge 1 + f(i-1) + 2 \]
	So, we obtain $f(i) = 3i + 1$.
	
	Finally, for merge sort we obtain, on a list of size $2^i$:
	\[ f(i) \ge 1 + 3*2^{i-1} + 1 + 1 + f(i-1) + 1 + f(i-1) + 2 + (3 + K_c)2^i + 1 + 1 \]
	\[ f(i) \ge 8 + (4.5 + K_c)2^i + 2f(i-1) \]
	So we can take $f(i) = 8(2^{i+1}-1) + (4.5 + K_c)i2^i$. 
	
	\subsection{Another Way to Merge}
	
	\begin{figure}
		\centering
		\begin{framed}
			\begin{lstlisting}[style=ocaml,emph={l,a,x,y,q,r,b,z,c,a',b',c',d},mathescape=true]
			!merge($l_0$,$l_1$,a). match $l_0$ with 
			| [] -> $\overline{a}$<$l_1$>
			| x::q -> match $l_1$ with 
				| [] -> $\overline{a}$<$l_0$>
				| y::r -> (nub)(nu$c_1$)(nu$c_2$)($\mathtt{tick}.\overline{compare}$<x,y,b> | $\overline{merge}$<q,$l_1$,$c_1$> | $\overline{merge}$<$l_0$,r,$c_2$> | 
				b(z). if z then $c_1$($l_2$).$\overline{a}$<x::$l_2$> else $c_2$($l_2$).$\overline{a}$<y::$l_2$>
				)
			
			\end{lstlisting}
		\end{framed}
		\caption{Alternative Merge}
		\label{f:alternativemerge}
	\end{figure}
	
	An alternative version of merge is given in Figure~\ref{f:alternativemerge}. The idea of this alternative version compared to Figure~\ref{f:fusionsort} is to compute both results of the conditional before even receiving the results of the conditional (which can take a long time if the comparison is costly). With this version, we obtain the following typing for the parallel complexity:
	 
	Let us pose the following context:
	\[ \Gamma := merge : \Iserv[0][(i,j)][1][{\Ilis[0][i], \Ilis[0][j],\Ioch[1][{\Ilis[0][i+j]}]}] \]
	
	And we pose $\Gamma_o$ the same context with output server instead of input/output.
	We also pose:
	\[ \Gamma' := \Gamma_o, l_0 : \Ilis[0][i], l_1 : \Ilis[0][j], a : \Ioch[1][{\Ilis[0][i+j]}] \]
	\[ \Gamma'' := \Gamma', x : \IB, q : \Ilis[0][i-1], y: \IB, r : \Ilis[0][j-1] \]
	\[ \Gamma''' := \Gamma'', b: \Ich[1][\Ibool], c_1: \Ich[1][{\Ilis[0][i+j-1]}], c_2: \Ich[1][{\Ilis[0][i+j-1]}] \]
	\begin{prooftree}
		\footnotesize
		\AxiomC{}
		\UnaryInfC{$\Itype[(i,j)][0 \le 0][\Gamma']{\pout[a][l_1]}{1}$}
		\AxiomC{}
		\UnaryInfC{$\Itype[(i,j)][(i \ge 1, j \ge 1)][{\Idecr[\Gamma'''][1]}]{\pout[compare][x,y,b]}{0}$}
		\UnaryInfC{$\Itype[(i,j)][(i \ge 1, j\ge 1)][\Gamma''']{\tick. \pout[compare][x,y,b]}{1}$}
		\AxiomC{$\pi$}
		\AxiomC{$\pi'$}
		\TrinaryInfC{$\Itype[(i,j)][(i \ge 1, j \ge 1)][\Gamma''']{\tick. \pout[compare][x,y,b] \parr \cdots}{1}$}
		\UnaryInfC{$\Itype[(i,j)][(i \ge 1, j \ge 1)][\Gamma'']{\pnu[b,c_1,c_2] (\tick . \pout[compare][x,y,b] \parr \cdots)}{1}$}
		\doubleLine
		\UnaryInfC{$\Itype[(i,j)][(i \ge 1)][{\Gamma', x : \IB, q : \Ilis[0][i-1]}]{\pifl[l_1][{\pout[a][l_0]}][y][r][\cdots]}{1}$}
		\BinaryInfC{$\Itype[(i,j)][\cdot][\Gamma']{\pifl[l_0][{\pout[a][l_1]}][x][q][\cdots]}{1}$}
		\UnaryInfC{$\Itype[\cdot][\cdot][\Gamma]{\pserv[merge][l_0,l_1,a][\cdots]}{0}$}
	\end{prooftree}
	
	where $\pi$ is the proof that the calls to merge have complexity $1$ (direct by the complexity given in $\Gamma$)
	and $\pi'$ is the following proof:
	
	\begin{prooftree}
		\footnotesize
		\AxiomC{}
		\UnaryInfC{$\Itype[(i,j)][(i \ge 1, j \ge 1)][{\Idecr[\Gamma'''][1]}, l_2 : {\Ilis[0][i+j-1]}]{\pout[a][x::l_2]}{0}$} 
		\UnaryInfC{$\Itype[(i,j)][(i \ge 1, j \ge 1)][{\Idecr[\Gamma'''][1]}]{\pin[c_2][l_2][{\pout[a][x::l_2]}]}{0}$}
		\AxiomC{$\cdots$} 
		\BinaryInfC{$\Itype[(i,j)][(i \ge 1, j \ge 1)][{\Idecr[\Gamma'''][1]}]{\pif[z][{\pin[c_1][l_2][{\pout[a][x::l_2]}]}][{\cdots}]}{0}$}
		\UnaryInfC{$\Itype[(i,j)][(i \ge 1, j \ge 1)][\Gamma''']{\pin[b][z][\cdots]}{1}$}
	\end{prooftree}
	And the typing for the other branch of the conditional is similar. So, in the end we got a complexity of one. This is because if you consider the reduction of this term, a lot of comparison (exponential in the size of the input list) are done in the same time in parallel. Another possibility would be to do exhaustively all the comparisons between the two list in parallel without repetition and then merging the list, again with a parallel complexity of one if done correctly. This way of doing things makes sense if the cost of a comparison is really huge and we have a lot of processors. However, if we consider the work complexity, we obtain an exponential complexity, and in the case of doing all comparison we would obtain a square complexity. In practice, depending on the size of the input or the number of processors, one or the other version is better. The thing is that it is important to take in consideration those two notions of complexity.  
	 
	\bibliography{Bibliography}
	\bibliographystyle{plain} 
	\end{document}